\documentclass[journal, draftclsnofoot, onecolumn, 10pt]{IEEEtran}
\usepackage{graphicx}
\graphicspath{file:///home/stack/AVWC}
\usepackage{amssymb}

\usepackage[T1]{fontenc}

\usepackage{cite}

\ifCLASSINFOpdf
\else
\fi
\usepackage{amsmath}
\newtheorem{thm}{Theorem}
\newtheorem{cor}[thm]{Corollary}
\newtheorem{lem}[thm]{Lemma}
\newtheorem{prop}[thm]{Proposition}
\newtheorem{defn}{Definition}
\newtheorem{rem}{Remark}
\newtheorem{exmp}{Example}

\usepackage{amsfonts}
\usepackage{graphicx}
\begin{document}

%
\title{Arbitrarily Varying Wiretap Channel with State Sequence Known or Unknown at the Receiver}
%
%
%

\author{Dan~He
        and Yuan~Luo
\thanks{Dan He is with the State Key Laboratory of Integrated Networks Services, Xidian University, Xi'an, China. Email: dhe@stu.xidian.edu.cn.}
\thanks{Yuan Luo is with the Department of Computer Science and Engineering, Shanghai Jiao Tong University, Shanghai, China. Email: yuanluo@sjtu.edu.cn.}
}

\IEEEpeerreviewmaketitle

\maketitle

\begin{abstract}
The secrecy capacity problems over the general arbitrarily varying wiretap channel, with respect to the maximal decoding error probability and strong secrecy criterion, are considered, where the channel state sequence may be known or unknown at the receiver. In the mean time, it is always assumed that the channel state sequence is known at the eavesdropper and unknown at the transmitter. Capacity results of both stochastic code (with random encoder and deterministic decoder) and random code (with random encoder and decoder) are discussed. This model includes the previous models of classic AVWC as special cases. Single-letter lower bounds on the secrecy capacities are given, which are proved to be the secrecy capacities when the main channel is less noisy than the wiretap channel. The coding scheme is based on Csisz\'{a}r's almost independent coloring scheme and Ahlswede's elimination technique. Moreover, a new kind of typical sequence with respect to states is defined for this coding scheme. It is concluded that the secrecy capacity of stochastic code is identical to that of random code when the receiver knows the state sequence. Meanwhile, random code may achieve larger secrecy capacity when the state sequence is unknown by the receiver.
\end{abstract}


\section{Introduction}\label{sec:introduction}
Arbitrarily varying channel (AVC) is the most general and difficult channel model in two-terminal discrete memoryless systems (DMSs). Unlike other well-known DMSs, the capacity of an AVC may be affected by the criteria of decoding error probabilities (average or maximal) and the classes of coding schemes (deterministic, stochastic or random, see Subsection \ref{sec:avc}). To be precise, when studying AVC, the following six kinds of capacities are considered \cite{Ahlswede-1978}:
\begin{enumerate}
  \item capacity of deterministic code with respect to the average decoding error probability,
  \item capacity of deterministic code with respect to the maximal decoding error probability,
  \item capacity of stochastic code with respect to the average decoding error probability,
  \item capacity of stochastic code with respect to the maximal decoding error probability,
  \item capacity of random code with respect to the average decoding error probability,
  \item capacity of random code with respect to the maximal decoding error probability.
\end{enumerate}
The idea of applying random code to the AVC originated from Blackwell et al \cite{BBT-1960}, see also Lemma 2.6.10 in \cite{CT_DMS}. A random code is a pair of random encoder and decoder \((F,\Phi)\), distributed over a collection of deterministic codes, where \(F\) and \(\Phi\) are correlated. Before each transmission, there exists a third party telling the transmitter and receiver the exact pair of deterministic encoder and decoder to be used. That kind of information is called the common randomness (CR) of the random code.
CR is critical to random code, but it costs extra bandwidth. To solve that problem, Ahlswede \cite{Ahlswede-1978} developed an elimination technique to decrease the rate of CR to arbitrarily small, and determined all the other five capacities of the AVC except that of deterministic code with the maximal decoding error, which is still open. Some results can be found in \cite{Csiszar-1988}.

AVC can be treated as the model of channels with states. However, unlike the models introduced in \cite{Shannon-1958,Gelfand-1980}, the probability mass function of the channel states is unknown in the case of AVC. The problem of AVC with non-causal state sequence known at the encoder was studied by Ahlswede \cite{Ahlswede-1986}, which was actually an extension of \cite{Gelfand-1980}. An extension of Ahlswede's work can be found in \cite{Winshtok-2006}. Moreover, Csisz\'{a}r and Narayan \cite{Csiszar-1988-2} considered the AVC whose average costs of input and state sequence are constrained, and provided the capacity of random code.

The importance of arbitrarily varying wiretap channel (AVWC), first introduced in \cite{bloch-2008}, is self-evident. When the AVWC has constrained state sequence, it includes discrete memoryless wiretap channels \cite{wiretap,BCC}, wiretap channel of type II \cite{wiretap2,wiretap2_1994,luo-2006,wiretap2_2,dan-ISIT,dan-entropy,dan-entropy2,wiretap2_1} and compound wiretap channels \cite{liang-2008,Bjelakovic-2013a,schaefer-2014,schaefer-2015} as special cases. Moreover, the coding scheme in AVWC can also be applied to wiretap channels with stochastic channel states \cite{Mitrpant-2004,dai-2012,dai-2017}. When considering an AVWC, we can assume that the eavesdropper knows exactly the channel state sequence. Moreover, we often assume that the eavesdropper is able to control the state sequence during the communication \cite{He-2011,Bjelakovic-2013}.

Bjelakovi\'{c} et al \cite{Bjelakovic-2013} considered a special class of AVWC, where the wiretap channel could be regarded as an equivalent discrete memoryless channel (DMC). The lower bounds on the capacities of stochastic code and random code with respect to the strong secrecy criterion \cite{Maurer-1994,AI_SC} were given there. The key idea was migrating the coding scheme for compound wiretap channels \cite{Bjelakovic-2013a} to AVWC with the help of the robustification technique introduced by Ahlswede \cite{Ahlswede-1986}. If a random code is applied to the AVWC, the communication is manipulated with the help CR. If the CR is known by the eavesdropper, he/she may control the state sequence according to the CR. This situation was considered in \cite{Boch-2013}. On account of the elimination technique introduced by Ahlswede \cite{Ahlswede-1978}, it suffices to let the rate of CR be arbitrarily small when the eavesdropper is absent. However, when the eavesdropper exists and the rate of CR is positive, the CR can serve as a secret key to increase the secrecy capacity, as studied in \cite{Boch-2016}.

Goldfand et al \cite{Goldfeld-1} considered another special case of AVWC, where the channel state sequence was constrained in a certain type. The coding scheme there was based on a stronger version of Wyner's soft covering lemma \cite{wyner-1975b}, which was first used to deal with the capacity problem of extended wiretap channel II in \cite{wiretap2_2}. The model discussed in \cite{Goldfeld-1} discarded the assumption that the wiretap channel was a DMC. Meanwhile, the capacity was on semantic secrecy criterion. However, that coding scheme  can not be applied to the general AVWC without constraints on state sequence.

Some other results on AVWC include multi-letter description of the secrecy capacity \cite{wiese-2015,wiese-2016} and the continuity on the secrecy capacity \cite{Boch-2016,boch-2015}.

This paper considers the strong secrecy capacity problems of general AVWC without any limits, whose channel state sequence is supposed to be known at the eavesdropper and unknown at the transmitter. To be concrete, the following three communication models are considered.

\begin{itemize}

\item \emph{Arbitrarily varying wiretap channel with channel state sequence unknown at the receiver (AVWC).} Lower bounds on secrecy capacities of stochastic and random codes over this model with respect to the maximal decoding error probability and the strong secrecy criterion are given. The average decoding error is not considered in the context of wiretap channel since relaxing the criterion of decoding error does not yield a larger lower bound. This model includes the models in \cite{Bjelakovic-2013,Boch-2013} as special cases. The secrecy capacities are determined when the main channel is severely less noisy (defined in Subsection \ref{sec:less}) than the wiretap channel. Moreover, the coding scheme introduced in this paper can also be readily used to the model of AVWC with constrained type of the state sequence \cite{Goldfeld-1}.
  \item \emph{Arbitrarily varying channel with channel state sequence known at the receiver (AVC-CSR).} This communication model was first studied by \cite{stambler-1975}. We present a new proof here as a preliminary to the proofs of the third model. The five kinds of capacities, except that of deterministic code with respect to the maximal decoding error, are determined. Unlike the AVC, it is proved that those five capacities of an AVC-CSR are identical.
  \item \emph{Arbitrarily varying wiretap channel with channel state sequence known at the receiver (AVWC-CSR).} Lower bounds on secrecy capacities of stochastic and random codes over this model with respect to the maximal decoding error probability and the strong secrecy criterion are given. Unlike the general AVWC, the secrecy capacity of stochastic code over an AVWC-CSR is identical to that of random code. Moreover, the secrecy capacity is determined when the main channel is strongly less noisy (defined in Subsection \ref{sec:less}) than the wiretap channel.
\end{itemize}

This paper mainly makes the following two contributions. 1. Introduce a tool of typical sequences with respect to the channel states. This tool can be used to bound the values of both the decoding error probability at the legitimate receiver and the exposed source information to the eavesdropper. 2. Extend the almost independent coloring scheme developed by Csisz\'{a}r. This scheme could construct a partition on a given ``good'' codebook to satisfy the strong secrecy criterion. This scheme has recently been used to deal with the problems of extended wiretap channel II \cite{dan-ISIT,dan-entropy}.

The remainder of this paper is organized as follows. Section \ref{sec:pre} presents the preliminaries for this paper, including the new definitions of typical sequences and some known results on AVC. Section \ref{sec:main} summarizes main results of this paper, including lower bounds on secrecy capacities of AVWC and AVWC-CSR, along with capacities of AVC-CSR. In Section \ref{sec:basic}, we introduce some lemmas derived from Csisz\'{a}r's almost independent coloring scheme, which are basic tools for the analysis of security. The proofs of main results are given in Section \ref{section:proofs}. In section \ref{section:example}, we determine the secrecy capacities of less noisy AVWC and AVWC-CSR, and consider the models of AVWC with constrained state sequence. Finally, Section \ref{section:conclusion} concludes this paper.

\section{Preliminaries}\label{sec:pre}

Throughout this paper, random variables, sample values and alphabets (sets) are denoted by capital letters, lower case letters and calligraphic letters, respectively. A similar convention is applied to random vectors and their sample values. The probability mass function of a given random variable \(X\) is denoted by \(P_X\). Moreover, \(P_{XY}\) and \(P_{Y|X}\) denote the joint and conditional probability mass functions of the random variable pair \((X, Y)\), respectively. For a given AVC or AVC-CSR, \(e\) is used to represent the maximal decoding error probability of a deterministic code, \(\bar{e}\) represents the average decoding error probability, and \(e_m\) represents the decoding error probability of message \(m\). When the code is random or stochastic, \(e, \bar{e}\) and \(e_m\) are random functions of it. In that case, we would use \(\lambda\), \(\bar{\lambda}\) and \(\lambda_m\) to represent the expectations of \(e, \bar{e}\) and \(e_m\), respectively. The capacities with respect to the maximal and average decoding error probabilities are denoted by \(C\) and \(\bar{C}\), respectively. Some important notations are summarized in Table \ref{table1} for reference.

\begin{table}
\centering
\caption{Important Notations Used in This Paper}\label{table1}
\begin{tabular}{|p{4cm}|p{10cm}|}
\hline
NOTATION & DESCRIPTION \\
\hline
\(X\) & Common input of both the main AVC and the wiretap AVC (one time of transmission)\\
\hline
\(Y_s,Z_s\) & Channel outputs of the main AVC and the wiretap AVC when the channel state is \(s\) (one time of transmission)\\
\hline
\(W_s,V_s\) & Transition probability matrices of the main AVC and the wiretap AVC when the channel state is \(s\)\\
\hline
\(\mathcal{S}\) & The finite state set of both the main AVC and the wiretap AVC \\
\hline
\(Y_{\mathcal{S}},Z_{\mathcal{S}}\) & \(Y_{\mathcal{S}}=(Y_s:s\in\mathcal{S})\) and \(Z_{\mathcal{S}}=(Z_s:s\in\mathcal{S})\)\\
\hline
\(\mathcal{W},\mathcal{V}\) & \(\mathcal{W}=\{W_s:s\in\mathcal{S}\}\) and \(\mathcal{V}=\{V_s:s\in\mathcal{S}\}\). \(\mathcal{W}\) is used to specify the main AVC and \(\mathcal{V}\) is used to specify the wiretap AVC. \\
\hline
\(q\) & A specific probability mass function over \(\mathcal{S}\) \\
\hline
\(W_q\) & \(W_q=\sum_{s\in\mathcal{S}}q(s)\cdot W_s\) is a convex combination of the transition matrices from \(\mathcal{W}\)\\
\hline
\(Y_q\) & It follows that \(\Pr\{Y_q=y|X=x\} = W_q(y|x)\) \\
\hline
\(\mathcal{P}(\mathcal{S})\) & The collection of all probability mass functions over \(\mathcal{S}\) \\
\hline
\(\mathcal{M}\) & The message set \\
\hline
\(M\) & The source message uniformly distributed over \(\mathcal{M}\) \\
\hline
\(X^N\) & Common input of both the main AVC and the wiretap AVC (\(N\) times of transmission)\\
\hline
\(Y^N(s^N),Z^N(s^N)\) & Channel outputs of the main AVC and the wiretap AVC when the channel state sequence is \(s^N\) (\(N\) times of transmission)\\
\hline
\((f,\phi)\) & A pair of deterministic encoder and decoder \\
\hline
\((F,\phi)\) & A pair of stochastic encoder and decoder \\
\hline
\((F,\Phi)\) & A pair of random encoder and decoder \\
\hline
\(\lambda(\mathcal{W},f,\phi)\),\(\lambda(\mathcal{W},F,\phi)\), \(\lambda(\mathcal{W},F,\Phi)\) & Maximal decoding error probabilities of different classes of codes over AVC \(\mathcal{W}\) \\
\hline
\(\lambda^{CSR}(\mathcal{W},f,\phi)\),\(\lambda^{CSR}(\mathcal{W},F,\phi)\),  \(\lambda^{CSR}(\mathcal{W},F,\Phi)\) & Maximal decoding error probabilities of different classes of codes over AVC-CSR \(\mathcal{W}\) \\
\hline
\(U\) & An auxiliary random variable to formulate the secrecy capacity results. \(U\rightarrow X \rightarrow (Y_{\mathcal{S}},Z_{\mathcal{S}})\) forms a Markov chain.\\
\hline
\(\mathcal{C}\) & A deterministic codebook \\
\hline
\(\mathbf{C}\) & A random codebook \\
\hline
\(X^N(\mathcal{C})\) & A random sequence uniformly distributed over the deterministic codebook \(\mathcal{C}\) \\
\hline
\(Z^N(\mathcal{C},s^N)\) & Output of the wiretap AVC when the channel input is \(X^N(\mathcal{C})\) and state sequence is \(s^N\) \\
\hline
\end{tabular}
\end{table}

This section provides basic tools for the subsequent sections, including information measurements with respect to channel states, basic results of AVC, and typicality under the state sequence.

\subsection{Information measurements with respect to channel states}

This subsection defines some functions on information measurements with respect to the channel states, which would be used to characterize properties of typical sequences in Subsection \ref{sec:typical}.

 Let \(X\) be a random variable of probability mass function \(P_X\), and \(Y_{\mathcal{S}} = (Y_s: s\in \mathcal{S})\) be a collection of random variables satisfying that
\begin{equation}\label{eq:measure:1}
\operatorname{Pr}\{X=x, Y_s=y\} = P_{XY_s}(x, y) = P_X(x) W_s(y|x)
\end{equation}
for \(s\in \mathcal{S}\), \(x\in \mathcal{X}\) and \(y\in \mathcal{Y}\), where \(\mathcal{S}\) is a finite set and \(\{W_s: s \in \mathcal{S}\}\) is a family of transition matrices. For any probability mass function \(P\) on \(\mathcal{S}\), denote
\begin{equation}\label{eq:pre:14}
\bar{H}_P(Y_{\mathcal{S}}) = \sum_{s\in\mathcal{S}} P(s) H(Y_s),
\bar{H}_P(Y_{\mathcal{S}}|X) = \sum_{s\in\mathcal{S}} P(s) H(Y_s|X)
\text{ and }
\displaystyle \bar{I}_P(X; Y_{\mathcal{S}}) = \sum_{s\in\mathcal{S}} P(s) I(X; Y_s).
\end{equation}
It follows clearly that
\[
\bar{I}_P(X; Y_{\mathcal{S}}) = \bar{H}_P(Y_{\mathcal{S}}) - \bar{H}_P(Y_{\mathcal{S}}|X).
\]
Moreover, we also have
\begin{equation}\label{eq:pre:12}
\min_{s\in\mathcal{S}} I(X; Y_s) \leq \bar{I}_P(X; Y_{\mathcal{S}}) \leq \max_{s\in\mathcal{S}} I(X; Y_s).
\end{equation}

We give an example below, so that the reader can have a clearer idea on those notations.

\begin{exmp}
Let \(\mathcal{X}=\mathcal{Y}=\mathcal{S}=\{0,1\}\), and \(\mathcal{W} = \{W_0, W_1\}\) with
\[
W_0 =
\Bigg[
\begin{matrix}
1-q_0 & q_0\\
q_0 & 1-q_0
\end{matrix}
\Bigg]
\text{ and }
W_1 = \Bigg[
\begin{matrix}
1-q_1 & q_1\\
q_1 & 1-q_1
\end{matrix}
\Bigg]
\]
for some \(0 \leq q_0, q_1 \leq 1\). Suppose that the random variables \(X\) and \(Y_{\mathcal{S}} = \{Y_0, Y_1\}\) satisfy
\[
\Pr\{X=x, Y_s=y\}=\frac{1}{2}W_s(y|x)
\]
for \(x,y,s\in\{0, 1\}\). Then we have
\[
H(Y_s|X)=h(q_s) \text{ and } I(X; Y_s) = 1- h(q_s)
\]
for \(s=0,1\). Furthermore, denote by \(P\) a probability mass function on \(\mathcal{S}\) such that \(P(0)=1-P(1)=p\) for some real number \(0\leq p \leq 1\). Then it follows that
\[
\bar{H}_P(Y_{\mathcal{S}}|X) = pH(Y_0|X)+(1-p)H(Y_1|X)=ph(q_0)+(1-p)h(q_1)
\]
and
\[
\bar{I}_P(X;Y_{\mathcal{S}})=pI(X;Y_0)+(1-p)I(X;Y_1)=1-ph(q_0)-(1-p)h(q_1).
\]
\end{exmp}

\subsection{Some known results of AVC}\label{sec:avc}

This subsection gives the definition of AVC, and lists some capacity results that have been known. In general, an AVC is specified by a family of transition matrices, whose size is allowed to be infinite, according to \cite{CT_DMS,Ahlswede-1978}. However, \textbf{we constrain the size of this family to be finite so as to focus our attention on the security problems of AVWC and AVWC-CSR.}

\begin{figure}
     \centering
     \includegraphics[width=12cm]{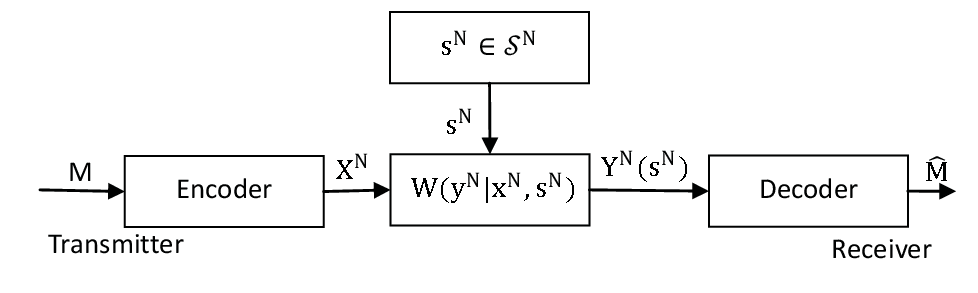}
     \caption{AVC, Arbitrarily varying channel.}\label{Fig:avc}
 \end{figure}

\textbf{The definition}

\begin{defn}\label{def:avc}
\emph{(AVC, Arbitrarily Varying Channel)} An AVC, depicted in Fig. \ref{Fig:avc}, is specified by a finite collection of transition matrices \(\mathcal{W} = \{W_s(y|x): x \in \mathcal{X}, y \in \mathcal{Y} \text{ and } s \in \mathcal{S}\}\), where \(\mathcal{S}\) is the finite state set of the channel. For each transmission, when the channel state \(s\) is given, the transition probability matrix of the channel is determined to be \(W_s\). Denote by \(X^N\) and \(Y^N\) the random sequences of the channel input and output, respectively. Then \(Y^N\) may be any one of the random sequences from \(\{Y^N(s^N): s^N \in \mathcal{S}^N\}\), where \(Y^N(s^N)\) is the channel output under the state sequence \(s^N\) satisfying that
\begin{equation}\label{eq:avc:1}
\operatorname{Pr}\{Y^N(s^N)=y^N| X^N=x^N\} =  W(y^N|x^N, s^N)
\end{equation}
with
\[
W(y^N|x^N, s^N) = \prod_{i=1}^N W_{s_i}(y_i|x_i)
\]
being the probability that channel outputs \(y^N\) when \(x^N\) is transmitted under the state sequence \(s^N\).
\end{defn}

To characterize the capacity results, let \(\mathcal{P}(\mathcal{S}) = \{q: 0 \leq  q(s) \leq 1, \sum_{s\in\mathcal{S}} q(s) = 1\}\) be the convex hull of \(\mathcal{S}\) and
\begin{equation}\label{eq:pre:13}
\bar{\mathcal{W}} = \{W_q: \mathcal{P}(\mathcal{S})\}
\end{equation}
be the convex hull of \(\mathcal{W}\) with
\begin{equation}\label{eq:pre:4}
W_{q}(y|x) = \sum_{s\in\mathcal{S}} q(s) W_s(y|x)
\end{equation}
being a convex combination of the transition matrices from \(\mathcal{W}\), where \(q \in \mathcal{P}(\mathcal{S})\). Notice that the size of \(\mathcal{W}\) is finite, but the size of \(\bar{\mathcal{W}}\) may be infinite.

\textbf{The capacity results}

As mentioned in Section \ref{sec:introduction}, the capacity of an AVC is related the classes of coding schemes (deterministic, stochastic or random) and the criteria of decoding error probability (average or maximal). In the rest of this subsection, we list the results of different classes of coding schemes, given in formulas \eqref{eq:pre:5}, \eqref{eq:pre:7} and \eqref{eq:pre:6}, all of which come from Ahlswede \cite{Ahlswede-1986}.

\emph{Deterministic Code\cite{Ahlswede-1986}.} A deterministic code over an AVC is specified by a pair of mappings \((f, \phi)\) with \(f: \mathcal{M} \mapsto \mathcal{X}^N\) and \(\phi: \mathcal{Y}^N \mapsto \mathcal{M}\). The maximal and average decoding error probabilities are defined as
      \[
      \lambda(\mathcal{W}, f, \phi) = \max_{s^N\in \mathcal{S}^N} \lambda(\mathcal{W}, f, \phi, s^N)
      \text{ and }
      \bar{\lambda}(\mathcal{W}, f, \phi) = \max_{s^N\in \mathcal{S}^N} \bar{\lambda}(\mathcal{W}, f, \phi, s^N),
      \]
      respectively, where
      \[
      \lambda(\mathcal{W}, f, \phi, s^N) = \max_{m \in \mathcal{M}}\lambda_m(\mathcal{W}, f, \phi, s^N),
      \]
      \[
      \bar{\lambda}(\mathcal{W}, f, \phi, s^N) = \frac{1}{|\mathcal{M}|}\sum_{m \in \mathcal{M}}\lambda_m(\mathcal{W}, f, \phi, s^N),
      \]
      \[
      \lambda_m(\mathcal{W}, f, \phi, s^N) =  e_m(\mathcal{W}, f, \phi, s^N) = 1- W(\phi^{-1}(m)|f(m), s^N)
      \]
      and
      \begin{equation}\label{eq:pre:3}
      W(\phi^{-1}(m)|f(m), s^N) = W(\phi^{-1}(m)|x^N, s^N) =\sum_{y^N\in\phi^{-1}(m)}W(y^N|x^N,s^N)
      \end{equation}
      is the probability that the channel output lies in the decoding set \(\phi^{-1}(m)\) of the message \(m\) when the channel input is \(f(m)=x^N\) and the state sequence is \(s^N\).

The transmission rate of the code is given by
\[
R = \frac{1}{N} \log |\mathcal{M}|
\]
and the capacity of the deterministic code over the AVC \(\mathcal{W}\) with respect to the average decoding error probability is given by
\begin{equation}\label{eq:pre:5}
\bar{C}^{DC}(\mathcal{W}) = \max_{X} \min_{q \in \mathcal{P}(\mathcal{S}) }I(X; Y_q)
\end{equation}
if \(\bar{C}^{DC}(\mathcal{W}) > 0\), where
\begin{equation}\label{eq:pre:9}
\Pr\{X=x,Y_q=y\}=P_X(x)W_q(y|x)
\end{equation}
with \(W_q\) given by \eqref{eq:pre:4}. The capacity of the maximal decoding error probability is still unknown.

\emph{Stochastic Code\cite{Ahlswede-1986}.} A stochastic code is specified by a pair \((F, \phi)\), where \(\{F(m), m \in \mathcal{M}\}\) is a collection of random sequences distributed over \(\mathcal{X}^N\) and \(\phi\) is a deterministic mapping \(\mathcal{Y}^N \mapsto \mathcal{M}\). The maximal and average decoding error probabilities are given by
\begin{equation}\label{eq:pre:8}
      \lambda(\mathcal{W}, F, \phi) = \max_{s^N\in \mathcal{S}^N} \lambda(\mathcal{W}, F, \phi, s^N)
      \text{ and }
      \bar{\lambda}(\mathcal{W}, F, \phi) = \max_{s^N\in \mathcal{S}^N} \bar{\lambda}(\mathcal{W}, F, \phi, s^N)
\end{equation}
respectively, where
\begin{equation}\label{eq:pre:10}
      \lambda(\mathcal{W}, F, \phi, s^N) = \max_{m \in \mathcal{M}}\lambda_m(\mathcal{W}, F, \phi, s^N), \bar{\lambda}(\mathcal{W}, F, \phi, s^N) = \frac{1}{|\mathcal{M}|}\sum_{m \in \mathcal{M}}\lambda_m(\mathcal{W}, F, \phi, s^N)
\end{equation}
with
\begin{equation}\label{eq:pre:11}
      \lambda_m(\mathcal{W}, F, \phi, s^N) = E[e_m(\mathcal{W}, F, \phi, s^N)] \text{ and } e_m(\mathcal{W}, F, \phi, s^N) = 1- W(\phi^{-1}(m)|F(m), s^N).
\end{equation}
Notice that \(W(\phi^{-1}(m)|F(m), s^N)\) is a discrete random function of \(F(m)\). To be precise,
\[
\Pr\{W(\phi^{-1}(m)|F(m), s^N) = a\} = \sum_{x^N\in\mathcal{X}^N:W(\phi^{-1}(m)|x^N, s^N)=a}\Pr\{F(m)=x^N\}
\]
for every \(a\in [0,1]\), where \(W(\phi^{-1}(m)|x^N, s^N)\) is given by \eqref{eq:pre:3}.

The capacities of stochastic code over the AVC \(\mathcal{W}\) with respect to the maximal and average decoding probabilities are identical to \(\bar{C}^{DC}(\mathcal{W})\), i.e.
\begin{equation}\label{eq:pre:7}
C^{SC}(\mathcal{W}) = \bar{C}^{SC}(\mathcal{W})= \bar{C}^{DC}(\mathcal{W}).
\end{equation}

\emph{Random code\cite{Ahlswede-1986}.} A random code is specified by a pair of random encoder and decoder \((F, \Phi)\) distributed over a collection of deterministic encoder and decoder pairs \(\{(f_g,\phi_g): g \in \mathcal{G}\}\), where \(\mathcal{G}\) is the index set of the deterministic codes, whose size is related to the codeword length \(N\). The maximal and average decoding error probabilities are given by
\[
      \lambda(\mathcal{W}, F, \Phi) = \max_{s^N\in \mathcal{S}^N} \lambda(\mathcal{W}, F, \Phi, s^N)
      \text{ and }
      \bar{\lambda}(\mathcal{W}, F, \Phi) = \max_{s^N\in \mathcal{S}^N} \bar{\lambda}(\mathcal{W}, F, \Phi, s^N)
\]
respectively, where
\begin{equation}\label{eq:pre:15}
      \lambda(\mathcal{W}, F, \Phi, s^N) = \max_{m \in \mathcal{M}}\lambda_m(\mathcal{W}, F, \Phi, s^N), \bar{\lambda}(\mathcal{W}, F, \Phi, s^N) = \frac{1}{|\mathcal{M}|}\sum_{m \in \mathcal{M}}\lambda_m(\mathcal{W}, F, \Phi, s^N)
\end{equation}
with
\[
      \lambda_m(\mathcal{W}, F, \Phi, s^N) = E[e_m(\mathcal{W}, F, \Phi, s^N)] \text{ and } e_m(\mathcal{W}, F, \Phi, s^N) = 1- W(\Phi^{-1}(m)|F(m), s^N).
\]
Notice that \(W(\Phi^{-1}(m)|F(m), s^N)\) is a discrete random function of \((F,\Phi)\). To be precise,
\[
\Pr\{W(\Phi^{-1}(m)|F(m), s^N) = a\} = \sum_{g\in\mathcal{G}:W(\phi_g^{-1}(m)|f_g(m), s^N)=a}\Pr\{(F,\Phi)=(f_g,\phi_g)\}
\]
for every \(a\in [0,1]\), where \(W(\phi_g^{-1}(m)|f_g(m), s^N)\) is given by \eqref{eq:pre:3}.

The capacities of random code over AVC \(\mathcal{W}\) with respect to the maximal and average decoding probabilities are identical. They follow that
\begin{equation}\label{eq:pre:6}
C^{RC}(\mathcal{W}) = \bar{C}^{RC}(\mathcal{W}) = \max_{X} \min_{q \in \mathcal{P}(\mathcal{S}) }I(X; Y_q),
\end{equation}
where \(\{(X,Y_q): q\in\mathcal{P}(\mathcal{S})\}\) is a collection of random variables satisfying \eqref{eq:pre:9}.

\begin{rem}
We present some remarks here to have a further explanation on the coding schemes and capacity results.
\begin{itemize}
  \item A popular technique to establish the standard channel coding theorem is the random coding scheme, but notice that the random coding scheme is essentially different from the random code discussed here. In fact, the random coding scheme would produce a deterministic code.
  \item A stochastic code is a special case of a random code, with the decoder being deterministic. To show this, let \((F,\phi)\) be a stochastic code. Moreover, since the decoding error probabilities defined in \eqref{eq:pre:8} are only related to the margin distributions of the random codewords \((F(m),m\in\mathcal{M})\), we assume that the random codewords are mutually independent without loss of generality. Let \(\{(f_g,\phi): g\in\mathcal{G}\}\) be the collection of all possible deterministic codes of length \(N\), with the decoder being fixed. Then, the stochastic code is actually a random code distributed over \(\{(f_g,\phi): g\in\mathcal{G}\}\), such that
      \[
      \Pr\{(F,\phi)=(f_g,\phi)\}=\prod_{m\in\mathcal{M}}\Pr\{F(m)=f_g(m)\}.
      \]
  \item Formulas \eqref{eq:pre:5}, \eqref{eq:pre:7} and \eqref{eq:pre:6} give that the five capacities, except the one of deterministic code with respect to the maximal decoding probability, are identical if they are all positive, i.e.
      \[
      \bar{C}^{DC}(\mathcal{W})=\bar{C}^{SC}(\mathcal{W})=C^{SC}(\mathcal{W})=\bar{C}^{RC}(\mathcal{W})=C^{RC}(\mathcal{W})
      \]
      if \(\bar{C}^{DC}(\mathcal{W}) > 0\). Notice that the precondition \(\bar{C}^{DC}(\mathcal{W}) > 0\) is important since it is possible that \(\bar{C}^{DC}(\mathcal{W})=\bar{C}^{SC}(\mathcal{W})=C^{SC}(\mathcal{W})=0\) while \(\bar{C}^{RC}(\mathcal{W})=C^{RC}(\mathcal{W})>0\).
  \item When the random code is applied, there exists a third party providing extra information to tell the transmitter and the legitimate receiver which pair of deterministic encoder and decoder is to be used before each transmission. The information from the third party is called the CR of the random code. Throughout this paper, we assume that the CR contains full information of both encoder and decoder, different from the settings in \cite{Bjelakovic-2013}.
\end{itemize}
\end{rem}

\subsection{Typicality under state sequence}\label{sec:typical}

In the model of arbitrarily varying channel (AVC), the transition probability of the output sequence under a given input sequence is related to the value of the state sequence. Therefore, to define the typical sequences over the AVC, it is reasonable to take the state sequence into consideration. This subsection gives the definitions of typical sequence under the state sequence, and lists some useful properties. These results provide new ideas dealing with the reliable and secure transmission problems of AVC and AVWC. The new definitions originate from the letter typical sequences defined in Chapter 1 of \cite{MU_IT}. We first present the original definitions, followed by the new ones.

\textbf{The original definitions}

For any \(\delta \geq 0\), the \(\delta\)-letter typical set \(T^N_\delta (P_X)\) with respect to the probability mass function \(P_X\) on \(\mathcal{X}\), is the set of \(x^N \in \mathcal{X}^N\) satisfying
\[
|\frac{1}{N}N(a:x^N) - P_X(a)| \leq \delta P_X(a) \text{ for all } a \in \mathcal{X},
\]
where \(N(a:x^N)\) is the number of positions of \(x^N\) having the letter \(a \in \mathcal{X}\). The jointly typical set \(T^N_\delta (P_{XY})\) with respect to the joint probability mass function \(P_{XY}\) and the conditionally typical set \(T^N_\delta(P_{XY}|x^N)\) of \(x^N \in \mathcal{X}^N\), are defined accordingly.

\textbf{The new definitions}

The propositions and corollaries claimed in this part will be briefly proved in Appendix \ref{app:typical}.

For any \(s^N \in \mathcal{S}^N\), \(a \in \mathcal{S}\) and \(\eta >¡¡0\), denote by
\begin{equation}\label{eq:pre:16}
\mathcal{I}(a: s^N) = \{1\leq i \leq N, s_i = a\}
\end{equation}
the collection of indices of components in $s^N$ whose values are $a$, and let
\begin{equation}\label{eq:pre:17}
\mathcal{S}(s^N, \eta) = \{a \in \mathcal{S}, |\mathcal{I}(a: s^N)| > \frac{N\eta}{|\mathcal{S}|}\}.
\end{equation}
Furthermore, define the type of a given state sequence \(s^N\) as a probability mass function \(P_{s^N}\) over \(\mathcal{S}\) such that
\[
P_{s^N}(a) = |\mathcal{I}(a: s^N)|/N
\]
for \(a\in\mathcal{S}\).

\begin{rem}
The collection \(\mathcal{S}(s^N, \eta)\) defined in \eqref{eq:pre:17} contains the states which occur sufficiently many times in the sequence \(s^N\). This collection is critical to the new definition of typical sequence. Roughly speaking, we are only interested in the states from \(\mathcal{S}(s^N, \eta)\), and ignore the states outside.
\end{rem}

\begin{defn}\label{def:strong}
   The letter typical set \(\tilde{T}^N[X,s^N]_{\delta,\eta}\) with respect to the random variable \(X\) under the state sequence \(s^N\), is the set of \(x^N \in \mathcal{X}^N\) such that \(P_X(x_i) > 0\) for all \(1\leq i \leq N\), and \(x_{\mathcal{I}(a: s^N)} \in T^{\mu_a}_{\delta}(P_X)\) for all \(a \in \mathcal{S}(s^N, \eta)\), where \(x_{\mathcal{I}(a: s^N)} = (x_i, i \in \mathcal{I}(a: s^N))\) is a \(\mu_a\)-subvector of \(x^N\) and \(\mu_a = |\mathcal{I}(a: s^N)|\).
\end{defn}

We provide an example below to help the reader have a clearer idea on the definition.

 \begin{exmp}
 Let \(\mathcal{S}=\mathcal{X}=\{0,1\}\), and \(X\) be the random variable satisfying
 \[
 \Pr\{X=0\} = \Pr\{X=1\} = \frac{1}{2}.
 \]
 Set \(\eta=0.5\) and \(\delta = 0.12\). Then we have the following conclusions.
 \begin{itemize}
   \item Suppose that \(N=20\), \(s^N=00000000001111111111\) and \(x^N=00000111110000001111\). It follows that \(\mathcal{S}(s^N,\eta)=\{0,1\}\), \(\mathcal{I}_0=\mathcal{I}(0: s^N)=\{1,2,...,10\}\), \(\mathcal{I}_1=\mathcal{I}(1: s^N)=\{11,12,...,20\}\), \(x_{\mathcal{I}_0}=0000011111\) and \(x_{\mathcal{I}_1}=0000001111\). This indicates that \(x_{\mathcal{I}_0}\in T^{10}_{\delta}(P_X)\) and \(x_{\mathcal{I}_1}\in T^{10}_{\delta}(P_X)\). Therefore, \(x^N\in \tilde{T}^N[X, s^N]_{\delta,\eta}\).
   \item Suppose that \(N=20\), \(s^N=00000000001111111111\) and \(x^N=00000111110000000111\). It follows that \(\mathcal{S}(s^N,\eta)=\{0,1\}\), \(\mathcal{I}_0=\mathcal{I}(0: s^N)=\{1,2,...,10\}\), \(\mathcal{I}_1=\mathcal{I}(1: s^N)=\{11,12,...,20\}\), \(x_{\mathcal{I}_0}=0000011111\) and \(x_{\mathcal{I}_1}=0000000111\). This indicates that \(x_{\mathcal{I}_0}\in T^{10}_{\delta}(P_X)\) and \(x_{\mathcal{I}_1}\notin T^{10}_{\delta}(P_X)\). Therefore, \(x^N\notin \tilde{T}^N[X, s^N]_{\delta,\eta}\).
   \item Suppose that \(N=20\), \(s^N=00000000000000000011\) and \(x^N=00000000011111111111\). It follows that \(\mathcal{S}(s^N,\eta)=\{0\}\), \(\mathcal{I}_0=\mathcal{I}(0: s^N)=\{1,2,...,18\}\) and \(x_{\mathcal{I}_0}=000000000111111111\). This indicates that \(x_{\mathcal{I}_0}\in T^{10}_{\delta}(P_X)\). Therefore, \(x^N\in \tilde{T}^N[X, s^N]_{\delta,\eta}\). \textbf{Notice that \(1\notin \mathcal{S}(s^N,\eta)\), so there is no need to consider the typicality of \(x_{\mathcal{I}_1}\), where \(\mathcal{I}_1=\mathcal{I}(1: s^N)=\{19,20\}\).}
 \end{itemize}
 \end{exmp}

 \begin{prop}\label{remarkTX}
  Suppose that \(X_1, X_2,...,X_N\) are \(N\) i.i.d. random variables with the same generic probability mass function as that of \(X\). For any given \(\eta > 0\) and \(\delta < m_X\), it follows that
  \[\operatorname{Pr}\{X^N \in \tilde{T}^N[X, s^N]_{\delta,\eta}\} > 1 - 2^{-N\nu_1}\]
  for some \(\nu_1 > 0\) \textbf{unrelated to the length \(N\) and the state sequence \(s^N\)}, where \(m_X=\min_{x \in \mathcal{X}:P_X(x)>0}P_X(x)\).
 \end{prop}

\begin{defn}\label{def:strong2}
   The letter typical set \(\tilde{T}^N[Y_{\mathcal{S}}, s^N]_{\delta,\eta}\) with respect to the random variables \(Y_{\mathcal{S}}\) under the state sequence \(s^N\) is the set of \(y^N \in \mathcal{Y}^N\) such that \(P_{Y_{s_i}}(y_i) > 0\) for all \(1\leq i \leq N\), and \(y_{\mathcal{I}(a: s^N)} \in T^{\mu_a}_{\delta}(P_{Y_a})\) for all \(a \in \mathcal{S}(s^N, \eta)\).
\end{defn}

Interestingly, we have two definitions on typical sequences with respect to state sequence, namely Definitions \ref{def:strong} and \ref{def:strong2}. This, in fact, comes from the property of AVC. Definition \ref{def:strong} is for the typicality of the channel input, whose distribution is independent of the state, while Definition \ref{def:strong2} is for the output, whose distribution is related to the state.

\begin{prop}\label{remarkTY}
Let \(Y^N(s^N)\) be a random vector satisfying that
\begin{equation}\label{eq:pre:2}
\begin{array}{c}
\operatorname{Pr}\{Y^N(s^N) = y^N\} = \prod_{i=1}^N \operatorname{Pr}\{Y_{s_i} = y_i\}.
\end{array}
\end{equation}
Then for any \(y^N \in \tilde{T}^N[Y_{\mathcal{S}}, s^N]_{\delta,\eta}\), it follows that
\[
2^{-N[(1+\delta)\bar{H}_{P_{s^N}}(Y_{\mathcal{S}})-\eta\log m_{Y_{\mathcal{S}}}]} <  \operatorname{Pr}\{Y^N(s^N) = y^N\} < 2^{-N(1-\delta-\eta)\bar{H}_{P_{s^N}}(Y_{\mathcal{S}})},
\]
where \(m_{Y_{\mathcal{S}}} = \min_{(y,s)\in \mathcal{Y}\times\mathcal{S}: P_{Y_s}(y) > 0} P_{Y_s}(y)\) and \(P_{s^N}\) is the type of \(s^N\).
\end{prop}

\begin{cor}\label{cor2}
It is satisfied that \(|\tilde{T}^N[Y_{\mathcal{S}}, s^N]_{\delta,\eta}| < 2^{N[(1+\delta)\bar{H}_{P_{s^N}}(Y_{\mathcal{S}})-\eta\log m_{Y_{\mathcal{S}}}]}\).
\end{cor}

\begin{defn}\label{def:joint}
The jointly typical set \(\tilde{T}^N[XY_{\mathcal{S}}, s^N]_\delta\) with respect to \((X,Y_{\mathcal{S}})\) under the state sequence \(s^N\) is the set of \((x^N,y^N) \in \mathcal{X}^N \times \mathcal{Y}^N\) satisfying \(P_{XY_{s_i}}(x_i, y_i) > 0\) for all \(1\leq i \leq N\), and \((x_{\mathcal{I}(a:s^N)}, y_{\mathcal{I}(a:s^N)}) \in T^{\mu_a}_{\delta}(P_{XY_a})\) for all \(a \in \mathcal{S}(s^N, \eta)\), where \(\mu_a = |\mathcal{I}(a:s^N)|\).
\end{defn}

\begin{prop}\label{remarkConsist}
\((x^N,y^N) \in  \tilde{T}^N[XY_{\mathcal{S}},s^N]_{\delta,\eta}\) implies that \(x^N \in  \tilde{T}^N[X,s^N]_{\delta,\eta}\) and \(y^N \in  \tilde{T}^N[Y_{\mathcal{S}},s^N]_{\delta,\eta}\).
\end{prop}

\begin{defn}\label{def:condition}
For any given \(x^N \in \mathcal{X}^N\), the conditionally typical set of \(x^N\) under the state sequence \(s^N\) is defined as
\[
\tilde{T}^N[XY_{\mathcal{S}}, s^N|x^N]_{\delta,\eta}=\{y^N \in \mathcal{Y}^N: (x^N,y^N) \in  \tilde{T}^N[XY_{\mathcal{S}}, s^N]_{\delta,\eta}\}.
\]
\end{defn}

\begin{prop}\label{rem:joint}
Let \((X^N, Y^N(s^N))\) be a pair of random sequences with the conditional mass function
\begin{equation}\label{eq:pre:1}
\begin{array}{c}
\operatorname{Pr}\{Y^N(s^N)=y^N|X^N=x^N\} = \prod_{i=1}^N P_{Y_{s_i}|X}(y_i|x_i).
\end{array}
\end{equation}
Then for any \((x^N, y^N) \in \tilde{T}^N[XY_{\mathcal{S}},s^N]_{\delta,\eta}\), it follows that
\[
2^{-N[(1+\delta)\bar{H}_{P_{s^N}}(Y_{\mathcal{S}}|X) - \eta \log m_{XY_{\mathcal{S}}}]}
\leq \operatorname{Pr}\{Y^N(s^N)=y^N|X^N=x^N\}
\leq 2^{-N(1-\delta-\eta)\bar{H}_{P_{s^N}}(Y_{\mathcal{S}}|X)},
\]
where
\[
m_{XY_{\mathcal{S}}} = \min_{(x,y,s)\in \mathcal{X}\times\mathcal{Y}\times\mathcal{S}: P_{X}(x)W_s(y|x) > 0} P_{X}(x)W_s(y|x).
\]
\end{prop}

\begin{cor}\label{cor1}
Let \(Y^N(s^N)\) be a random vector satisfying Formula \eqref{eq:pre:2}. For any \(x^N \in \tilde{T}^N[X,s^N]_{\delta,\eta}\), it follows that
\[
\operatorname{Pr}\{Y^N(s^N) \in \tilde{T}^N[XY_{\mathcal{S}}, s^N|x^N]_{\delta,\eta}\}
< 2^{-N[\bar{I}_{P_{s^N}}(X; Y_{\mathcal{S}}) - (2\delta +\eta) \bar{H}_{P_{s^N}}(Y_{\mathcal{S}}) + \eta \log m_{XY_{\mathcal{S}}}]}.
\]
\end{cor}

\begin{prop}\label{remarkTXY}
Let \((X^N, Y^N(s^N))\) be a pair of random sequences satisfying \eqref{eq:pre:1}. For any \(x^N \in \tilde{T}^N[X,s^N]_{\delta,\eta}\) and \(\delta < m_{XY_{\mathcal{S}}}\), we have
\[
\operatorname{Pr}\{Y^N(s^N) \in \tilde{T}^N[XY_{\mathcal{S}}, s^N|x^N]_{2\delta, \eta} | X^N = x^N\} > 1 - 2^{-N\nu_2}
\]
for some \(\nu_2 > 0\) \textbf{unrelated to the length \(N\) and state sequence \(s^N\)}.
\end{prop}

\begin{rem}\label{rem:important}
If Proposition \ref{remarkTX} holds for some \(\nu_1 > 0\), then it still holds if we decrease the value of \(\nu_1\). That means the value of \(\nu_1\) can be sufficiently small, so can the value of \(\nu_2\) in Proposition \ref{remarkTXY}. In the subsequent sections, we always treat \(\nu_1\) and \(\nu_2\) as constant real numbers, whose values are sufficiently small. Moreover, the property that \(\nu_1\) and \(\nu_2\) are unrelated to \(s^N\), is of great significance. \textbf{This indicates that Propositions \ref{remarkTX} and \ref{remarkTXY} are possible to hold simultaneously for all state sequences from \(\mathcal{S}^N\)}.
\end{rem}

\section{Communication Models and Main Results}\label{sec:main}

This section presents the formal statements of the three communication models in this paper, and lists the main results.

Subsection \ref{sec:avc:csr} provides the capacity results of AVC-CSR. Unlike the classic AVC, capacities of stochastic and random codes over an AVC-CSR are identical. This model was first studied in \cite{stambler-1975}. We present the results here as a preliminary to the discussion in Subsection \ref{sec:main:avwc:csr}.

Subsection \ref{sec:avwc} discusses the secrecy capacity results of AVWC. Lower bounds on secrecy capacities of random code and stochastic code over the AVWC are given. Stochastic code is a key technique to deal with the problem of secure transmission over the wiretap channel \cite{wiretap}. It is also clear that random code does not provide a larger secrecy capacity of a discrete memoryless wiretap channel than stochastic code. However, it is proved in \cite{Bjelakovic-2013} that the secrecy capacity of random code over AVWC can be strictly larger than that of stochastic code. In fact, that kind of situation happens only when the capacity of stochastic code over the main AVC is 0.

Subsection \ref{sec:avwc:csr} discusses the secrecy capacity results of AVWC-CSR. It is proved that the secrecy capacity of stochastic code over the AVWC-CSR is identical to that of random code. However, the value of secrecy capacity is unknown in general. A lower bound on the secrecy capacity is given in that subsection.

When talking about the secrecy capacity, we always consider it with respect to the maximal decoding error probability since it is more challenging than that of the average decoding error.

\subsection{AVC with channel state sequence known at the receiver}\label{sec:avc:csr}

This subsection introduces the capacity results of AVC-CSR. Just like the discussion of AVC in Subsection \ref{sec:avc}, we will discuss the capacities of different classes of coding schemes with respect to different criteria of decoding error probabilities. Theorem \ref{thm:avc:csr} claims that the capacities of the following five cases are identical: 1) deterministic code with respect to the average decoding error probability, 2) stochastic code with the average and 3) maximal decoding error probabilities, and 4) random code with the average and 5) maximal decoding error probabilities. The capacity of deterministic code with the maximal decoding error probability is unknown. Proposition \ref{thm:avc:csr:dc} discusses the positivity of it.

\textbf{The definitions}

\begin{figure}
     \centering
     \includegraphics[width=12cm]{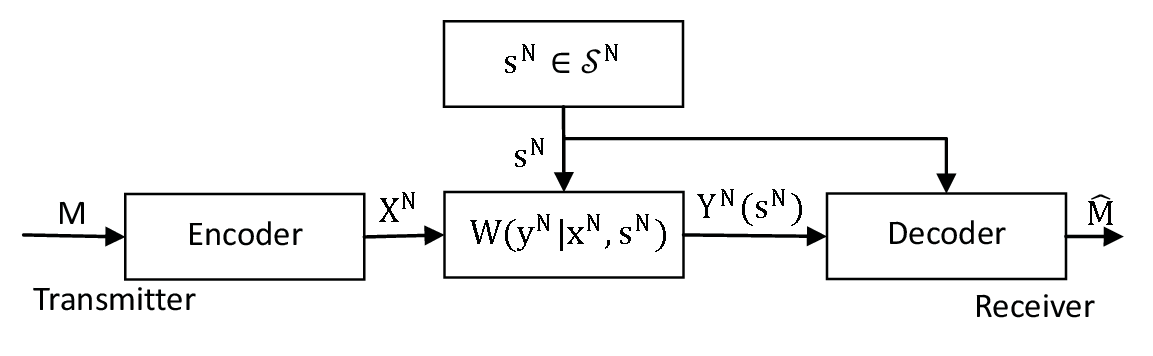}
     \caption{AVC-CSR, Arbitrarily varying channel with channel state sequence known at the receiver.}\label{Fig:avc:csr}
 \end{figure}

\begin{defn}
\emph{(AVC-CSR, Arbitrarily Varying Channel with channel states known at the receiver)} Just like the AVC, an AVC-CSR is also specified by a finite collection \(\mathcal{W}\) of transition probability matrices. The only difference is that the receiver has access to the channel state sequence in the case of AVC-CSR. This would affect the definition of the decoder. See Fig. \ref{Fig:avc:csr}.
\end{defn}

\begin{defn}\label{def:avc:csr:dc}
\emph{(Deterministic code over AVC-CSR)} Let an AVC-CSR \(\mathcal{W}\) be given. A deterministic code over it is specified by a pair of mappings \((f, \phi)\) with \(f: \mathcal{M} \mapsto \mathcal{X}^N\) and \(\phi: \mathcal{Y}^N \times \mathcal{S}^N \mapsto \mathcal{M}\). The maximal and average decoding error probabilities are defined as
      \[
      \lambda^{CSR}(\mathcal{W}, f, \phi) = \max_{s^N\in \mathcal{S}^N} \lambda^{CSR}(\mathcal{W}, f, \phi, s^N)
      \text{ and }
      \bar{\lambda}^{CSR}(\mathcal{W}, f, \phi) = \max_{s^N\in \mathcal{S}^N} \bar{\lambda}^{CSR}(\mathcal{W}, f, \phi, s^N)
      \]
      respectively, where
      \[
      \lambda^{CSR}(\mathcal{W}, f, \phi, s^N) = \max_{m \in \mathcal{M}}\lambda_m^{CSR}(\mathcal{W}, f, \phi, s^N), \bar{\lambda}^{CSR}(\mathcal{W}, f, \phi, s^N) = \frac{1}{|\mathcal{M}|}\sum_{m \in \mathcal{M}}\lambda_m^{CSR}(\mathcal{W}, f, \phi, s^N)
      \]
      and
      \[
      \lambda_m^{CSR}(\mathcal{W}, f, \phi, s^N) =  e_m^{CSR}(\mathcal{W}, f, \phi, s^N) = 1- W(\phi^{-1}(m, s^N)|f(m), s^N)
      \]
with the decoding set \(\phi^{-1}(m, s^N)\) given by
\[
\phi^{-1}(m, s^N) =\{y^N\in\mathcal{Y}^N: \phi(y^N,s^N)=m\}.
\]
\end{defn}

The stochastic code and random code over AVC-CSR are defined accordingly.

\textbf{The capacities}

Let the capacities of the deterministic, stochastic and random codes over a given AVC-CSR \(\mathcal{W}\) with respect to the maximal and average decoding error probabilities be denoted by \(C^{CSR-DC}(\mathcal{W})\), \(\bar{C}^{CSR-DC}(\mathcal{W})\), \(C^{CSR-SC}(\mathcal{W})\) \(\bar{C}^{CSR-SC}(\mathcal{W})\), \(C^{CSR-RC}(\mathcal{W})\) and \(\bar{C}^{CSR-RC}(\mathcal{W})\), respectively. We have the following results.
\begin{thm}\label{thm:avc:csr}
For every AVC-CSR \(\mathcal{W}\), the capacities of the deterministic, stochastic and random codes satisfy that
\[
\bar{C}^{CSR-DC}(\mathcal{W}) = C^{CSR-SC}(\mathcal{W}) = \bar{C}^{CSR-SC}(\mathcal{W}) = C^{CSR-RC}(\mathcal{W}) = \bar{C}^{CSR-RC}(\mathcal{W}) = \max_{X}\min_{s\in\mathcal{S}} I(X;Y_s),
\]
where \(\{(X,Y_s):s\in\mathcal{S}\}\) is a collection of random variables satisfying \eqref{eq:measure:1}.
\end{thm}

The converse half of Theorem \ref{thm:avc:csr} is established by recalling the facts that the capacity of an AVC-CSR cannot exceed that of the corresponding compound channel, and that the capacity of the compound channel specified by $\mathcal{W}$ is $\max_{X}\min_{s\in\mathcal{S}} I(X;Y_s)$. The direct half of Theorem \ref{thm:avc:csr} was first established in \cite{stambler-1975}, see also Problem 2.6.14 in \cite{CT_DMS}. We present the proof of the direct half in Section \ref{sec:thm:avc:csr} as a preliminary to the proof of Theorem \ref{thm:avwc:csr:rc}.

\begin{rem}\label{rem:avc:csr}
Comparing the results in Section \ref{sec:avc} and that in Theorem \ref{thm:avc:csr}, we see that the capacity of an AVC may be strictly smaller than that of the corresponding AVC-CSR. For example, let
\[
\mathcal{W} = \Bigg\{
\Bigg[
\begin{matrix}
1 & 0\\
0 & 1
\end{matrix}
\Bigg]
,
\Bigg[
\begin{matrix}
0 & 1\\
1 & 0
\end{matrix}
\Bigg]
 \Bigg\}.
\]
It follows clearly that the capacity of AVC \(\mathcal{W}\) is 0 even if random code is applied. On the other hand, the transmission rate \(R =  1\) can be achieved over AVC-CSR \(\mathcal{W}\) by a deterministic code since the channel is totally noiseless when the channel states are known at the receiver.
\end{rem}

The capacity of deterministic code over the AVC-CSR with respect to the maximal decoding error probability, i.e., the value of \(C^{CSR-DC}(\mathcal{W})\), is still unknown. The following proposition gives the necessary and sufficient condition ensuring the positivity of \(C^{CSR-DC}(\mathcal{W})\).

\begin{prop}\label{thm:avc:csr:dc}
For every AVC-CSR \(\mathcal{W}\), the capacity \(C^{CSR-DC}(\mathcal{W})\) of deterministic code with respect to the maximal decoding error probability is positive if and only if there exist a pair of \(x,x' \in \mathcal{X}\), such that for all channel states \(s \in \mathcal{S}\), there exists \(y \in \mathcal{Y}\) satisfying \(W(y|x,s) \neq W(y|x',s).\)
\end{prop}

The proof of Proposition \ref{thm:avc:csr:dc} is given in Appendix \ref{app:thm:avc:csr:dc}.

\begin{rem}
The capacity of deterministic code over an AVC-CSR with respect to the average decoding error can be strictly larger than that of maximal decoding error. As an example, let \(\mathcal{S}=\mathcal{X}=\{0,1,2\}\), \(\mathcal{Y}=\{0,1\}\) and
\[
\mathcal{W} = \Bigg\{
\Bigg[
\begin{matrix}
1/3 & 2/3\\
2/3 & 1/3\\
1/2 & 1/2\\
\end{matrix}
\Bigg]
,
\Bigg[
\begin{matrix}
1/3 & 2/3\\
1/2 & 1/2\\
2/3 & 1/3\\
\end{matrix}
\Bigg]
,
\Bigg[
\begin{matrix}
1/2 & 1/2\\
1/3 & 2/3\\
2/3 & 1/3\\
\end{matrix}
\Bigg]
 \Bigg\}.
\]
Proposition \ref{thm:avc:csr:dc} claims that the capacity of deterministic code with the maximal decoding error is 0, while Theorem \ref{thm:avc:csr} claims that the capacity with respect to the average decoding error is \((4-2\log 3)/3\).
\end{rem}

\subsection{AVWC with state sequence unknown at the receiver}\label{sec:avwc}

This subsection gives a pair of lower bounds on the secrecy capacities of stochastic and random codes over the AVWC with respect to the strong secrecy criterion. The formal definition of the general AVWC is given in Definition \ref{def:avwc}, see also Remark \ref{rem1} for the explanation of generality. Definitions \ref{def:avwc:sc} and \ref{def:avwc:rc} formulate the achievability of stohastic code and random code, respectively. The corresponding lower bounds are given in Theorems \ref{thm:avwc:rc} and \ref{thm:avwc:sc}, respectively.

\begin{figure}
     \centering
     \includegraphics[width=12cm]{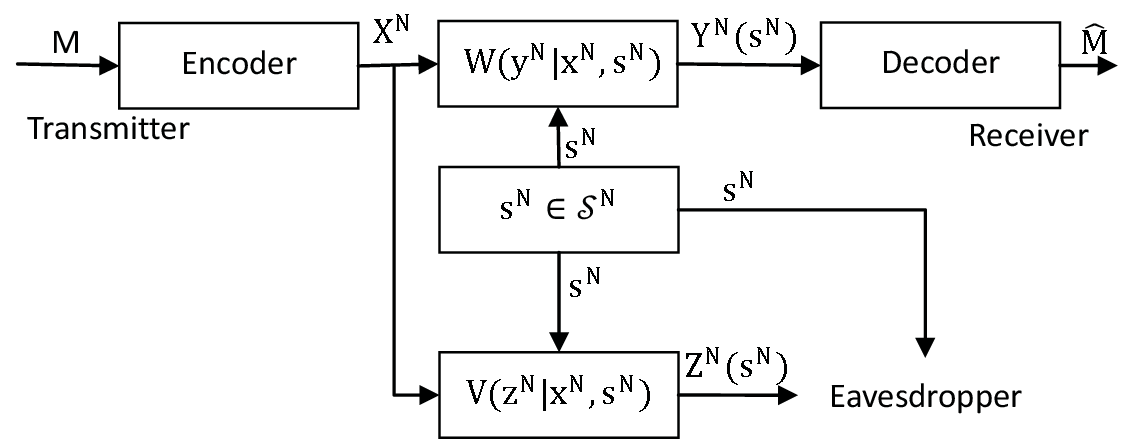}
     \caption{AVWC, Arbitrarily varying wiretap channel with channel state sequence unknown at the receiver.}\label{Fig:avwc}
 \end{figure}

 \textbf{The definitions}

\begin{defn}\label{def:avwc}
\emph{(AVWC, arbitrarily varying wiretap channel)} An AVWC, depicted in Fig. \ref{Fig:avwc}, is specified by a pair \((\mathcal{W}, \mathcal{V})\), where \(\mathcal{W} = \{W_s(y|x): x\in\mathcal{X}, y\in\mathcal{Y} \text{ and } s\in\mathcal{S}\}\) is a finite collection of transition probability matrices specifying the main AVC, and \(\mathcal{V} = \{V_s(z|x): x\in\mathcal{X}, z\in\mathcal{Z} \text{ and }  s\in\mathcal{S}\}\) specifies the wiretap AVC. Let \(X^N\) be the input of the channels, and \(Y^N(s^N)\) and \(Z^N(s^N)\) be the outputs of main AVC and wiretap AVC, respectively, when the state sequence is \(s^N\). We have
\begin{equation}\label{eq:def:avwc:y}
\operatorname{Pr}\{Y^N(s^N)=y^N|X^N=x^N\}
= W(y^N|x^N,s^N)
= \prod_{i=1}^N W_{s_i}(y_i|x_i)
\end{equation}
and
\begin{equation}\label{eq:def:avwc:z}
\operatorname{Pr}\{Z^N(s^N)=z^N|X^N=x^N\}
= V(z^N|x^N,s^N)
= \prod_{i=1}^N V_{s_i}(z_i|x_i).
\end{equation}
\end{defn}

\begin{rem}\label{rem1}
\emph{(Generality of Definition \ref{def:avwc})} In the definition of AVWC, we assume that the main AVC and the wiretap AVC share the same state set, which is possible to be false in general. However, any communication model of AVWC can be transformed into an equivalent model covered by Definition \ref{def:avwc}. To show this, consider an AVWC \((\mathcal{W}, \mathcal{V})\) with \(\mathcal{W} = \{W_s: s \in \mathcal{S}\}\) and \(\mathcal{V} = \{V_t: t \in \mathcal{T}\}\), where \(\mathcal{T}\) is a finite state set of the wiretap AVC \(\mathcal{V}\). Given the state sequences \(s^N\) of the main AVC and \(t^N\) of the wiretap AVC, the channel input and outputs satisfy that
\[
\operatorname{Pr}\{Y^N(s^N)=y^N|X^N=x^N\}
= W(y^N|x^N,s^N)
= \prod_{i=1}^N W_{s_i}(y_i|x_i)
\]
and
\[
\operatorname{Pr}\{Z^N(t^N)=z^N|X^N=x^N\}
= V(z^N|x^N,t^N)
= \prod_{i=1}^N V_{t_i}(z_i|x_i).
\]
The equivalent AVWC \((\tilde{\mathcal{W}}, \tilde{\mathcal{V}})\) can be constructed with \(\tilde{\mathcal{W}} = \{\tilde{W}_{s,t} = W_s: (s, t) \in \mathcal{S} \times \mathcal{T}\}\) and \(\tilde{\mathcal{V}} = \{\tilde{V}_{s,t} = V_t: (s, t) \in \mathcal{S} \times \mathcal{T}\}\), where the main AVC and the wiretap AVC share the same state set \(\mathcal{S} \times \mathcal{T}\). This indicates that the model of AVWC in Definition \ref{def:avwc} is general.
\end{rem}

\begin{defn}\label{def:avwc:sc}
\emph{(Secure achievability of stochastic code over AVWC)} For any given AVWC \((\mathcal{W}, \mathcal{V})\), a non-negative real number \(R\) is said to be achievable by stochastic code, with respect to the maximal decoding error probability and the strong secrecy criterion, if for every \(\epsilon > 0\), there exists a stochastic code \((F, \phi)\) over that AVWC such that
\[
\frac{1}{N} \log |\mathcal{M}| > R - \epsilon, \lambda(\mathcal{W}, F, \phi) < \epsilon
\text{ and } \max_{s^N \in \mathcal{S}^N} I(M; Z^N(s^N)) < \epsilon
\]
when \(N\) is sufficiently large, where \(M\) is the source message uniformly distributed over the message set \(\mathcal{M}\) and \(Z^N(s^N)\) is given by \eqref{eq:def:avwc:z}.
\end{defn}

\begin{defn}\label{def:avwc:rc}
\emph{(Secure achievability of random code over AVWC)} For any given AVWC \((\mathcal{W}, \mathcal{V})\), a non-negative real number \(R\) is said to be achievable by random code, with respect to the maximal decoding error probability and the strong secrecy criterion, if for every \(\epsilon > 0\), there exists a random code code \((F, \Phi)\), such that
\begin{equation}\label{eq:avwc:1}
\frac{1}{N} \log |\mathcal{M}| > R - \epsilon, \lambda(\mathcal{W}, F, \Phi) < \epsilon
\text{ and } \max_{s^N \in \mathcal{S}^N} I(M; Z^N(s^N)|\Phi) < \epsilon
\end{equation}
when \(N\) is sufficiently large, where \(M\) is the source message uniformly distributed over the message set \(\mathcal{M}\) and \(Z^N(s^N)\) is given by \eqref{eq:def:avwc:z}.
\end{defn}

The definition on secrecy of random code, namely \(\max_{s^N \in \mathcal{S}^N} I(M; Z^N(s^N)|\Phi) < \epsilon\), first introduced in \cite{Boch-2013}, is a bit different from that of stochastic code. To have a further explanation on the secrecy criterion, notice that
\[
I(M; Z^N(s^N)|\Phi) = I(M; Z^N(s^N)|\Phi) + I(M; \Phi) = I(M; Z^N(s^N),\Phi).
\]
Therefore, the requirement of \(I(M; Z^N(s^N)|\Phi) < \epsilon\) implies that the eavesdropper is almost ignorant of the source message, even if the full information of the random decoder is informed.
It is reasonable to consider the case where the information of the random decoder is exposed to the eavesdropper.
According to the definition of random code, the encoder knows exactly which deterministic decoder is used before each transmission. In consequence, when \(H(\Phi)\) is large enough, the specification of the random decoder \(\Phi\) can serve as the secret key for secure transmission. To be particular, we have Proposition \ref{prop1} as follows.

\begin{prop}\label{prop1}
For any given AVWC \((\mathcal{W}, \mathcal{V})\), real numbers \(R < \max_{X}\min_{q\in\mathcal{P}(\mathcal{S})} I(X;Y_q)\) and \(\epsilon > 0\), there exists a random code \((F, \Phi)\)such that
\[
\frac{1}{N} \log |\mathcal{M}| > R - \epsilon, \lambda(\mathcal{W}, F, \Phi) < \epsilon
\text{ and } \max_{s^N \in \mathcal{S}^N} I(M; Z^N(s^N)) = 0
\]
when \(N\) is sufficiently large, where \(M\) is the source message uniformly distributed over the message set \(\mathcal{M}\) and \(Z^N(s^N)\) is given by \eqref{eq:def:avwc:z}.
\end{prop}

\begin{proof}
It is a direct consequence of Theorem 1 in \cite{Boch-2016}.
\end{proof}

\begin{rem}
According to the proposition, when the transmission of CR is absolutely secure, the presence of the wiretapper does not affect the secure transmission, and the capacity \(\max_{X}\min_{q\in\mathcal{P}(\mathcal{S})} I(X;Y_q)\) of the main AVC \(\mathcal{W}\) can be achieved, no matter what the wiretap AVC \(\mathcal{V}\) is. However, this requires the rate of CR to be sufficiently large. The authors of \cite{Boch-2016} considered a more general case where the rate of CR was limited, and obtained a multi-letter capacity.
\end{rem}

\begin{rem}
When the random decoder \(\Phi\) is deterministic, the random code is specialized as a stochastic code, in which case
\[
I(M; Z^N(s^N)|\Phi)=I(M; Z^N(s^N)).
\]
Therefore, Definition \ref{def:avwc:rc} is consistent with Definition \ref{def:avwc:sc}.
\end{rem}

\textbf{Lower bounds on secrecy capacities}

\begin{thm}\label{thm:avwc:rc}
\emph{(Lower bound on the secrecy capacity of random code over AVWC)} Given an AVWC \((\mathcal{W}, \mathcal{V})\), each real number \(R\) satisfying
\begin{equation}\label{eq:thm:avwc:rc}
R \leq \max_{U} [\min_{q\in\mathcal{P}(\mathcal{S})}I(U;Y_q) - \max_{s\in\mathcal{S}}I(U; Z_s)]
\end{equation}
is achievable by random code with respect to the maximal decoding error probability and the strong secrecy criterion, where \(U\) is an auxiliary random variable satisfying
\begin{equation}\label{eq:thm:avwc:rc:1}
\operatorname{Pr}\{U=u,Z_s = z\}
= \sum_{x\in\mathcal{X}}\operatorname{Pr}\{U=u, X=x\}V_s(z|x)
\end{equation}
and
\begin{equation}\label{eq:thm:avwc:rc:2}
\Pr\{U=u, Y_q=y\}=\sum_{x\in\mathcal{X}}\Pr\{U=u,X=x\}W_q(y|x)
\end{equation}
with \(W_q\) given by \eqref{eq:pre:4}. It suffices to let \(U\) be distributed over an alphabet \(\mathcal{U}\), whose size satisfies \(|\mathcal{U}| \leq |\mathcal{X}|\).
\end{thm}

The proof of Theorem \ref{thm:avwc:rc} is given in Subsection \ref{sec:thm:avwc:rc}.

\begin{rem}
When \(X\) is fixed, the mutual information \(I(X;Y)\) is a convex function of the conditional probability mass function \(P_{Y|X}\). Therefore, it follows that
\(\max_{s\in\mathcal{S}}I(U; Z_s) = \max_{q\in\mathcal{P}(\mathcal{S})}I(U; Z_q)\) for every fixed \(U\), and Formula \eqref{eq:thm:avwc:rc} can be rewritten as
\[
R \leq \max_{U} [\min_{q\in\mathcal{P}(\mathcal{S})}I(U;Y_q) - \max_{q'\in\mathcal{P}(\mathcal{S})}I(U; Z_{q'})],
\]
which coincides with the results in \cite{Bjelakovic-2013}.
\end{rem}

To establish the lower bound on secrecy capacity of stochastic code, we need the following proposition.

\begin{prop}\label{prop:avwc:rc:sc}
(Theorem 2 in \cite{Bjelakovic-2013}) For any AVWC \((\mathcal{W}, \mathcal{V})\), its secrecy capacity of stochastic code and that of random code are identical, if the capacity of stochastic code over the main AVC \(\mathcal{W}\) is positive.
\end{prop}

Combining Theorem \ref{thm:avwc:rc} and Proposition \ref{prop:avwc:rc:sc} immediately yields the following theorem.

\begin{thm}\label{thm:avwc:sc}
\emph{(Lower bound on the secrecy capacity of stochastic code over AVWC)} Given an AVWC \((\mathcal{W}, \mathcal{V})\), every real number \(R\) satisfying
\[
R \leq \max_{U} [\min_{q\in\mathcal{P}(\mathcal{S})}I(U;Y_q) - \max_{s\in\mathcal{S}}I(U; Z_s)] = \max_{U} [\min_{q\in\mathcal{P}(\mathcal{S})}I(U;Y_q) - \max_{q'\in\mathcal{P}(\mathcal{S})}I(U; Z_{q'})]
\]
is achievable by the stochastic code with respect to the maximal decoding error probability and the strong secrecy criterion, if the capacity of stochastic code over the main AVC \(\mathcal{W}\) is positive. The random variable \(U\) has the same features as that introduced in Theorem \ref{thm:avwc:rc}.
\end{thm}

\subsection{AVWC with channel state sequence known at the receiver}\label{sec:main:avwc:csr}\label{sec:avwc:csr}

This subsection discusses the secrecy capacity of the AVWC-CSR with respect to the strong secrecy criterion. Proposition \ref{prop:avwc:csr:rc:sc} claims that the secrecy capacity of stochastic code is identical to that of random code. A lower bound on the secrecy capacity is given in Theorem \ref{thm:avwc:csr:rc}.

\begin{figure}
     \centering
     \includegraphics[width=12cm]{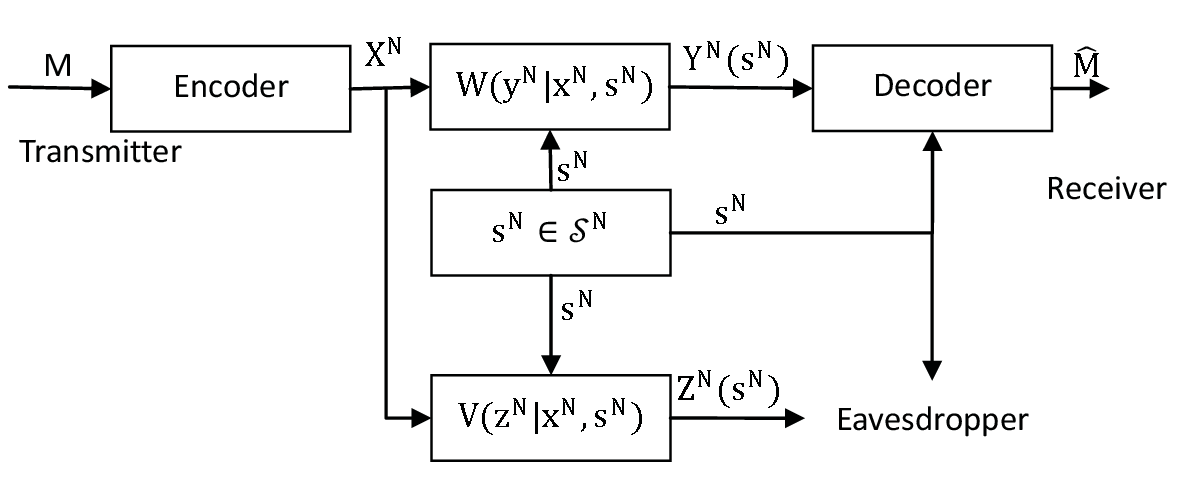}
     \caption{AVWC-CSR, Arbitrarily varying wiretap channel with channel state sequence known at the receiver.}\label{Fig:avwc:csr}
 \end{figure}

\textbf{The definitions}

\begin{defn}
\emph{(AVWC-CSR, Arbitrarily Varying Wiretap Channel with channel states known at the receiver)} Just like the AVWC, an AVWC-CSR can also be specified by a pair \((\mathcal{W}, \mathcal{V})\). The only difference is that the receiver has access to the channel state sequence in the case of AVWC-CSR. See Fig. \ref{Fig:avwc:csr}.
\end{defn}

\begin{defn}\label{def:avwc:csr:sc}
\emph{(Stochastic code over AVWC-CSR)} Let an AVWC-CSR \((\mathcal{W}, \mathcal{V})\) be given. A stochastic code over it is specified by a pair \((F, \phi)\), where \(\{F(m): m \in \mathcal{M}\}\) is a collection of random sequences distributed over \(\mathcal{X}^N\) and \(\phi: \mathcal{Y}^N \times \mathcal{S}^N \mapsto \mathcal{M}\) is a deterministic mapping.
\end{defn}

A random code \((F, \Phi)\) over an AVWC-CSR is defined accordingly.

\begin{defn}
\emph{(Achievability of stochastic code over AVWC-CSR)} For any given AVWC-CSR \((\mathcal{W}, \mathcal{V})\), a non-negative real number \(R\) is said to be achievable by stochastic code, with respect to the maximal decoding error probability and the strong secrecy criterion, if for every \(\epsilon > 0\), there exists a stochastic code \((F, \phi)\) over that AVWC-CSR such that
\[
\frac{1}{N} \log |\mathcal{M}| > R - \epsilon, \lambda^{CSR}(\mathcal{W}, F, \phi) < \epsilon
\text{ and } \max_{s^N \in \mathcal{S}^N} I(M; Z^N(s^N)) < \epsilon
\]
when \(N\) is sufficiently large, where \(M\) is the source message uniformly distributed over \(\mathcal{M}\) and \(Z^N(s^N)\) is given by \eqref{eq:def:avwc:z}.
\end{defn}

\begin{defn}
\emph{(Achievability of random code over AVWC-CSR)} For any given AVWC-CSR \((\mathcal{W}, \mathcal{V})\), a non-negative real number \(R\) is said to be achievable by random code, with respect to the maximal decoding error probability and the strong secrecy criterion, if for every \(\epsilon > 0\), there exists a random code \((F, \Phi)\) over that AVWC-CSR such that
\[
\frac{1}{N} \log |\mathcal{M}| > R - \epsilon, \lambda^{CSR}(\mathcal{W}, F, \Phi) < \epsilon
\text{ and } \max_{s^N \in \mathcal{S}^N} I(M; Z^N(s^N)|\Phi) < \epsilon
\]
when \(N\) is sufficiently large, where \(M\) is the source message uniformly distributed over \(\mathcal{M}\) and \(Z^N(s^N)\) is given by \eqref{eq:def:avwc:z}.
\end{defn}

\textbf{A lower bound of the secrecy capacity}

It has been shown in Subsection \ref{sec:avc:csr} that the capacity of stochastic code over an AVC-CSR is identical to that of random code. This directly yields that the secrecy capacity of stochastic code over an AVWC-CSR is also identical to that of random code, as claimed in the following Proposition.

\begin{prop}\label{prop:avwc:csr:rc:sc}
For any AVWC-CSR \((\mathcal{W}, \mathcal{V})\), its capacity of stochastic code and that of random code are identical.
\end{prop}

The proof of Proposition \ref{prop:avwc:csr:rc:sc} is the same as that of Proposition \ref{prop:avwc:rc:sc}, and is hence omitted.

The following theorem gives a lower bound on the capacity of random code or stochastic code.

\begin{thm}\label{thm:avwc:csr:rc}
\emph{(Lower bound on the secrecy capacity of stochastic code over AVWC-CSR)} Given an AVWC-CSR \((\mathcal{W}, \mathcal{V})\), every real number \(R\) satisfying
\begin{equation}\label{eq:thm:avwc:csr:1}
R \leq \max_{U} [\min_{s\in\mathcal{S}}I(U;Y_s) - \max_{s'\in\mathcal{S}}I(U; Z_{s'})].
\end{equation}
is achievable by stochastic code with respect to the maximal decoding error probability and the strong secrecy criterion, where \(U\) is an auxiliary random variable satisfying
\begin{equation}\label{eq:thm:avwc:csr:rc:1}
\operatorname{Pr}\{U=u,Z_{s'} = z\}
= \sum_{x\in\mathcal{X}}\operatorname{Pr}\{U=u, X=x\}V_{s'}(z|x)
\end{equation}
and
\begin{equation}\label{eq:thm:avwc:csr:rc:2}
\Pr\{U=u, Y_s=y\}=\sum_{x\in\mathcal{X}}\Pr\{U=u,X=x\}W_s(y|x).
\end{equation}
It suffices to let \(U\) be distributed over an alphabet \(\mathcal{U}\), whose size satisfies \(|\mathcal{U}| \leq |\mathcal{X}|\).
\end{thm}

The proof of Theorem \ref{thm:avwc:csr:rc} is given in Subsection \ref{sec:thm:avwc:csr:rc}.

\section{Basic Results for Secure Partitions}\label{sec:basic}

This section provides a general partitioning scheme to ensure secure transmission against wiretapping through AVC. To be particular, Lemma \ref{lem:goodc} claims that the ``good'' codebook, introduced in Definition \ref{def:goodc}, can be obtained by a randomly generating scheme. Lemmas \ref{lem:goodp} and \ref{lem:goodp2} construct secure partitions over the ``good'' codebook to ensure secure transmission. Those secure partitions can be readily used to prove the secrecy capacity results of AVWC and AVWC-CSR. See Section \ref{section:proofs} for details.

\textbf{Lemmas on secure partitions}

\begin{defn}\label{def:goodc}
\emph{(``good'' codebook)} Suppose that \(\mathcal{C}=\{x^N(l)\}_{l=1}^{L'}\) is a codebook of size \(L'=2^{NR'}\) for \(R'>0\). Let \(X\) be a random variable distributed over the alphabet \(\mathcal{X}\). We further assume that \(\delta\), \(\eta\) and \(\nu_1\) are real numbers such that Proposition \ref{remarkTX} holds. The codebook \(\mathcal{C}\) is called ``good'' with respect to the random variable \(X\) if it follows that
\[
|\tilde{T}^N(\mathcal{C}, s^N)| > (1-2\cdot2^{-N\nu_1})L' \text{ for all } s^N \in \mathcal{S}^N,
\]
where \(\tilde{T}^N(\mathcal{C}, s^N) = \mathcal{C} \cap \tilde{T}^N[X,s^N]_{\delta,\eta}\) is the collection of typical codewords with respect to the state sequence \(s^N\).
\end{defn}

If a codebook is ``good'', for any given state sequence \(s^N\), it follows that almost all the codewords are typical with respect to \(s^N\). The following lemma claims that a ``good'' codebook can be constructed by a random scheme with a high probability.

\begin{lem}\label{lem:goodc}
\emph{(The existence of ``good'' codebook)} Let \(\mathbf{C}=\{X^N(l)\}_{l=1}^{L'}\) be a random codebook of size \(L'=2^{NR'}\) for \(R'>0\), such that
\begin{equation}\label{eq:lem:goodc}
\operatorname{Pr}\{\mathbf{C} = \mathcal{C}\} =\prod_{l=1}^{L'} \operatorname{Pr}\{X^N(l)=x^N(l)\} = \prod_{l=1}^{L'} \prod_{i=1}^N P_X(x_i(l))
\end{equation}
for any specific codebook \(\mathcal{C}=\{x^N(l)\}_{l=1}^{L'}\), where \(P_X\) represents the probability mass function of the random variable \(X\). Then the probability of \(\mathbf{C}\) being ``good'' with respect to \(X\) is bounded by
\[
\operatorname{Pr}\{\mathbf{C} \text{ is ``good''}\} > 1 - \epsilon_1,
\]
where \(\epsilon_1 \rightarrow 0\) as \(N \rightarrow \infty\).
\end{lem}

The proof of Lemma \ref{lem:goodc} is given in Appendix \ref{app:lem:goodc}.

\begin{lem}\label{lem:goodp}
 \emph{(Secure partition over ``good'' codebook)} Let \(\mathcal{V}=\{V_s: s \in \mathcal{S}\}\) be an AVC, \(X\) be a random variable over \(\mathcal{X}\), and \(Z_{\mathcal{S}}\) be a collection of random variables such that
 \begin{equation}\label{eq:goodp:25}
 \operatorname{Pr}\{Z_s=z|X=x\} = P_{Z_s|X}(z|x) = V_s(z|x)
 \end{equation}
 for \(z \in \mathcal{Z}\), \(x \in \mathcal{X}\) and \(s \in \mathcal{S}\).
 Suppose that a codebook \(\mathcal{C}\), containing \(L'=2^{NR'}\) codewords of length \(N\), is ``good'' with respect to \(X\), where \(R' > 0\) is a constant real number.
 Then for any \(\tau > 0, \epsilon > 0\) and
 \[
 L < L' \cdot 2^{-N[\max_{s\in\mathcal{S}}I(X;Z_s) + \tau]} = 2^{N[R' - \max_{s\in\mathcal{S}}I(X;Z_s) - \tau]},
 \]
 there exists a secure equipartition \(\{\mathcal{C}_m\}_{m=1}^L\) on it such that
\begin{equation}\label{eq:goodp:20}
I(\tilde{M}; Z^N(\mathcal{C}, s^N)) < \epsilon,
\end{equation}
for every \(s^N \in \mathcal{S}^N\), when \(N\) is sufficiently large. The random sequence \(Z^N(\mathcal{C}, s^N)\) is the output of the AVC when the state sequence is \(s^N\) and the channel input is \(X^N({\mathcal{C}})\). The random sequence \(X^N(\mathcal{C})\) is uniformly distributed over the codebook \(\mathcal{C}\). This indicates that
\begin{equation}\label{eq:goodp:26}
\displaystyle\operatorname{Pr}\{X^N({\mathcal{C}})=x^N, Z^N(\mathcal{C}, s^N)=z^N\}= L'^{-1}\cdot \prod_{i=1}^N P_{Z_{s_i}|X}(z_i|x_i)= L'^{-1}\cdot \prod_{i=1}^N V_{s_i}(z|x)
\end{equation}
 for \(z^N \in \mathcal{Z}^N\), \(x^N \in \mathcal{C}\) and \(s^N \in \mathcal{S}^N\), where \(P_{Z_{s_i}|X}\) is given in \eqref{eq:goodp:25}.
The random variable \(\tilde{M}\) is the index of subcode containing \(X^N(\mathcal{C})\), i.e., \(X^N(\mathcal{C}) \in \mathcal{C}_{\tilde{M}}\).
\end{lem}

\begin{rem}
It is clear that the random variable \(\tilde{M}\) is uniformly distributed over \([1:L]\). When applying Lemma \ref{lem:goodp} to a specific coding scheme, \(\tilde{M}\) represents the source message, \(X^N({\mathcal{C}})\) is the channel input and \(Z^N(\mathcal{C}, s^N)\) is the output of the wiretap channel under the state sequence \(s^N\). Therefore, Formula \eqref{eq:goodp:20} indicates that the information exposed to the eavesdropper is vanishing for every state sequence. Notice that the result of Lemma \ref{lem:goodp} is independent of the main AVC. In fact, it will be shown in Section \ref{section:proofs} that the main channel determines the upper bound of \(L'\).
\end{rem}

Lemma \ref{lem:goodp} can be readily extended to the following lemma.

\begin{lem}\label{lem:goodp2}
\emph{(Secure partition over ``good'' codebook with constraint on states)} Let \(\mathcal{V}=\{V_s: s \in \mathcal{S}\}\) be an AVC, \(X\) be a random variable over \(\mathcal{X}\), and \(Z_{\mathcal{S}}\) be a collection of random variables satisfying \eqref{eq:goodp:25}.
 Suppose that a codebook \(\mathcal{C}\), containing \(L'=2^{NR'}\) codewords of length \(N\), is ``good'' with respect to \(X\), where \(R' > 0\) is a constant real number.
 Furthermore, let \(\mathfrak{S}_n, n \in \mathbb{N}\) be a series of sequence collections such that \(\mathfrak{S}_n \subseteq \mathcal{S}^n\).
 If there exists a real number \(R_d\) such that
 \[
 R_d = \limsup_{n\rightarrow\infty} \max_{s^n\in \mathfrak{S}_n} \bar{I}_{P_{s^n}}(X;Z_{\mathcal{S}}),
 \]
 then for any \(\tau > 0, \epsilon > 0\) and
 \[
 L < L' \cdot 2^{-N[R_d + \tau]} = 2^{N[R' - R_d - \tau]},
 \]
 there exists a secure partition \(\{\mathcal{C}_m\}_{m=1}^L\) on it such that
\[
I(\tilde{M}; Z^N(\mathcal{C}, s^N)) < \epsilon,
\]
for every \(s^N \in \mathfrak{S}_N\), when \(N\) is sufficiently large, where \(\tilde{M}\) and \(Z^N(\mathcal{C}, s^N)\) are the same as that introduced in Lemma \ref{lem:goodp}.
\end{lem}

Lemma \ref{lem:goodp2} can be used to deal with secure transmission over the communication model of AVWC with constrained state sequence. See Subsection \ref{sec:avwc:constrained:seq} for details.

The proofs of Lemma \ref{lem:goodp} and Lemma \ref{lem:goodp2} are similar. The proof of Lemma \ref{lem:goodp} is given below. The proof of Lemma \ref{lem:goodp2} is outlined in Appendix \ref{app:lem:goodp2}.

\textbf{Proof of Lemma \ref{lem:goodp}}

The proof is organized as the following three steps.

\begin{itemize}
\item Step 1 proves that
\begin{equation}\label{eq:goodp:1}
\operatorname{Pr}\{(X^N(\mathcal{C}), Z^N(\mathcal{C}, s^N)) \in \tilde{T}^N[XZ_{\mathcal{S}}]_{2\delta, \eta}\} > 1 - 2^{-N\nu_3}
\end{equation}
for every \(s^N \in \mathcal{S}^N\), where \((X^N(\mathcal{C}),Z^N(\mathcal{C}, s^N))\) is a pair of random sequences satisfying \eqref{eq:goodp:26} and \(\nu_3\) is some positive real number.

\item Step 2 proves the existence of a mapping \(f: \mathcal{C} \mapsto [1: L]\) such that
\begin{equation}\label{eq:goodp:5}
I(M_f; Z^N(\mathcal{C}, s^N)) < \epsilon_2
\end{equation}
for every \(s^N \in \mathcal{S}^N\), where \(\epsilon_2 \rightarrow 0\) as \(N \rightarrow \infty\) and \(M_f = f(X^N(\mathcal{C}))\). The proof is based on Lemma \ref{lemma8}.

 \item The mapping $f$ constructed in Step 2 will derive a secure partition on \(\mathcal{C}\). However, that partition is not necessarily equally divided. Step 3 constructs a desired secure equipartition to establish \eqref{eq:goodp:20} with the help of Lemma \ref{lemma9}.
\end{itemize}

\textbf{Proof of Step 1.} On account of the fact that \(X^N(\mathcal{C})\) is uniformly distributed over the ``good'' codebook \(\mathcal{C}\) (see Definition \ref{def:goodc}), it follows that
\[
\operatorname{Pr}\{X^N(\mathcal{C}) \in \tilde{T}^N[X,s^N]_{\delta,\eta}\} > 1 - 2\cdot 2^{-N\nu_1}
\]
for every \(s^N \in \mathcal{S}^N\). Combining Proposition \ref{remarkTXY} (see also Remark \ref{rem:important}), it follows that
\begin{equation}\label{eq:goodp:6}
\operatorname{Pr}\{(X^N(\mathcal{C}), Z^N(\mathcal{C}, s^N)) \in \tilde{T}^N[XZ_{\mathcal{S}}]_{2\delta, \eta}\} > 1 - 2\cdot 2^{-N\nu_1} - 2^{-N\nu_2} > 1-2^{-N\nu_3}
\end{equation}
for every \(s^N \in \mathcal{S}^N\) when \(\nu_3\) is sufficiently small. The proof of Step 1 is completed. \(\hfill\blacksquare\)

\textbf{Proof of Step 2.} The key idea of the proof is based on Csisz\'{a}r's almost independent coloring scheme in \cite{AI_SC}. The proof first realizes the parameters introduced in Lemma \ref{lemma8}, which is \eqref{eq:goodp:11}. Then Lemma \ref{lemma8} claims the existence of a mapping \(f:\mathcal{C}\mapsto[1:L']\) satisfying \eqref{eq:goodp:4} and \eqref{eq:goodp:3}. Formula \eqref{eq:goodp:5} is finally established from \eqref{eq:goodp:3}.

\begin{lem}\label{lemma8}
   (Lemma 3.1 in \cite{CR_C}) Let \(\mathcal{P}\) be a set of probability mass functions on \(\mathcal{A}\). If there exist \(0 < \varepsilon < \frac{1}{9}\) and \(l > 0\) such that
   \begin{equation}\label{eq:goodp:7}
   \begin{array}{c}
   \sum_{a: P(a) > l^{-1}} P(a) \leq \varepsilon
   \end{array}
   \end{equation}
   for all \(P \in \mathcal{P}\), then for any positive integer \(k\) satisfying \(k\log k \leq \frac{\epsilon^2l}{3\log(2|\mathcal{P}|)}\), there exists a function \(f: \mathcal{A} \mapsto [1: k]\) such that
   \[
   \begin{array}{c}
   \sum_{i=1}^k |P(f^{-1}(i)) - \frac{1}{k}| < 3\varepsilon
   \end{array}
   \]
   for all \(P \in \mathcal{P}\).
\end{lem}

To realize the parameters introduced in Lemma \ref{lemma8}, the main job is to construct \(\mathcal{P}\), which is a collection of probability mass functions on the codebook \(\mathcal{C}\) in our proof. The construction depends on a collection of subsets \(\mathcal{B}(\mathcal{C}, s^N)\) of \(\mathcal{Z}^N\) for all \(s^N \in \mathcal{S}^N\), which are defined as
\begin{equation}\label{eq:goodp:16}
\mathcal{B}(\mathcal{C}, s^N) = \mathcal{B}_0(\mathcal{C}, s^N) \setminus \mathcal{B}_1(\mathcal{C}, s^N),
\end{equation}
where
\begin{equation}\label{eq:goodp:12}
\mathcal{B}_0(\mathcal{C}, s^N) = \{z^N \in \tilde{T}^N[Z_{\mathcal{S}},s^N]_{2\delta,\eta}: \Psi(\mathcal{C}, s^N, z^N) < 2^{-N\nu_3/2}\}
\end{equation}
with
\[
\Psi(\mathcal{C}, s^N, z^N) = \operatorname{Pr}\{X^N(\mathcal{C}) \notin \tilde{T}^N[XZ_{\mathcal{S}}|s^N, z^N]_{2\delta, \eta} | Z^N(\mathcal{C}, s^N) = z^N\},
\]
and
\begin{equation}\label{eq:goodp:15}
\mathcal{B}_1(\mathcal{C}, s^N) =  \{z^N \in \mathcal{Z}^N: \operatorname{Pr}\{Z^N(\mathcal{C}, s^N) = z^N\} <  2^{-N\nu_3/2} \prod_{i = 1}^N P_{Z_{s_i}}(z_i)\}.
\end{equation}

Firstly, we claim a useful property of \(\mathcal{B}(\mathcal{C}, s^N)\) in \eqref{eq:goodp:2}. For every \(s^N \in \mathcal{S}^N\), Formulas \eqref{eq:goodp:6} and \eqref{eq:goodp:12} give
\[
\operatorname{Pr}\{Z^N(\mathcal{C}, s^N) \in \mathcal{B}_0(\mathcal{C}, s^N)\} > 1 - 2^{-N\nu_3/2}.
\]
Meanwhile, it follows from \eqref{eq:goodp:15} that
\[
\operatorname{Pr}\{Z^N(\mathcal{C}, s^N) \in \mathcal{B}_1(\mathcal{C}, s^N)\} < 2^{-N\nu_3/2}
\]
for every \(s^N \in \mathcal{S}^N\). Combining the last two formulas above gives
\begin{equation}\label{eq:goodp:2}
\operatorname{Pr}\{Z^N(\mathcal{C}, s^N) \in \mathcal{B}(\mathcal{C}, s^N)\} > 1- 2\cdot 2^{-N\nu_3/2}
\end{equation}
for every \(s^N \in \mathcal{S}^N\).

Secondly, with the help of \(\mathcal{B}(\mathcal{C}, s^N)\), the parameters introduced in Lemma \ref{lemma8} are realized as
\begin{equation}\label{eq:goodp:11}
\begin{array}{c}
\mathcal{A}=\mathcal{C}, \varepsilon= 2^{-N\nu_3/2},\\
 l = 2^{N [R' - \max_{s\in\mathcal{S}}I(X;Z_s) - \tau/2]}, \\
 k = L < 2^{N [R' - \max_{s\in\mathcal{S}}I(X;Z_s) - \tau]},\\
\mathcal{P} = \{P_{s^N, z^N}: s^N \in \mathcal{S}^N, z^N \in \mathcal{B}(\mathcal{C}, s^N)\} \cup \{P_0\},
\end{array}
\end{equation}
where
\begin{equation}\label{eq:achievable:p_0}
P_0(x^N) = \operatorname{Pr}\{X^N(\mathcal{C}) = x^N\}
\end{equation}
and
\begin{equation}\label{eq:achievable:p_zn}
P_{s^N, z^N}(x^N) = \operatorname{Pr}\{X^N(\mathcal{C}) = x^N | Z^N(\mathcal{C}, s^N) = z^N\}.
\end{equation}
The verification that parameters in \eqref{eq:goodp:11} satisfy the preconditions in Lemma \ref{lemma8}, is given in Appendix \ref{app:goodp:11}.
Using Lemma \ref{lemma8} with parameters from \eqref{eq:goodp:11}, there exists \(f : \mathcal{C} \mapsto [1: k]\) satisfying that
\begin{equation}\label{eq:goodp:4}
\sum_{i=1}^k |\operatorname{Pr}\{X^N(\mathcal{C}) \in f^{-1}(i)\} - \frac{1}{k}| = \sum_{i=1}^k |\operatorname{Pr}\{M_f = i\} - \frac{1}{k}| < 3\varepsilon
\end{equation}
and
\begin{equation}\label{eq:goodp:3}
\sum_{i=1}^k |\operatorname{Pr}\{X^N(\mathcal{C}) \in f^{-1}(i) | Z^N(\mathcal{C}, s^N) = z^N\} - \frac{1}{k}| =  \sum_{i=1}^k |\operatorname{Pr}\{M_f = i | Z^N(\mathcal{C}, s^N) = z^N\} - \frac{1}{k}| <  3\varepsilon
\end{equation}
for all \(s^N \in \mathcal{S}^N\) and \(z^N \in \mathcal{B}(\mathcal{C}, s^N)\), where \(M_f = f(X^N(\mathcal{C}))\).

Finally, recall that \(L = k\). Therefore, Formula \eqref{eq:goodp:3} and the uniform continuity of entropy (see Lemma 1.2.7 in \cite{CT_DMS}) gives
\[
\begin{array}{l}
H(M_f|Z^N(\mathcal{C}, s^N) = z^N) \geq \log L - 3\varepsilon \log \frac{L}{3\varepsilon}
\end{array}
\]
for all \(s^N \in \mathcal{S}^N\) and \(z^N \in \mathcal{B}(\mathcal{C}, s^N)\). Combining \eqref{eq:goodp:2} and the formula above, it is satisfied that
\[
\begin{array}{lll}
H(M_f|Z^N(\mathcal{C}, s^N)) > (1-2\varepsilon)(\log L - 3\varepsilon \log \frac{L}{3\varepsilon})
\end{array}
\]
for every \(s^N \in \mathcal{S}^N\). Since \(H(M_f) \leq \log L\), we arrive at
\[
I(M_f; Z^N(\mathcal{C}, s^N)) < 5\varepsilon \log L - 3\varepsilon \log (3\varepsilon)
\]
for every \(s^N \in \mathcal{S}^N\). Formula \eqref{eq:goodp:5} is established by substituting \(\varepsilon=2^{-N\nu_3/4}\) and \(L < 2^{N[R' - \max_{s\in\mathcal{S}}I(X;Z_s) - \tau]}\) into the formula above. The proof of Step 2 is completed. \(\hfill\blacksquare\)

\textbf{Proof of Step 3.} The proof depends on the following lemma.
\begin{lem}\label{lemma9}
(Lemma 4 in \cite{dan-entropy}) For any given codebook \(\mathcal{C}\), if the function \(f: \mathcal{C} \mapsto [1: L]\) satisfies \eqref{eq:goodp:4}, then there exists a partition \(\{\mathcal{C}_{m}\}_{m=1}^{L}\) on \(\mathcal{C}\) such that
\begin{enumerate}
  \item \(|\mathcal{C}_{m}| = \frac{L'}{L}\) for all \(m \in [1: L]\),
  \item \(H(\tilde{M}|M_f) < 4 \sqrt{\varepsilon} \log L\),
\end{enumerate}
where \(\tilde{M}\) is the index of the bin containing \(X^N(\mathcal{C})\), i.e., \(X^N(\mathcal{C}) \in \mathcal{C}_{\tilde{M}}\).
\end{lem}

From Lemma \ref{lemma9} and Formula \eqref{eq:goodp:5}, we have
\begin{equation}\label{eq:goodp:21}
\begin{array}{lll}
I(\tilde{M}; Z^N(\mathcal{C}, s^N)) &\leq& I(\tilde{M}, M_f; Z^N(\mathcal{C}, s^N)) \\
& =& I(M_f; Z^N(\mathcal{C}, s^N)) + I(\tilde{M}; Z^N(\mathcal{C}, s^N) | M_f)\\
& \leq& I(M_f; Z^N(\mathcal{C}, s^N)) + H(\tilde{M} | M_f) \leq \epsilon_2 + 4 \sqrt{\varepsilon} \log L
\end{array}
\end{equation}
for all \(s^N \in \mathcal{S}^N\).
Since \(\varepsilon=2^{-N\nu_3/4}\) and \(L < 2^{N[R' - \max_{s\in\mathcal{S}}I(X;Z_s) - \tau]}\), it follows that \(4 \sqrt{\varepsilon} \log L \rightarrow 0\) as \(N \rightarrow \infty\). Therefore, the rightmost side of Formula \eqref{eq:goodp:21} is vanishing, accomplishing the proof of Step 3.

The proof of Lemma \ref{lem:goodp} is completed. \(\hfill\blacksquare\)

\section{Proofs of the Main Theorems}\label{section:proofs}

\subsection{Proof of Theorem \ref{thm:avc:csr}}\label{sec:thm:avc:csr}

This subsection proves the direct part of Theorem \ref{thm:avc:csr}. In particular, it suffices to establish the capacity of deterministic code over the AVC-CSR \(\mathcal{W}\) with respect to the average decoding error probability. The other results can be proved by standard technique in \cite{Ahlswede-1978}.
More precisely, we will prove that for any given random variables \((X, Y_{\mathcal{S}})\) satisfying \eqref{eq:measure:1}, and any real numbers \(0 < \tau < \min_{s\in\mathcal{S}}I(X; Y_s)\) and \(\epsilon > 0\), there exists a deterministic code \((f, \phi)\) over that AVC-CSR such that
\[
\frac{1}{N} \log |\mathcal{M}| > \min_{s\in\mathcal{S}}I(X; Y_s) - \tau \text{ and } \bar{\lambda}^{CSR}(\mathcal{W}, f, \phi) < \epsilon
\]
when the block length \(N\) is sufficiently large, where \(f: \mathcal{M} \mapsto \mathcal{X}^N\) and \(\phi: \mathcal{Y}^N \times \mathcal{S}^N \mapsto  \mathcal{M}\).

We first introduce the coding scheme over the AVC-CSR, and then Lemma \ref{lem:avc:dc} claims that the coding scheme is effective.

The coding scheme is designed as follows.
\begin{itemize}
\item \textbf{Codebook Genaration.} Let random variables \((X, Y_{\mathcal{S}})\) and positive real number \(\tau\) be given. Denote by \(\mathbf{C} = \{X^N(m)\}_{m=1}^L\) a random codebook satisfying \eqref{eq:lem:goodc} with \(L'=L\), where \(2^{N[\min_{s\in\mathcal{S}}I(X; Y_s) - \tau]} < L < 2^{N[\min_{s\in\mathcal{S}}I(X; Y_s) - \tau/2]}\). The final deterministic codebook used by the encoder is a sample value generated by \(\mathbf{C}\).

\item \textbf{Encoder.} Let \(\mathcal{M} = [1:L]\) be the message set, and the random variable \(M\) be the source message uniformly distributed over \(\mathcal{M}\). Suppose that the codebook \(\mathcal{C} = \{x^N(m)\}_{m=1}^L\) is specified. The source message \(M\) is then encoded as \(f(M) = x^N(M)\).

\item \textbf{Decoder.} The decoding scheme is constructed by an iteration scheme. Suppose that the codebook \(\mathcal{C} = \{x^N(m)\}_{m=1}^L\) is specified. Choose a sufficiently small real number \(\nu_2\) with \(\nu_2 < \tau/8\) such that Proposition \ref{remarkTXY} holds. For any given \(s^N \in \mathcal{S}^N\), the decoding process is manipulated as follows.

Denote \(\tilde{\mathcal{D}}_0(s^N) = \mathcal{D}_0(s^N) = \emptyset\).
For \(1 \leq m \leq L\), let \[\mathcal{D}_m(s^N) = \tilde{T}^N[XY_{\mathcal{S}}, s^N| x^N(m)]_{2\delta,\eta}\backslash\tilde{\mathcal{D}}_{m-1}(s^N),\] if it is satisfied that
\begin{equation}\label{eq:avc:csr:1}
x^N(m) \in \tilde{T}^N[X,s^N]_{\delta,\eta}
\end{equation}
and
\begin{equation}\label{eq:avc:csr:2}
W(\tilde{T}^N[XY_{\mathcal{S}}, s^N| x^N(m)]_{2\delta,\eta}/\tilde{\mathcal{D}}_{m-1}(s^N)|x^N(l), s^N) > 1-2\cdot 2^{-N{\nu_2}}.
\end{equation}
Otherwise, set \(\mathcal{D}_m(s^N) = \emptyset\). After the value of \(\mathcal{D}_m(s^N)\) is determined,
let \(\tilde{\mathcal{D}}_m(s^N) = \tilde{\mathcal{D}}_{m-1}(s^N) \cup \mathcal{D}_m(s^N)\).

 Suppose the state sequence \(s^N\) is given. The received sequence \(y^N\) is decoded as \(\phi(y^N,s^N) = \hat{m}\) if \(y^n \in \mathcal{D}_{\hat{m}}(s^N)\). If \(y^N\) is out of \(\tilde{\mathcal{D}}_L(s^N)\), declare a decoding error.
\end{itemize}

\begin{rem}
The coding scheme above is quite similar to that of DMCs except for the decoder. Here are some notes for a further explanation. For each source message $m$, $\mathcal{D}_m(s^N)$ is the decoding set of $m$, and $\tilde{\mathcal{D}}_m(s^N)$ represents the union of the decoding sets of messages $1$ to $m$. The decoding sets are related to $s^N$ since the receiver knows the state sequence. Further, when $s^N$ is fixed, the decoding sets of all the messages are disjoint. According to Formula \eqref{eq:avc:csr:2}, when $\mathcal{D}_m(s^N)$ is non-empty, the decoding error probability of message $m$ under the state sequence $s^N$ must be less than $2\cdot 2^{-N{\nu_2}}$. Therefore, the decoding error probability of a source message is either less than $2\cdot 2^{-N{\nu_2}}$ or exactly 1. Meanwhile, on account of \eqref{eq:avc:csr:2}, $\mathcal{D}_m(s^N)$ is non-empty only if the related codeword $x^N(m)$ is typical under the state sequence $s^N$. Finally, for a fixed message $m$, it is possible that its decoding error probability is small under a certain state sequence, and is 1 under another.
\end{rem}

The following lemma claims that the coding scheme above is effective.

 \begin{lem}\label{lem:avc:dc}
 For any given codebook \(\mathcal{C} = \{x^N(m)\}_{m=1}^L\), let \((f, \phi)\) be the pair of encoder and decoder produced by the coding scheme above. Denote by
 \begin{equation}\label{eq:lem:avc:dc}
 \bar{e}^{CSR}(\mathcal{C}, s^N) = \bar{e}^{CSR}(f,\phi, s^N) = \frac{1}{L}\sum_{m=1}^{L} [1-W(\phi^{-1}(m, s^N)|f(m),s^N)]
 \end{equation}
 the average decoding error probability of the coding scheme under the state sequence $s^N$, where \(f(m)=x^N(m)\) and \(\phi^{-1}(m, s^N) = \mathcal{D}_m(s^N)\) according to the coding scheme above. Then for any \(\epsilon > 0\), it follows that
 \begin{equation}\label{eq:thm:avc:dc}
 \operatorname{Pr}\Bigg\{ \max_{s^N \in \mathcal{S}^N} \bar{e}^{CSR}(\mathbf{C}, s^N) > \epsilon \Bigg\} < \epsilon_3
 \end{equation}
 where \(\mathbf{C}\) is the random codebook satisfying \eqref{eq:lem:goodc} with \(L'=L\), and \(\epsilon_3 \rightarrow 0\) as \(N \rightarrow \infty\).
 \end{lem}

 The proof of Lemma \ref{lem:avc:dc}, similar to that of Theorem 5 in \cite{Ahlswede-1982}, is given in Appendix \ref{app:lem:avc:dc}.

\subsection{Proof of Theorem \ref{thm:avwc:rc}}\label{sec:thm:avwc:rc}

In this subsection, we provide the proof of Theorem \ref{thm:avwc:rc}, which claims a lower bound on the secrecy capacity of random code over the AVWC \((\mathcal{W}, \mathcal{V})\). Let the collection of random variables \(\{(X, Y_q, Z_s): q\in\mathcal{P}(\mathcal{S}) \text{ and } s\in\mathcal{S}\}\) satisfy
\begin{equation}\label{eq:proof:thm:avwc:rc:6:y}
\operatorname{Pr}\{X=x,Y_q=y\} = P_X(x)W_q(y|x)
\end{equation}
and
\begin{equation}\label{eq:proof:thm:avwc:rc:6:z}
\operatorname{Pr}\{X=x,Z_s=z\} = P_X(x)V_s(z|x)
\end{equation}
with \(W_q\) given by \eqref{eq:pre:4}. It suffices to prove that for any real numbers \(0 < \tau < \min_{q\in\mathcal{P}(\mathcal{S})}(X; Y_{q}) - \max_{s\in\mathcal{S}}I(X; Z_{s})\) and \(\epsilon > 0\), there exists a pair of random encoder and decoder \((F, \Phi)\) over that AVWC \((\mathcal{W}, \mathcal{V})\) such that
\begin{equation}\label{eq:proof:thm:avwc:rc:1}
\frac{1}{N} \log |\mathcal{M}| > \min_{q\in\mathcal{P}(\mathcal{S})}(X; Y_{q}) - \max_{s\in\mathcal{S}}I(X; Z_{s}) - \tau, \lambda(\mathcal{W}, F, \Phi) < \epsilon
\text{ and } \max_{s^N \in \mathcal{S}^N} I(M; Z^N(s^N)|\Phi) < \epsilon,
\end{equation}
where \(M\) is the source message uniformly distributed over the message set \(\mathcal{M}\), and the random sequence \(Z^{N}(s^{N})\), which is defined in \eqref{eq:def:avwc:z}, is the output of the wiretap AVC under the state sequence \(s^{N}\).

In the remainder of this subsection, we first introduce two preliminary lemmas, and then present the formal proof.

\textbf{The preliminary lemmas.}

 The two lemmas introduced here are briefly described as follows. Lemma \ref{lem:avc:rc} claims the capacity of random code over the AVC. However, knowing the capacity result is not enough for the proof of Theorem \ref{thm:avwc:rc}. The more important thing is that the capacity can be achieved by a random coding method, detailed in Remark \ref{rem:avc:rc}. Lemma \ref{lem:avc:ele} comes from the well-known elimination technique introduced by Ahlswede in \cite{Ahlswede-1978}, which was a key tool constructing a stochastic code from a random code. The main conclusion is that it suffices to let the random code be distributed over a collection of deterministic codes, whose size is quadratic to the codeword length.

\begin{lem}\label{lem:avc:rc}
(Lemma 2.6.10 in \cite{CT_DMS}) Let an AVC \(\mathcal{W}\), a collection of random variables \((X, Y_{\mathcal{S}})\) satisfying Formula \eqref{eq:measure:1}, and real numbers \(0 < \tau' < \min_{q\in\mathcal{P}(\mathcal{S})}(X; Y_{q})\) and \(\epsilon' > 0\) be given. There exists a pair of random encoder and decoder \((F, \Phi)\) distributed over a certain family of deterministic encoder-decoder pairs \(\{(f_g, \phi_g): g \in \mathcal{G}\}\) with \(f_g: \mathcal{M}' \mapsto \mathcal{X}^N\) and \(\phi_g: \mathcal{Y}^N \mapsto \mathcal{M}'\) such that
\[\frac{1}{N} \log |\mathcal{M}'| > \min_{q\in\mathcal{P}(\mathcal{S})}(X; Y_{q}) - \tau' \text{ and } \lambda(\mathcal{W}, F, \Phi) < \epsilon'\]
when \(N\) is sufficiently large.
\end{lem}

\begin{rem}\label{rem:avc:rc}
There exist several ways proving Lemma \ref{lem:avc:rc}. In particular, the random codebook \(\mathbf{C} = \{F(l)=X^N(l)\}_{l=1}^{L'}\) constructed in the proof of \cite{CT_DMS} satisfies \eqref{eq:lem:goodc} with \(L'=|\mathcal{M}'|\). This indicates that the codewords are generated independently based on the probability mass function \(P_X\) of the random variable \(X\). In consequence, Lemma \ref{lem:goodc} claims that the random codebook \(\mathbf{C}\) is ``good'' with high probability.
\end{rem}

\begin{lem}\label{lem:avc:ele}
(\emph{Elimination technique}) Suppose that \((F, \Phi)\) is a pair of random encoder and decoder of length \(N\) over an AVC (resp. AVC-CSR) \(\mathcal{W}\), the size of whose message set is \(|\mathcal{M}'| = 2^{NR'}\) for some constant real number \(R' > 0\), such that
\[
\lambda(\mathcal{W}, F, \Phi) < \epsilon' \text{ resp. } \lambda^{CSR}(\mathcal{W}, F, \Phi) < \epsilon'
\]
for some \(\epsilon' > 0\). Let \((F_i, \Phi_i), 1 \leq i \leq N^2\) be a series of i.i.d. random encoder-decoder pairs with the same probability mass function as that of \((F, \Phi)\). Then for any \(\epsilon > 4\epsilon'\), it follows that
\[
\operatorname{Pr}\Bigg\{\max_{m\in \mathcal{M}}\max_{s^N \in \mathcal{S}^N} \Bigg[ \sum_{k=1}^{N^2} \frac{1}{N^2} e_m(\mathcal{W}, F_k, \Phi_k, s^N) \Bigg] < \epsilon/2 \Bigg\}
> 1 - \epsilon_4
\]
resp.
\[
\operatorname{Pr}\Bigg\{\max_{m\in \mathcal{M}}\max_{s^N \in \mathcal{S}^N} \Bigg[ \sum_{k=1}^{N^2} \frac{1}{N^2} e_m^{CSR}(\mathcal{W}, F_k, \Phi_k, s^N) \Bigg] < \epsilon/2 \Bigg\}
> 1 - \epsilon_4,
\]
where \(\epsilon_4 \rightarrow 0\) as \(N \rightarrow \infty\).
\end{lem}

\begin{proof}
The result of AVWC was given in \cite{Ahlswede-1978}. The result of AVWC-CSR can be proved similarly.
\end{proof}

\begin{rem}\label{rem:avc:ele}
  Combining Lemma \ref{lem:avc:rc}, Remark \ref{rem:avc:rc} and Lemma \ref{lem:avc:ele}, it is concluded that there exists a series of deterministic encoder-decoder pairs \((f_k, \phi_k), 1 \leq k \leq N^2\) such that
 \[
 \max_{m\in \mathcal{M}}\max_{s^N \in \mathcal{S}^N} \Bigg[ \sum_{k=1}^{N^2} \frac{1}{N^2} \lambda_m(\mathcal{W}, f_k, \phi_k, s^N) \Bigg] = \max_{m\in \mathcal{M}}\max_{s^N \in \mathcal{S}^N} \Bigg[ \sum_{k=1}^{N^2} \frac{1}{N^2} e_m(\mathcal{W}, f_k, \phi_k, s^N) \Bigg] < \epsilon/2.
 \]
Let \((\bar{F}, \bar{\Phi})\) be a pair of random encoder and decoder uniformly distributed over \(\{(f_k, \phi_k): 1 \leq k \leq N^2\}\). It follows clearly that
\[
\lambda(\mathcal{W}, \bar{F}, \bar{\Phi}) < \epsilon/2.
\]
The CR in this new random code is a random variable uniformly distributed over the index set \([1: N^2]\). Therefore, the rate of the CR is \(\frac{1}{N} \log N^2 \rightarrow 0\) as \(N \rightarrow \infty\). This indicates that the rate of CR can be dramatically smaller than the rate of the message.
\end{rem}

\textbf{The formal proof.}

The proof is divided into the following steps.
\begin{itemize}
  \item Let \(N\) be a sufficiently large even integer and \(\tau'\) be a sufficiently small positive real number. Step 1 constructs a series of deterministic codes \((f_k, \phi_k), 1 \leq k \leq N^2/2\), such that the codebooks \(\mathcal{C}_k = \{f_k(m): 1\leq m \leq L'\}\), \(1 \leq k \leq N^2/2\) are ``good'' with respect to \(X\) and
      \begin{equation}\label{eq:proof:thm:avwc:rc:2}
      \max_{m\in \mathcal{M}}\max_{s^N \in \mathcal{S}^N} \Bigg[ \sum_{k=1}^{N^2/2} \frac{2}{N^2} \lambda_m(\mathcal{W}, f_k, \phi_k, s^N) \Bigg] < \epsilon,
      \end{equation}
      where \(f_k: \mathcal{M}' \mapsto \mathcal{X}^N\), \(\phi_k: \mathcal{Y}^N \mapsto \mathcal{M}'\), and the size \(L'\) of the message set \(\mathcal{M}'\) satisfies that \(L' > 2^{N[\min_{q\in\mathcal{P}(\mathcal{S})}I(X; Y_{q}) - \tau']}\).
  \item Step 2 constructs the random encoder and decoder \((\bar{F}, \bar{\Phi})\) satisfying \eqref{eq:proof:thm:avwc:rc:1}.
\end{itemize}

\textbf{Proof of Step 1.} On account of Lemma \ref{lem:avc:rc}, for any \(\tau' > 0\), there exists a pair of random encoder and decoder \((F, \Phi)\) distributed over a certain family of deterministic codes \(\{(f_g, \phi_g): g \in \mathcal{G}\}\) with \(f_g: \mathcal{M}' \mapsto \mathcal{X}^N\) and \(\phi_g: \mathcal{Y}^N \mapsto \mathcal{M}'\) such that \(\lambda(\mathcal{W}, F, \Phi) < \epsilon'\), where \(\epsilon'=\epsilon/5\) and the size \(L'\) of the message set \(\mathcal{M}'\) satisfies \(L' > 2^{N[\min_{q\in\mathcal{P}(\mathcal{S})}I(X; Y_{q}) - \tau']}\). Moreover, according to Remark \ref{rem:avc:rc}, the random codebook \(\mathbf{C}=\{F(m): 1\leq m \leq L'\}\) satisfies \eqref{eq:lem:goodc}.

Let \((F_k, \Phi_k), 1 \leq k \leq N^2\) be a series of i.i.d. random codes with the same probability mass function as that of \((F, \Phi)\). Denote \(\mathbf{C}_k=\{F_k(m): 1\leq m \leq L'\}\) and
\[
\psi(\mathbf{C}_k) =
\begin{cases}
0 & \text{ if } \mathbf{C}_k \text{ is ``good'' }, \cr
1 & \text{ otherwise.}
\end{cases}
\]
It follows from Lemma \ref{lem:goodc} that
\[
\operatorname{Pr}\{\sum_{k=1}^{N^2} \psi(\mathbf{C}_k) > N^2/2\} < 2\epsilon_1.
\]
Combining Lemma \ref{lem:avc:ele}, there exists a series of deterministic codes \((f_k, \phi_k), 1 \leq k \leq N^2\) such that
\[
 \max_{m\in \mathcal{M}}\max_{s^N \in \mathcal{S}^N} \Bigg[ \sum_{k=1}^{N^2} \frac{1}{N^2} \lambda_m(\mathcal{W}, f_k, \phi_k, s^N) \Bigg] < \epsilon/2
 \]
 and
 \[
  \sum_{k=1}^{N^2} \psi(\mathcal{C}_k) \leq N^2/2,
 \]
 where \(f_k: \mathcal{M}' \mapsto \mathcal{X}^N\), \(\phi_k: \mathcal{Y}^N \mapsto \mathcal{M}'\) and \(\mathcal{C}_k =\{f_k(m):m\in\mathcal{M}'\}\).

 Let \(\psi(\mathcal{C}_k) = 0\) for \(1\leq k \leq N^2/2\) without loss of generality, i.e., the first \(N^2/2\) (deterministic) codebooks are ``good''. It follows clearly that those codebooks satisfy Formula \eqref{eq:proof:thm:avwc:rc:2}, completing the proof of Step 1. \(\hfill\blacksquare\)

 \textbf{Proof of Step 2.} For every ``good'' codebook \(\mathcal{C}_k, 1 \leq k \leq N^2/2\), Lemma \ref{lem:goodp} claims that if the value of \(L\) satisfies that
 \begin{equation}\label{eq:proof:thm:avwc:rc:5}
 2^{N[\min_{p\in\mathcal{P}(\mathcal{S})}I(X;Y_p)-\max_{\mathcal{S}}I(X;Z_s) - 3\tau']} < L'\cdot 2^{-N[\max_{\mathcal{S}}I(X;Z_s) + 2\tau']} < L < L'\cdot2^{-N[\max_{\mathcal{S}}I(X;Z_s) + \tau']},
 \end{equation}
  then there exists a secure partition \(\{\mathcal{C}_{k, m}\}_{m=1}^L\) on it such that
  \begin{equation}\label{eq:proof:thm:avwc:rc:3}
  \max_{s^N \in \mathcal{S}^N} I(\tilde{M}_k; Z^N(\mathcal{C}_k, s^N)) < \epsilon
  \end{equation}
  where \(\tilde{M}_k\) is the index of subcode containing the random sequence \(X^N(\mathcal{C}_k)\) and \((X^N(\mathcal{C}_k),Z^N(\mathcal{C}_k, s^N))\) is a pair of random sequences satisfying \eqref{eq:goodp:26} with \(\mathcal{C}=\mathcal{C}_k\).

Now the random code can be defined as follows. For every given \(k\) and \(m\), denote \(\mathcal{C}_{k, m} = \{x^N(k, m, j): 1 \leq j \leq L'/L\}\). Let \(K\) and \(J\) be random variables uniformly distributed over \([1:N^2/2]\) and \([1: L'/L]\), respectively. Moreover, \(K, J\) and the source message \(M\) are mutually independent. The random encoder is defined as \(\bar{F}(M)=x^N(K, M, J)\) and the random decoder is \(\bar{\Phi} = \phi_{K}\). It follows clearly that the joint probability mass function of the source message \(M\) and the wiretap AVC output \(Z^N(s^N)\) is identical to that of \(\tilde{M}_k\) and \(Z^N(\mathcal{C}_k, s^N)\) when fixing \(K=k\). Therefore, for every \(s^N \in \mathcal{S}^N\) we have
 \begin{equation}\label{eq:proof:thm:avwc:rc:4}
 \begin{array}{lll}
 I(M; Z^N(s^N)|\bar{\Phi}) &\leq&  I(M; Z^N(s^N) | K)\\
 &=& \displaystyle\frac{2}{N^2}\sum_{k=1}^{N^2/2}I(M; Z^N(s^N) | K=k)\\
 &=& \displaystyle\frac{2}{N^2}\sum_{k=1}^{N^2/2}I(\tilde{M}_k; Z^N(\mathcal{C}, s^N_k))\\
 &<& \epsilon,
 \end{array}
 \end{equation}
 where the first equality follows because \(M\) is independent of \(K\) and \(\bar{\Phi}\) while \(\bar{\Phi}\) is a function of \(K\), and the last inequality follows from Formula \eqref{eq:proof:thm:avwc:rc:3}.

Furthermore, Formula \eqref{eq:proof:thm:avwc:rc:2} gives
\begin{equation}\label{eq:proof:thm:avwc:rc:4:2}
\lambda(\mathcal{W}, \bar{F}, \bar{\Phi}) = \max_{m\in \mathcal{M}}\max_{s^N \in \mathcal{S}^N} \Bigg[ \sum_{k=1}^{N^2/2} \frac{2}{N^2} \lambda_m(\mathcal{W}, f_k, \phi_k, s^N) \Bigg] < \epsilon,
\end{equation}
and Formula \eqref{eq:proof:thm:avwc:rc:5} gives
\begin{equation}\label{eq:proof:thm:avwc:rc:4:2}
\frac{1}{N}\log|\mathcal{M}| > \min_{q\in\mathcal{P}(\mathcal{S})}I(X; Y_{q})- \max_{z\in\mathcal{S}}I(X; Z_{s}) - \tau
\end{equation}
by setting \(\tau' < \tau/3\). Formula \eqref{eq:proof:thm:avwc:rc:1} is established by combining \eqref{eq:proof:thm:avwc:rc:4}-\eqref{eq:proof:thm:avwc:rc:4:2}, The Theorem is proved.\(\hfill\blacksquare\)

\subsection{Proof of Theorem \ref{thm:avwc:csr:rc}}\label{sec:thm:avwc:csr:rc}

After establishing the following lemma, the omitted proof of Theorem \ref{thm:avwc:csr:rc} is quite similar to that of Theorem \ref{thm:avwc:rc} introduced in the previous subsection, see Remark \ref{rem:avwc:csr} for more details.

\begin{lem}\label{lem:avwc:csr:rc}
Let an AVC-CSR \(\mathcal{W}\), a collection of random variables \((X, Y_{\mathcal{S}})\) satisfying \eqref{eq:avc:1}, and two real numbers \(0 < \tau' < \min_{s\in\mathcal{S}}I(X; Y_s)\) and \(\epsilon > 0\) be given. When the codeword length \(N\) is sufficiently large, there exists a series of deterministic codebooks \(\mathcal{C}_k, 1\leq k \leq N^2\) having the following properties.
\begin{itemize}
  \item all the codebooks are ``good'' with respect to \(X\);
  \item the size \(L'\) of each codebook satisfies \(L' > 2^{N(\min_{s\in\mathcal{S}}I(X; Y_s)-\tau')}\);
  \item Let \((f_k,\phi_k)\) be the pair of encoder and decoder derived from \(\mathcal{C}_k\) according to the coding scheme in Subsection \ref{sec:thm:avc:csr}. The decoding error probability satisfies
      \begin{equation}\label{eq:sec:thm:avwc:csr:rc}
      \max_{m\in \mathcal{M'}}\max_{s^N \in \mathcal{S}^N} \Bigg[ \sum_{k=1}^{N^2} \frac{1}{N^2} e_m^{CSR}(\mathcal{W}, f_k, \phi_k, s^N) \Bigg] < \epsilon/2,
      \end{equation}
      where \(\mathcal{M}'\) is the message set.
\end{itemize}
\end{lem}

The proof below shows an important technique constructing a random code with a vanishing maximal decoding error from an arbitrary code with a vanishing average decoding error.

\textbf{Proof of Lemma \ref{lem:avwc:csr:rc}.} Combining Lemmas \ref{lem:goodc} and \ref{lem:avc:dc}, it is concluded that there exists a ``good'' codebook \(\mathcal{C}\) such that
\[
\max_{s^N\in\mathcal{S}^N}\bar{e}^{CSR}(\mathcal{C}, s^N) < \epsilon,
\]
where \(\bar{e}^{CSR}(\mathcal{C}, s^N)\) is given by \eqref{eq:lem:avc:dc}, and the size of the codebook satisfies \(L' > 2^{N(\min_{s\in\mathcal{S}}I(X; Y_s)-\tau')}\). Denote by \((f,\phi)\) the pair of encoder and decoder derived from \(\mathcal{C}\) by the coding scheme introduced in Subsection \ref{sec:thm:avc:csr}, with the message set \(\mathcal{M}'=[1:L']\). Let \(\Pi\) be the collection of all possible permutations on \([1:L']\). For any \(\pi \in \Pi\), let \((f_{\pi}, \phi_{\pi})\) be the encoder-decoder pair such that \(f_{\pi}(m)=f(\pi(m))\) for \(m\in [1:L']\) and \(\phi_{\pi}(y^N, s^N)=\pi^{-1}(\phi(y^N, s^N))\) for \((y^N,s^N) \in \mathcal{Y}^N\times \mathcal{S}^N\). Let \((F,\Phi)\) be the random code uniformly distributed over \(\{(f_{\pi},\phi_{\pi}): \pi\in\Pi\}\). It follows clearly that
\[
\lambda^{CSR}(\mathcal{W},F,\Phi) < \epsilon.
\]
Furthermore, Lemma \ref{lem:avc:ele} claims the existence of an encoder-decoder series \((f_k,\phi_k), 1\leq k \leq N^2\) satisfying \eqref{eq:sec:thm:avwc:csr:rc}. Finally, codebooks \(\mathcal{C}_k=\{f_k(m): m\in\mathcal{M}'\}, 1\leq k \leq N^2\) are ``good'' since they are all permutations of \(\mathcal{C}\). The lemma is proved. \(\hfill\blacksquare\)

\begin{rem}\label{rem:avwc:csr}
Lemma \ref{lem:avwc:csr:rc} constructs a collection of ``good'' codebooks satisfying \eqref{eq:sec:thm:avwc:csr:rc}. With this, we can skip the proof of Step 1 in Section \ref{sec:thm:avwc:rc} and jump to Step 2 directly.
\end{rem}

\section{Discussions and Implications}\label{section:example}\label{sec:example}

This section presents some more detailed discussions on AVWC and AVWC-CSR. Subsection \ref{sec:outer} gives a pair of simple upper bounds on secrecy capacities of AVWC and AVWC-CSR. Subsection \ref{sec:less} introduces the concepts of less noisy AVCs. Three types of less noisiness are given there. With those definitions, we determine the secrecy capacity of AVWC with severely less noisiness, and the secrecy capacity of AVWC-CSR with strongly less noisiness. Subsection \ref{sec:avwc:constrained:seq} studies the models of AVWC with constrained cost or types on state sequence, and establishes a pair of lower bounds on the secrecy capacities, which include the result in Subsection \ref{sec:avwc} as a special case.

\subsection{Simple upper bounds on capacities of AVWC and AVWC-CSR}\label{sec:outer}

This subsection introduces a pair of upper bounds on the secrecy capacities of stochastic code over the AVWC and AVWC-CSR, with respect to the maximal decoding error probability, which is detailed in Lemma \ref{lem:outer1}. Moreover, Example \ref{example1} shows that the upper bounds are different from the corresponding lower bounds in Section \ref{sec:main}.

\begin{lem}\label{lem:outer1}
The secrecy capacity $C_s^{AVWC-SC}(\mathcal{W},\mathcal{V})$ reps. $C_s^{AVWC-CSR-SC}(\mathcal{W},\mathcal{V})$ of stochastic code over a given AVWC $(\mathcal{W},\mathcal{V})$ resp. AVWC-CSR $(\mathcal{W},\mathcal{V})$ must satisfy that
\[
C_s^{AVWC-SC}(\mathcal{W},\mathcal{V}) \leq \min_{q\in\mathcal{P}(\mathcal{S}),s\in\mathcal{S}} \max_{U_{q,s}} [I(U_{q,s};Y_q)-I(U_{q,s};Z_s)]
\]
resp.
\[
C_s^{AVWC-CSR-SC}(\mathcal{W},\mathcal{V}) \leq \min_{s,s'\in\mathcal{S}} \max_{U_{s,s'}} [I(U_{s,s'};Y_s)-I(U_{s,s'};Z_{s'})],
\]
where $\mathcal{W}=\{W_s:s\in\mathcal{S}\}$ and $\mathcal{V}=\{V_s:s\in\mathcal{S}\}$ are a pair of transition matrix collections, $\{(U_{q,s},Y_q,Z_s): q\in\mathcal{P}(\mathcal{S}), s\in\mathcal{S}\}$ is a collection of random variables satisfying \eqref{eq:thm:avwc:rc:1} and \eqref{eq:thm:avwc:rc:2}, and $\{(U_{s,s'},Y_s,Z_{s'}): s,s'\in\mathcal{S}\}$ is a collection of random variables satisfying \eqref{eq:thm:avwc:csr:rc:1} and \eqref{eq:thm:avwc:csr:rc:2}.
\end{lem}

The upper bound on the secrecy capacity of AVWC was established in \cite{Bjelakovic-2013}. We present the proof in Appendix \ref{app:lem:outer1} as a preliminary to the proof of Proposition \ref{prop:less:avwc} in Subsection \ref{sec:less}. The upper bound of AVWC-CSR can be proved similarly, and is hence omitted.

\begin{exmp}\label{example1}
It is clear that the upper bounds introduced in this subsection differ from the corresponding lower bounds in Subsections \ref{sec:avwc} and \ref{sec:avwc:csr}. As a simple example, consider an AVWC-CSR \((\mathcal{W},\mathcal{V})\), depicted in Fig. \ref{Fig:example1:avwc:csr1}, with \(\mathcal{W}=\{W_1,W_2\}\) and \(\mathcal{V}=\{V_1, V_2\}\), where
\[
W_1=
\Bigg[
\begin{matrix}
1 & 0\\
3p/2 & 1-3p/2
\end{matrix}
\Bigg]
,
W_2=
\Bigg[
\begin{matrix}
1-3p/2 & 3p/2\\
0 & 1
\end{matrix}
\Bigg]
 \text{ and }
 V_1 = V_2 =
\Bigg[
\begin{matrix}
1-p & p\\
p & 1-p
\end{matrix}
\Bigg].
\]
The upper bound and lower bound of the secrecy capacity is shown in Fig. \ref{Fig:example1:avwc:csr2}.
\end{exmp}

\begin{figure}
     \centering
     \includegraphics[width=8cm]{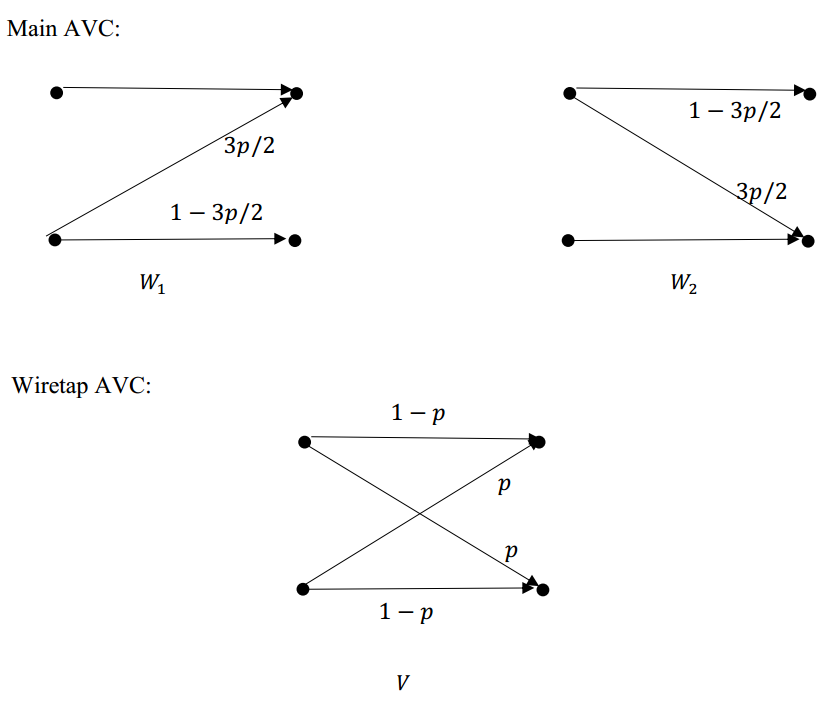}
     \caption{Channel model of the AVWC-CSR in Example \ref{example1}.}\label{Fig:example1:avwc:csr1}
 \end{figure}

\begin{figure}
     \centering
     \includegraphics[width=8cm]{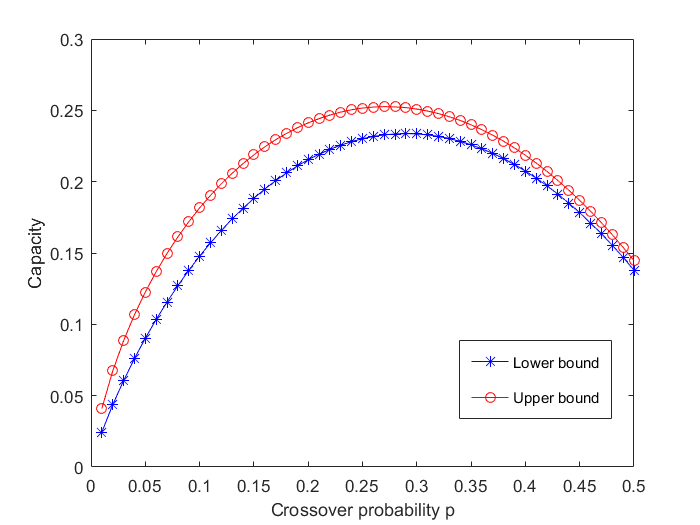}
     \caption{Lower and upper bounds on the secrecy capacity of the AVWC-CSR in Example \ref{example1}.}\label{Fig:example1:avwc:csr2}
 \end{figure}

\subsection{Capacities of less noisy AVWC and AVWC-CSR}\label{sec:less}

Example \ref{example1} shows that the upper bounds introduced in Subsection \ref{sec:outer} differ from the corresponding lower bounds established in Subsections \ref{sec:avwc} and \ref{sec:avwc:csr}. Therefore, the secrecy capacities of the general AVWC and AVWC-CSR are unknown. In this subsection, we introduce the concepts of degraded AVCs and less noisy AVCs, and establish the secrecy capacities of AVWC and AVWC-CSR in Propositions \ref{prop:less:avwc} and \ref{prop:less:avwc:csr}, when the main AVC is less noisy than the wiretap AVC.

\textbf{Definitions of degraded AVCs and less noisy AVCs.}

The concepts of degraded AVCs and less noisy AVCs come from those of DMCs. When considering DMCs, we would call DMC $V:\mathcal{X}\mapsto\mathcal{Z}$ is a degraded version of DMC $W:\mathcal{X}\mapsto\mathcal{Y}$, if there exists a transition matrix $V': \mathcal{Y} \mapsto \mathcal{Z}$ such that
\[
V(z|x) = \sum_{y\in\mathcal{Y}} W(y|x)V'(z|y)
\]
for $x\in\mathcal{X}$ and $z\in\mathcal{Z}$.

A more general relationship between two DMCs is less noisy channels. To be precise, DMC $W$ is said to be less noisy than DMC $V$ if
\[
I(U; Y) \geq I(U; Z)
\]
holds for any \((U,X,Y,Z)\) satisfying
\[
\Pr\{U=u,X=x,Y=y\}=\Pr\{U=u,X=x\}W(y|x)
\]
and
\[
\Pr\{U=u,X=x,Z=z\}=\Pr\{U=u,X=x\}V(z|x).
\]

In general, if DMC $V$ is a degraded version of DMC $W$, then DMC $W$ is less noisy than DMC $V$.

Based on the relationships between DMCs, three intuitive definitions on the relationship of degradation between AVCs \(\mathcal{V} = \{V_s(z|x):x\in \mathcal{X}, z\in \mathcal{Z}, s\in\mathcal{S}\}\)
and \(\mathcal{W} = \{W_s(y|x):x\in \mathcal{X}, y\in \mathcal{Y}, s\in\mathcal{S}\}\) are given as follows.

\begin{itemize}
  \item AVC $\mathcal{V}$ is said to be a weakly degraded version of AVC $\mathcal{W}$ if DMC $V_s$ is a degraded version of DMC $W_s$ for every $s\in\mathcal{S}$.
  \item AVC $\mathcal{V}$ is said to be a strongly degraded version of AVC $\mathcal{W}$ if DMC $V_{s'}$ is a degraded version of DMC $W_s$ for every $s',s\in\mathcal{S}$.
  \item AVC $\mathcal{V}$ is said to be a severely degraded version of AVC $\mathcal{W}$ if DMC $V_{q'}$ is a degraded version of DMC $W_q$ for every $q',q\in \mathcal{P}(\mathcal{S})$.
\end{itemize}

\begin{exmp}
It is clear that severe degradation implies strong degradation, while strong degradation implies weak degradation. However, the opposite is not true. To show this, let
\[
\mathcal{W}_1 = \Bigg\{
\Bigg[
\begin{matrix}
1 & 0\\
0 & 1
\end{matrix}
\Bigg]
,
\Bigg[
\begin{matrix}
1/3 & 2/3\\
2/3 & 1/3
\end{matrix}
\Bigg]
 \Bigg\},
\mathcal{V}_1 = \Bigg\{
\Bigg[
\begin{matrix}
1/4 & 3/4\\
3/4 & 1/4
\end{matrix}
\Bigg]
,
\Bigg[
\begin{matrix}
1/2 & 1/2\\
1/2 & 1/2
\end{matrix}
\Bigg]
 \Bigg\}
\]
and
\[
\mathcal{W}_2 = \Bigg\{
\Bigg[
\begin{matrix}
1 & 0\\
0 & 1
\end{matrix}
\Bigg]
,
\Bigg[
\begin{matrix}
0 & 1\\
1 & 0
\end{matrix}
\Bigg]
 \Bigg\},
\mathcal{V}_2 = \Bigg\{
\Bigg[
\begin{matrix}
1/3 & 2/3\\
2/3 & 1/3
\end{matrix}
\Bigg]
,
\Bigg[
\begin{matrix}
2/3 & 1/3\\
1/3 & 2/3
\end{matrix}
\Bigg]
\Bigg\}.
\]
It follows that $\mathcal{V}_1$ is a weakly degraded version of $\mathcal{W}_1$, while the strong degradation is not true. Meanwhile, $\mathcal{V}_2$ is a strongly degraded version of $\mathcal{W}_2$, while the severe degradation is not true.
\end{exmp}

Similarly, we have the definitions of less noisy AVCs as follows.
\begin{itemize}
  \item AVC $\mathcal{W}$ is said to be weakly less noisy than AVC $\mathcal{V}$ if DMC $W_s$ is less noisy than DMC $V_s$ for every $s\in\mathcal{S}$.
  \item AVC $\mathcal{W}$ is said to be strongly less noisy than AVC $\mathcal{V}$ if DMC $W_{s'}$ is less noisy than DMC $V_s$ for every $s',s\in\mathcal{S}$.
  \item AVC $\mathcal{W}$ is said to be severely less noisy than of AVC $\mathcal{V}$ if DMC $W_{q'}$ is less noisy than DMC $V_q$ for every $q',q\in \mathcal{P}(\mathcal{S})$.
\end{itemize}

Based on the definitions above, we have some direct results. First, severely less noisiness implies strongly less noisiness, and strongly less noisiness implies weakly less noisiness. Second, if AVC \(\mathcal{V}\) is a weakly degraded version of AVC \(\mathcal{W}\), then AVC \(\mathcal{W}\) is weakly less noisy than AVC \(\mathcal{V}\). Similar results hold respectively for strong and severe cases.

\textbf{The capacities.}

In the rest of this subsection, we determine the secrecy capacity of the AVWC where the main AVC is severely less noisy than the wiretap AVC, and the secrecy capacity of the AVWC-CSR where the main AVC is strongly less noisy than the wiretap AVC. The capacity of weakly less noisiness is not discussed here, but it is meaningful when considering the capacity results of AVWC with channel states known at the transmitter, which is beyond the topic of this paper.

\begin{prop}\label{prop:less:avwc}
\emph{(Secrecy capacity of severely less noisy AVWC)} Given an AVWC \((\mathcal{W},\mathcal{V})\), its secrecy capacity of stochastic code with respect to the maximal decoding error probability and strong secrecy criterion is
\begin{equation}\label{eq:sec:less:2}
C_s^{AVWC-SC}(\mathcal{W},\mathcal{V}) =\max_X \min_{q\in\mathcal{P}(\mathcal{S}),s\in\mathcal{S}} [I(X;Y_q)-I(X;Z_s)],
\end{equation}
if the capacity of stochastic code over the main AVC \(\mathcal{W}\) is positive, and the main AVC \(\mathcal{W}\) is severely less noisy than the wiretap AVC \(\mathcal{V}\), where \((X,Y_q,Z_s)\) satisfies
\begin{equation}\label{eq:sec:less:1}
\Pr\{Y_q=y|X=x\}=W_q(y|x), \Pr\{Z_s=z|X=x\}=V_s(y|x),
\end{equation}
and \(W_q\) is given by \eqref{eq:pre:4}.
\end{prop}

\begin{prop}\label{prop:less:avwc:csr}
\emph{(Secrecy capacity of strongly less noisy AVWC-CSR)} Given an AVWC-CSR \((\mathcal{W},\mathcal{V})\), its secrecy capacity of stochastic code with respect to the maximal decoding error probability and strong secrecy criterion is
\[
C_s^{AVWC-CSR-SC}(\mathcal{W},\mathcal{V}) =\max_X \min_{s,s'\in\mathcal{S}} [I(X;Y_s)-I(X;Z_{s'})],
\]
if the main AVC \(\mathcal{W}\) is strongly less noisy than the wiretap AVC \(\mathcal{V}\), where \((X,Y_s,Z_{s'})\) satisfies
\[
\Pr\{Y_s=y|X=x\}=W_s(y|x) \text{ and } \Pr\{Z_{s'}=z|X=x\}=V_{s'}(y|x).
\]
\end{prop}

The proofs of Propositions \ref{prop:less:avwc} and \ref{prop:less:avwc:csr} are similar. We only give the proof of Propositions \ref{prop:less:avwc} in Appendix \ref{app:prop:less:avwc}.

\begin{exmp}\label{example2}
In this example, we will compare the secrecy capacities of AVWC and AVWC-CSR where the main AVC \(\mathcal{W}\) is severely less noisy than the wiretap AVC \(\mathcal{V}\). Both of the wiretap channel models are specified by \((\mathcal{W},\mathcal{V})\), depicted in Fig. \ref{Fig:example2:avwc:less1}, with \(\mathcal{W}=\{W_1,W_2\}\) and \(\mathcal{V}=\{V_1,V_2\}\), where \(\mathcal{S}=\{1,2\}, \mathcal{X}=\{0,1\}\), \(\mathcal{Y}=\mathcal{Z}=\{0, e ,1\}\),
\[
W_1= \Bigg[
\begin{matrix}
1 & 0 & 0\\
0 & q & 1-q
\end{matrix}
\Bigg],
W_2= \Bigg[
\begin{matrix}
1-q & q & 0\\
0 & 0 & 1
\end{matrix}
\Bigg]
\text{ and }
V_1=V_2= \Bigg[
\begin{matrix}
1-p & 0 & p\\
p & 0 & 1-p
\end{matrix}
\Bigg].
\]
It follows clearly that the wiretap AVC \(\mathcal{V}\) is a strongly degraded version of the main AVC \(\mathcal{W}\), for every \(p\) and \(q\). Moreover, by the similar way of proving Property 2 in Theorem 3 of \cite{Nair-2010}, one can conclude that the main AVC \(\mathcal{W}\) is severely less noisy than the wiretap AVC \(\mathcal{V}\) when \(q \leq 2p(1-p)\). In this case, the secrecy capacities of the corresponding AVWC and AVWC-CSR are given in Fig. \ref{Fig:example2:avwc:less2}. Therefore, this example shows that the secrecy capacity of an AVWC-CSR can be strictly larger than that of the corresponding AVWC.
\end{exmp}

\begin{figure}
     \centering
     \includegraphics[width=8cm]{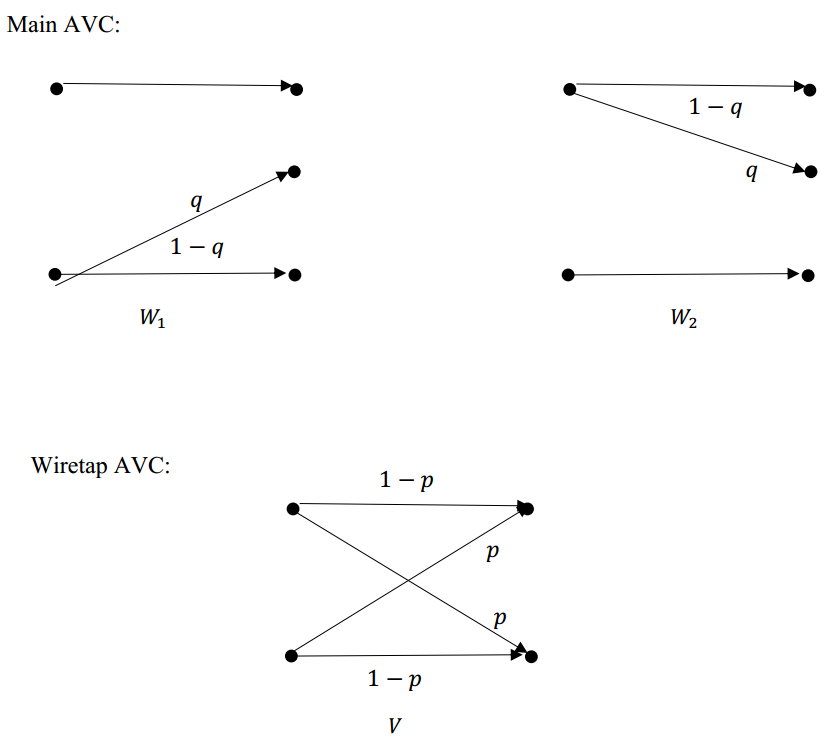}
     \caption{Channel model of main AVC and wiretap AVC in Example \ref{example2}.}\label{Fig:example2:avwc:less1}
 \end{figure}

\begin{figure}
     \centering
     \includegraphics[width=8cm]{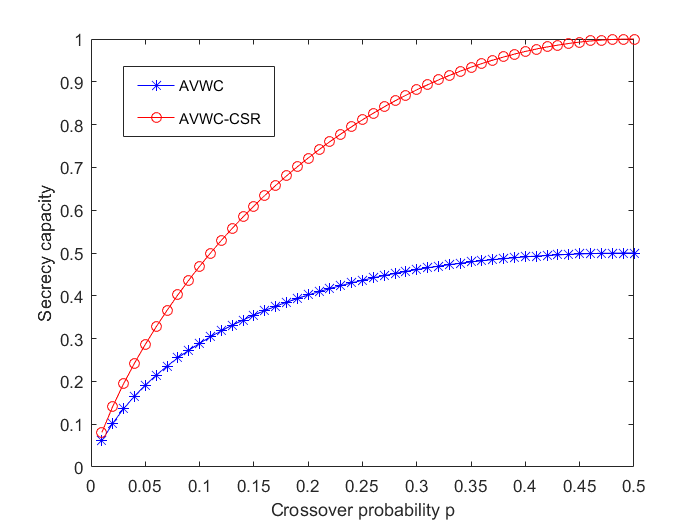}
     \caption{The secrecy capacities of the AVWC and AVWC-CSR in Example \ref{example2} with \(q=2p(1-p)\).}\label{Fig:example2:avwc:less2}
 \end{figure}

\subsection{AVWC with constrained state sequence}\label{sec:avwc:constrained:seq}

In the previous sections, we always assume that the state sequence \(s^N\) is able to run over all the possible values from \(\mathcal{S}^N\). In this subsection, we consider the cases where the state sequence is constrained by cost or types. All the results are about AVWC. The results of AVWC-CSR are not considered here since the capacity of AVC-CSR with constrained states is not established.

\textbf{AVWC with constrained cost on state sequence.}

\begin{defn}
\emph{(Cost of the state sequence)} Let \(c:\mathcal{S}\mapsto \mathbb{R}^*\) be a non-negative function on the state set \(\mathcal{S}\). We define the average cost of a state sequence \(s^N\) as
\begin{equation}\label{eq:avwc:rc:constraint:1}
c(s^N) = \frac{1}{N}\sum_{i=1}^N c(s_i).
\end{equation}
\end{defn}

\begin{defn}\label{def:avwc:sc:constraint}
\emph{(Secure achievability of random code over the AVWC with cost constraint on states)} For any given AVWC \((\mathcal{W}, \mathcal{V})\), a non-negative real number \(R\) is said to be achievable by random code, with respect to the maximal decoding error probability and the strong secrecy criterion, over that AVWC with cost constraint \(\Lambda > 0\) on the states, if for every \(\epsilon > 0\), there exists a random code \((F, \Phi)\) over that AVWC such that
\begin{equation}\label{eq:avwc:rc:constraint:3}
\frac{1}{N} \log |\mathcal{M}| > R - \epsilon, \lambda(\mathcal{W}, F, \Phi, \Lambda) < \epsilon
\text{ and } \max_{s^N \in \mathcal{S}^N: c(s^N)\leq\Lambda} I(M; Z^N(s^N) | \Phi) < \epsilon
\end{equation}
when \(N\) is sufficiently large, where \(M\) is the source message uniformly distributed over the message set \(\mathcal{M}\), \(Z^N(s^N)\) is given by \eqref{eq:def:avwc:z} and
\begin{equation}\label{eq:avwc:rc:constraint:4}
\lambda(\mathcal{W}, F, \Phi, \Lambda) = \max_{s^N \in \mathcal{S}^N: c(s^N)\leq\Lambda} \lambda(\mathcal{W}, F, \Phi, s^N)
\end{equation}
with \(\lambda(\mathcal{W}, F, \Phi, s^N)\) given by \eqref{eq:pre:15}.
\end{defn}

\begin{thm}\label{thm:avwc:rc:constraint}
\emph{(Lower bound on the secrecy capacity of random code over AVWC with cost constraint on states)} Given an AVWC \((\mathcal{W}, \mathcal{V})\) with cost constraint \(\Lambda > 0\) on the states, each real number \(R\) satisfying
\begin{equation}\label{eq:avwc:rc:constraint}
R \leq \max_{U} [\min_{q\in\mathcal{P}(\mathcal{S},\Lambda)}I(U;Y_q) - \max_{q'\in\mathcal{P}(\mathcal{S},\Lambda)}\bar{I}_{q'}(U; Z_{\mathcal{S}})]
\end{equation}
is achievable by random code, where \(\{(U,Y_q,Z_s): q\in\mathcal{P}(\mathcal{S}),s\in\mathcal{S}\}\) is a collection of random variables satisfying \eqref{eq:thm:avwc:rc:1} and \eqref{eq:thm:avwc:rc:2}, \(\bar{I}_{q'}\) is given by \eqref{eq:pre:14}, and
\begin{equation}\label{eq:avwc:rc:constraint:5}
\mathcal{P}(\mathcal{S},\Lambda)=\{q\in\mathcal{P}(\mathcal{S}): c(q)\leq \Lambda\}
\end{equation}
with
\begin{equation}\label{eq:avwc:rc:constraint:2}
c(q) = \sum_{s\in\mathcal{S}}c(s)q(s).
\end{equation}
\end{thm}

The proof of Theorem \ref{thm:avwc:rc:constraint} is outlined at the end of this subsection.

\begin{rem}
Let \(c_M=\max_{s\in\mathcal{S}}c(s)\). Then, when \(\Lambda \geq c_M\), we have \(\mathcal{P}(\mathcal{S},\Lambda) = \mathcal{P}(\mathcal{S})\). In this case, the bound formulated in \eqref{eq:avwc:rc:constraint} is specialized as
\[
\begin{array}{lll}
R &\leq& \max_{U} [\min_{q\in\mathcal{P}(\mathcal{S},\Lambda)}I(U;Y_q) - \max_{q'\in\mathcal{P}(\mathcal{S},\Lambda)}\bar{I}_{q'}(U; Z_{\mathcal{S}})]\\
&=& \max_{U} [\min_{q\in\mathcal{P}(\mathcal{S})}I(U;Y_q) - \max_{q'\in\mathcal{P}(\mathcal{S})}\bar{I}_{q'}(U; Z_{\mathcal{S}})]\\
&=& \max_{U} [\min_{q\in\mathcal{P}(\mathcal{S})}I(U;Y_q) - \max_{s\in\mathcal{S}}I(U; Z_{s})].
\end{array}
\]
Therefore, Theorem \ref{thm:avwc:rc:constraint} includes Theorem \ref{thm:avwc:rc} as a special case.
\end{rem}

\textbf{AVWC with constrained types on state sequence.}

By the definition in \eqref{eq:avwc:rc:constraint:1}, the cost of a state sequence \(s^N\) is totally determined by its type \(P_{s^N}\). To be precise, we have
\[
c(s^N)=c(P_{s^N}),
\]
where \(c(P_{s^N})\) is given by \eqref{eq:avwc:rc:constraint:2}. On account of that, Equations \eqref{eq:avwc:rc:constraint:3} and \eqref{eq:avwc:rc:constraint:4} in Definition \ref{def:avwc:sc:constraint} can be rewritten as
\[
\frac{1}{N} \log |\mathcal{M}| > R - \epsilon, \lambda(\mathcal{W}, F, \Phi, \Lambda) < \epsilon
\text{ and } \max_{s^N \in \mathcal{S}^N: P_{s^N}\in\mathcal{P}(\mathcal{S},\Lambda)} I(M; Z^N(s^N)) < \epsilon
\]
and
\[
\lambda(\mathcal{W}, F, \Phi, \Lambda) = \max_{s^N \in \mathcal{S}^N: P_{s^N}\in\mathcal{P}(\mathcal{S},\Lambda)} \lambda(\mathcal{W}, F, \Phi, s^N).
\]
where \(\mathcal{P}(\mathcal{S},\Lambda)\) is given in \eqref{eq:avwc:rc:constraint:5}. This indicates the definition of achievability is totally characterized by a collection \(\mathcal{P}(\mathcal{S},\Lambda)\) of probability mass functions on \(\mathcal{S}\). Aroused by this property, we have a slightly more general definition on AVWC with constrained state sequence as follows.

\begin{defn}
\emph{(Secure achievability of random code over AVWC with general constraint on states)} For any given AVWC \((\mathcal{W}, \mathcal{V})\), a non-negative real number \(R\) is said to be achievable by random code, with respect to the maximal decoding error probability and the strong secrecy criterion, over that AVWC with states constrained by \(\mathfrak{P}\subseteq\mathcal{P}(\mathcal{S})\), if for every \(\epsilon > 0\), there exists a random code \((F, \Phi)\) over that AVWC such that
\[
\frac{1}{N} \log |\mathcal{M}| > R - \epsilon, \lambda(\mathcal{W}, F, \Phi, \mathfrak{P}) < \epsilon
\text{ and } \max_{s^N \in \mathcal{S}^N: P_{s^N}\in\mathfrak{P}} I(M; Z^N(s^N)|\Phi) < \epsilon
\]
when \(N\) is sufficiently large, where where \(M\) is the source message uniformly distributed over the message set \(\mathcal{M}\), \(Z^N(s^N)\) is given by \eqref{eq:def:avwc:z}, and
\[
\lambda(\mathcal{W}, F, \Phi, \mathfrak{P}) = \max_{s^N \in \mathcal{S}^N: P_{s^N}\in\mathfrak{P}} \lambda(\mathcal{W}, F, \Phi, s^N)
\]
with \(\lambda(\mathcal{W}, F, \Phi, s^N)\) given by \eqref{eq:pre:15}.
\end{defn}

The following result can be proved with little effort.
\begin{prop}
\emph{(Lower bound on the secrecy capacity of random code over AVWC with general constraint on states)} Given an AVWC \((\mathcal{W}, \mathcal{V})\) with channel states constrained by \(\mathfrak{P}\), each real number \(R\) satisfying
\[
R \leq \max_{U} [\inf_{q\in\mathfrak{P}}I(U;Y_q) - \sup_{q'\in\mathfrak{P}}\bar{I}_{q'}(U; Z_{\mathcal{S}})]
\]
is achievable by random code, where $\{(U,Y_q,Z_s): q\in\mathcal{P}(\mathcal{S}), s\in\mathcal{S}\}$ is a collection of random variables satisfying \eqref{eq:thm:avwc:rc:1} and \eqref{eq:thm:avwc:rc:2}.
\end{prop}

As an extremely special case, when \(\mathfrak{P}=\{q\}\) has only one single element, the model is specialized as the model in \cite{Goldfeld-1}. In that case, it follows that

\begin{cor}\label{cor3}
 Given an AVWC \((\mathcal{W}, \mathcal{V})\) with channel states constrained on \(\{q\}\), each real number \(R\) satisfying
\[
R \leq \max_{U} [I(U;Y_q) - \bar{I}_{q}(U; Z_{\mathcal{S}})]
\]
is achievable by random code, where $\{(U,Y_q,Z_s): q\in\mathcal{P}(\mathcal{S}), s\in\mathcal{S}\}$ is a collection of random variables satisfying \eqref{eq:thm:avwc:rc:1} and \eqref{eq:thm:avwc:rc:2}.
\end{cor}

\begin{rem}
We should point out that the result in Corollary \ref{cor3} does not cover the results in \cite{Goldfeld-1}. That's because the capacity in \cite{Goldfeld-1} is with respect to semantic secrecy criterion, which is severer than the strong secrecy criterion in this paper.
\end{rem}

\textbf{Proof of Theorem \ref{thm:avwc:rc:constraint}.}

It suffices to prove that for any collection of random variables \(\{(X,Y_q,Z_s): q\in\mathcal{P}(\mathcal{S}) \text{ and } s\in\mathcal{S}\}\) satisfying \eqref{eq:proof:thm:avwc:rc:6:y} and \eqref{eq:proof:thm:avwc:rc:6:z}, and real numbers \(0 < \tau < \min_{q\in\mathcal{P}(\mathcal{S},\Lambda)}(X; Y_{q}) - \max_{q'\in\mathcal{P}(\mathcal{S},\Lambda)}\bar{I}_{q'}(X; Z_{\mathcal{S}})\) and \(\epsilon > 0\), there exists a pair of random encoder and decoder \((F, \Phi)\) over that AVWC \((\mathcal{W}, \mathcal{V})\) such that
\[
\frac{1}{N} \log |\mathcal{M}| > \min_{q\in\mathcal{P}(\mathcal{S},\Lambda)}(X; Y_{q}) - \max_{q'\in\mathcal{P}(\mathcal{S},\Lambda)}\bar{I}_{q'}(X; Z_{\mathcal{S}}) - \tau,
\]
\[
\lambda(\mathcal{W}, F, \Phi,\Lambda) < \epsilon
\]
and
\[
\max_{s^N: P_{s^N}\in\mathcal{P}(\mathcal{S},\Lambda)} I(M; Z^N(s^N)|\Phi) < \epsilon,
\]
where \(M\) is the source message uniformly distributed over the message set \(\mathcal{M}\), the random sequence \(Z^{N}(s^{N})\), which is defined in \eqref{eq:def:avwc:z}, is the output of wiretap AVC under the state sequence \(s^{N}\), \(\mathcal{P}(\mathcal{S},\Lambda)\) is give by \eqref{eq:avwc:rc:constraint:5}, \(\lambda(\mathcal{W}, F, \Phi,\Lambda)\) is given by \eqref{eq:avwc:rc:constraint:4}, and \(\bar{I}_{q'}\) is given by \eqref{eq:pre:14}.

The proof of Theorem \ref{thm:avwc:rc:constraint} is similar to that of Theorem \ref{thm:avwc:rc} in Subsection \ref{sec:thm:avwc:rc}. We first introduce a preliminary lemma and then present the outline of proof.

\begin{lem}\label{lem:avc:rc:constraint}
Let an AVC \(\mathcal{W}\) with cost constraint \(\Lambda\) on states, a collection of random variables \((X, Y_{\mathcal{S}})\) satisfying Formula \eqref{eq:measure:1}, and real numbers \(0 < \tau' < \min_{q\in\mathcal{P}(\mathcal{S},\Lambda)}(X; Y_{q})\) and \(\epsilon' > 0\) be given. There exists a pair of random encoder and decoder \((F, \Phi)\) distributed over a certain family of deterministic encoder-decoder pairs \(\{(f_g, \phi_g): g \in \mathcal{G}\}\) with \(f_g: \mathcal{M}' \mapsto \mathcal{X}^N\) and \(\phi_g: \mathcal{Y}^N \mapsto \mathcal{M}'\) such that
\[\frac{1}{N} \log |\mathcal{M}'| > \min_{q\in\mathcal{P}(\mathcal{S},\Lambda)}(X; Y_{q}) - \tau' \text{ and } \lambda(\mathcal{W}, F, \Phi,\Lambda) < \epsilon'\]
when \(N\) is sufficiently large, where \(\lambda(\mathcal{W}, F, \Phi,\Lambda)\) is given by \eqref{eq:avwc:rc:constraint:4} and \(\{(X,Y_q): q\in\mathcal{P}(\mathcal{S})\}\) is a collection of random variables satisfying \eqref{eq:pre:9}.
\end{lem}

\begin{proof}
It is a direct corollary of Theorem 3.1 in \cite{Csiszar-1988-2}.
\end{proof}

\begin{rem}\label{rem:avc:rc:constraint}
 Lemma \ref{lem:avc:rc:constraint} is an extension of Lemma \ref{lem:avc:rc}. The proofs of those two lemmas are similar. In other words, the random codebook \(\mathbf{C} = \{F(l)=X^N(l)\}_{l=1}^{L'}\) constructed in the the proof of Theorem 3.1 in \cite{Csiszar-1988-2} (which is more general than Lemma \ref{lem:avc:rc:constraint} in this paper) also satisfies \eqref{eq:lem:goodc} with \(L'=|\mathcal{M}'|\), i.e., the codewords are generated independently based on the probability mass function \(P_X\) of the random variable \(X\). Therefore, Lemma \ref{lem:goodc} claims that the random codebook \(\mathbf{C}\) is ``good'' with high probability.
\end{rem}

The proof of Theorem \ref{thm:avwc:rc:constraint}, similar to that of Theorem \ref{thm:avwc:rc} in Subsection \ref{sec:thm:avwc:rc}, is outlined as follows.

Similar to Step 1 of the proof in Subsection \ref{sec:thm:avwc:rc}, Lemma \ref{lem:avc:rc:constraint}, Remark \ref{rem:avc:rc:constraint} and Lemma \ref{lem:avc:ele} claim that there exists a series of deterministic codes \((f_k, \phi_k), 1 \leq k \leq N^2/2\), such that the codebooks \(\mathcal{C}_k = \{f_k(m): 1\leq m \leq L'\}\), \(1 \leq k \leq N^2/2\) are ``good'' with respect to \(X\) and
      \[
      \max_{m\in \mathcal{M}}\max_{s^N:P_{s^N}\in\mathcal{P}(\mathcal{S},\Lambda)} \Bigg[ \sum_{k=1}^{N^2/2} \frac{2}{N^2} e_m(\mathcal{W}, f_k, \phi_k, s^N) \Bigg] < \epsilon,
      \]
      where \(f_k: \mathcal{M}' \mapsto \mathcal{X}^N\), \(\phi_k: \mathcal{Y}^N \mapsto \mathcal{M}'\), and the size \(L'\) of the message set \(\mathcal{M}'\) satisfies that
      \[
      L' > 2^{N[\min_{q\in\mathcal{P}(\mathcal{S},\Lambda)}I(X; Y_{q}) - \tau']}.
      \]

 Let \(\mathfrak{S}_N=\{s^N: P_{s^N}\in\mathcal{P}(\mathcal{S},\Lambda)\}\) and
 \[
 R_d = \limsup_{n\rightarrow\infty} \max_{s^n\in \mathfrak{S}_n} \bar{I}_{P_{s^n}}(X;Z_{\mathcal{S}}) = \max_{q'\in\mathcal{P}(\mathcal{S},\Lambda)}\bar{I}_{q'}(X; Z_{\mathcal{S}}).
 \]
 For every ``good'' codebook \(\mathcal{C}_k, 1 \leq k \leq N^2/2\), Lemma \ref{lem:goodp2} claims that if the value of \(L\) satisfies that
 \[
 2^{N[\min_{q\in\mathcal{P}(\mathcal{S},\Lambda)}(X; Y_{q}) - \max_{q'\in\mathcal{P}(\mathcal{S},\Lambda)}\bar{I}_{q'}(X; Z_{\mathcal{S}}) - 3\tau'] }< L'\cdot 2^{-N[R_d + 2\tau']} < L <  L'\cdot 2^{-N[R_d + \tau']},
 \]
  then there exists a secure partition \(\{\mathcal{C}_{k, m}\}_{m=1}^L\) on it such that
  \[
  \max_{s^N \in \mathfrak{S}_N} I(\tilde{M}_k; Z^N(\mathcal{C}_k, s^N)) = \max_{s^N: P_{s^N}\in\mathcal{P}(\mathcal{S},\Lambda)} I(\tilde{M}_k; Z^N(\mathcal{C}_k, s^N))  < \epsilon,
  \]
  where \(\tilde{M}_k\) is the index of subcode containing the random sequence \(X^N(\mathcal{C}_k)\), and \((X^N(\mathcal{C}_k),Z^N(\mathcal{C}_k, s^N))\) is a pair of random sequence satisfying \eqref{eq:goodp:26} with \(\mathcal{C}=\mathcal{C}_k\).

  With the secure partitions \(\{\mathcal{C}_{k, m}\}_{m=1}^L\) over the ``good'' codebooks \(\mathcal{C}_k, 1\leq k \leq N^2/2\), the theorem is finally established by the coding scheme introduced in Step 2 of the proof in Subsection \ref{sec:thm:avwc:rc}. \(\hfill\blacksquare\)

\section{Conclusion}\label{section:conclusion}

This paper discusses the secrecy capacity results of general AVWC and AVWC-CSR. Lower bounds on the secrecy capacities of stochastic code and random code with respect to the maximal decoding error probability and strong secrecy criterion are given. It is concluded that the secrecy capacity of stochastic code over an AVWC may be strictly smaller than that of random code, but this situation happens only when the capacity of stochastic code over the main AVC is 0. Meanwhile, the secrecy capacity of stochastic code over AVWC-CSR is identical to that of random code in general.

The secrecy capacities of general AVWC and AVWC-CSR are unknown. However, we determine the secrecy capacity of AVWC where the main AVC is severely less noisy than the wiretap AVC, and that of AVWC-CSR where the main AVC is strongly less noisy than the wiretap AVC.

A new secure partitioning scheme, based on Csisz\'{a}r's almost independent coloring scheme, is proposed, which serves as a fundamental tool to prove the secrecy capacity results in this paper. This powerful scheme can be used to ensure secure transmission against wiretapping through AVC with or without constraint on state sequence.

\section{Acknowledgement}

This work was supported in part by China Program of International S\&T Cooperation 2016YFE0100300, and National Natural Science Foundation of China under Grants 61301178 and 61571293.

\begin{appendices}

\section{Proofs on Typicality}\label{app:typical}

This appendix gives the proofs on the properties of typicality with respect to the state sequence introduced in Subsection \ref{sec:typical}. Since the proofs are quite similar to those in Chapter 1 of \cite{MU_IT}, we only provide the outlines here.

\textbf{Proof of Proposition \ref{remarkTX}.} On account of Theorem 1.1 in \cite{MU_IT}, there exists a series of \(\sigma_a > 0, a \in \mathcal{S}\) such that
\[
\operatorname{Pr}\{X^N \in \tilde{T}^N[X,s^N]_{\delta,\eta}\}
= \displaystyle\prod_{a \in \mathcal{S}(s^N, \delta, \eta)} \operatorname{Pr}\{X_{\mathcal{I}(a:s^N)} \in T^{\mu_a}_{\delta}(P_X)\}
> \displaystyle\prod_{a \in \mathcal{S}(s^N, \delta, \eta)} (1-2^{-N\sigma_a})
>  1 - 2^{-N\nu_1}
\]
for some \(\nu_1 > 0\), where \(\mathcal{I}(a:s^N)\) is given by \eqref{eq:pre:16}, \(\mathcal{S}(s^N, \eta)\) is given by \eqref{eq:pre:17}, \(\mu_a=|\mathcal{I}(a:s^N)|\) and \(m_{Y_{\mathcal{S}}} = \min_{(y,s)\in \mathcal{Y}\times\mathcal{S}: P_{Y_s}(y) > 0} P_{Y_s}(y)\). The property that \(\nu_1\) is independent of \(N\) and \(s^N\), comes from the fact \(\sigma_a\) is only related to \(m_X=\min_{x\in\mathcal{X}:P_X(x)>0}P_X(x)\), \(\delta\) and \(\eta\). \(\hfill\blacksquare\)

\textbf{Proof of Proposition \ref{remarkTY}.} The proposition follows because for every \(y^N \in \tilde{T}^N[Y_{\mathcal{S}}]_{\delta, \eta}\), it is satisfied that
\[
2^{-\mu_a(1+\delta)H(Y_a)} \leq \operatorname{Pr}\{Y_{\mathcal{I}(a:s^N)}(s^N) = y_{\mathcal{I}(a:s^N)}\}
 \leq 2^{-\mu_a(1-\delta)H(Y_a)}
\]
for \(a \in \mathcal{S}(s^N, \eta)\), and
\[
2^{\mu_a\log m_{Y_{\mathcal{S}}}} \leq \operatorname{Pr}\{Y_{\mathcal{I}(a: s^N)}(s^N) = y_{\mathcal{I}(a: s^N)}\} \leq 1
\]
for \(a \notin \mathcal{S}(s^N, \eta)\). \(\hfill\blacksquare\)

Corollary \ref{cor2} is a direct consequence of Proposition \ref{remarkTY}, Proposition \ref{remarkConsist} is obtained immediately from definitions, Proposition \ref{rem:joint} can be proved similarly to Proposition \ref{remarkTY}, Corollary \ref{cor1} is a direct consequence of Proposition \ref{remarkTY} and Corollary \ref{cor2}, and Proposition \ref{remarkTXY} can be proved similarly to Proposition \ref{remarkTX}. The proofs are completed. \(\hfill\blacksquare\)

\section{Proof of Proposition \ref{thm:avc:csr:dc}} \label{app:thm:avc:csr:dc}

This appendix establishes the necessary and sufficient condition of positivity on the capacity of deterministic code over the AVC-CSR \(\mathcal{W}\) with respect to the maximal decoding error probability. To be particular, we would prove that the capacity is positive if and only if there exist a pair of \(x,x' \in \mathcal{X}\), such that for all channel states \(s \in \mathcal{S}\), there exists \(y \in \mathcal{Y}\) satisfying \(W(y|x,s) \neq W(y|x',s).\)

\textbf{Necessity.} We prove this by contradiction. for every pair of \(x\) and \(x'\), assume that there exists \(s(x, x')\in\mathcal{S}\), such that \(W(y|s(x, x'),x) = W(y| s(x, x'), x')\) for all \(y \in \mathcal{Y}\). Suppose that \((f, \phi)\) is a deterministic code of length \(N\) over the AVC-CSR  \(\mathcal{W}\) with positive transmission rate. This indicates that the size of the message set \(\mathcal{M}\) is at least two. We would prove that the maximal decoding error probability \({\lambda}^{CSR}(\mathcal{W}, f, \phi) \geq \frac{1}{2}\).

If \(\lambda_m^{CSR}(\mathcal{W}, f, \phi) \geq \frac{1}{2}\) for all \(m \in \mathcal{M}\), then the proof has been finished. Therefore, we can assume that \(\lambda_{m_0}^{CSR}(\mathcal{W}, f, \phi) = \epsilon < \frac{1}{2}\) for some \(m_0 \in \mathcal{M}\). Let the codeword of the message \(m_0\) be \(f(m_0)=x^N(0)\). Then, for any \(m_1 \in \mathcal{M}\) other than message \(m_0\), set \(x^N(1) = f(m_1)\) and choose the state sequence \(s^N\) with \(s_i = s(x_i(0), x_i(1))\) for \(1\leq i \leq N\). It follows that \(W(y^N|x^N(0), s^N) = W(y^N|x^N(1), s^N)\) for every \(y^N \in \mathcal{Y}^N\). This indicates that
\[
\begin{array}{lll}
\lambda_{m_1}^{CSR}(\mathcal{W}, f, \phi) & \geq & \lambda_{m_1}^{CSR}(\mathcal{W}, f, \phi, s^N)\\
&=& 1 - W(\phi^{-1}(m_1) | x^N(1), s^N)\\
&=& 1 - W(\phi^{-1}(m_1) | x^N(0), s^N)\\
&\geq& W(\phi^{-1}(m_0) | x^N(0), s^N)\\
&=& 1 - \lambda_{m_0}^{CSR}(\mathcal{W}, f, \phi, s^N)\\
&\geq& 1 - \lambda_{m_0}^{CSR}(\mathcal{W}, f, \phi)\\
&=& 1 - \epsilon > \frac{1}{2}.
\end{array}
\]
Therefore, we have \({\lambda}^{CSR}(\mathcal{W}, f, \phi) > \frac{1}{2}\). \(\hfill\blacksquare\)

\textbf{Sufficiency.} Let \(x_0\) and \(x_1\) be a pair of letters from \(\mathcal{X}\) satisfying that for all \(s \in \mathcal{S}\) there exists \(y\) such that \(W(y|s, x_0) \neq W(y|s, x_1)\). Let \(k\) be a sufficiently large integer, and \(\delta\) be sufficiently small such that for every \(s \in \mathcal{S}\),
\begin{itemize}
  \item \(T^k_{\delta}(W(\cdot|s, x_0))\) and \(T^k_{\delta}(W(\cdot|s, x_1))\) are disjoint,
  \item \(W(T^k_{\delta}(W(\cdot|s, x_0))| s^{\bigotimes k}, x_0^{\bigotimes k}) > 1 - 2^{-k\nu_1} \text{ and } W(T^k_{\delta}(W(\cdot|s, x_1))| s^{\bigotimes k}, x_1^{\bigotimes k}) > 1 - 2^{-k\nu_1}\),
\end{itemize}
 where \(\nu_1\) is some positive real number, \(T^k_{\delta}(W(\cdot|s,x_i))\) is the \(\delta\)-letter typical set with respect to the probability mass function \(W(\cdot|s,x_i)\) on \(\mathcal{Y}\), and \(a^{\bigotimes k}\) is a sequence with \(k\) copies of the letter \(a\).

Set \(k' = k\cdot |\mathcal{S}|\) and let \(\mathcal{E} = \{x_0^{\bigotimes k'}, x_1^{\bigotimes k'}\}\) be a message set with two elements. We will prove that the decoder is able to recover the message with the maximal decoding error probability \(< 2^{-k\nu_1}\) if choosing the message randomly from \(\mathcal{E}\) and transmitting it directly to the AVC-CSR. To this end, suppose that the state sequence of the transmission is \(s^{k'} \in \mathcal{S}^{k'}\). There must exist an \(s_0 \in \mathcal{S}\) occurring more than \(k\) times in the sequence \(s^{k'}\). Assume that the first \(k\) letters of \(s^{k'}\) are \(s_0\) without loss of generality. The decoder estimates the source message as \(x_0^{\bigotimes k'}\) if \(y^k \in W(T^k(W(\cdot|s_0, x_0))_{\delta}| s_0^{\bigotimes k}, x_0^{\bigotimes k})\), and estimates it as \(x_1^{\bigotimes k'}\) otherwise, where \(y^k\) is the first \(k\) letters received by the decoder. One can easily verify that the maximal decoding error probability is \(< 2^{-k\nu_1}\). In fact, this coding process constructs a virtual binary channel with crossover probability \(< 2^{-k\nu_1} < \frac{1}{2}\) when \(k\) is sufficiently large, whose capacity is positive. Therefore, a positive transmission rate with the maximal decoding error probability can be achieved. \(\hfill\blacksquare\)

\section{Proof of Lemma \ref{lem:goodc}}\label{app:lem:goodc}

This appendix proves that the ``good'' codebook defined in Definition \ref{def:goodc} can be generated by a random scheme. This is a direct consequence of the Chernoff bound. We first bound the value of \(\operatorname{Pr}\{|\tilde{T}^N(\mathbf{C}, s^N)| < (1-2\cdot2^{-N\nu_1})L'\}\) in \eqref{eq:app:lem:goodc:1}, where \(\tilde{T}^N(\mathbf{C}, s^N) = \mathbf{C} \cap \tilde{T}^N[X,s^N]_{\delta,\eta}\) is the collection of typical codewords with respect to the state sequence \(s^N\). Then, the lemma is established by the union bound in \eqref{eq:app:lem:goodc:2}.

For any \(1 \leq l \leq L'\) and \(s^N \in \mathcal{S}^N\), denote \(U(l, s^N) = 0\) if \(X^N(l) \in T^N[X]_{\delta,\eta}\), and \(U(l, s^N) = 1\) otherwise. By Formula \eqref{eq:lem:goodc} and Proposition \ref{remarkTX} (see also Remark \ref{rem:important}), it follows that
\[
E[U(l, s^N)] < 2^{-N\nu_1}
\]
for \(1\leq l \leq L'\) and \(s^N\in\mathcal{S}^N\).

On account of the Chernoff bound, it holds that
\begin{equation}\label{eq:app:lem:goodc:1}
\begin{array}{lll}
\displaystyle\operatorname{Pr}\{|\tilde{T}^N(\mathbf{C}, s^N)| < (1-2\cdot2^{-N\nu_1})L'\} &= & \displaystyle\operatorname{Pr}\Bigg\{\sum_{l=1}^{L'}U(l, s^N) > 2\cdot2^{-N\nu_1}L'\Bigg\}\\
&\leq& \displaystyle\operatorname{exp_2}\Bigg(-2\cdot2^{-N\nu_1}L'\Bigg) \cdot E\Bigg[\exp_2\Bigg(\sum_{l=1}^{L'}U(l, s^N)\Bigg)\Bigg]\\
&=& \displaystyle\operatorname{exp_2}\Bigg(-2\cdot2^{-N\nu_1}L'\Bigg) \cdot \prod_{l=1}^{L'} E\Bigg[\exp_2\Bigg(U(l, s^N)\Bigg)\Bigg]\\
&\overset{(*)}{\leq}& \displaystyle\operatorname{exp_2}\Bigg(-2\cdot2^{-N\nu_1}L'\Bigg) \cdot \prod_{l=1}^{L'}\exp_e\Bigg(E[U(l, s^N)]\Bigg)\\
&<& \displaystyle\operatorname{exp_2}\Bigg(-2\cdot2^{-N\nu_1}L'\Bigg) \cdot \prod_{l=1}^{L'}\exp_e\Bigg(2^{-N\nu_1}\Bigg)\\
&=& \operatorname{exp_2}\Bigg(-(2-\log e)\cdot2^{-N\nu_1}L'\Bigg),
\end{array}
\end{equation}
where \(\operatorname{exp}_2(x)\) represents \(2^x\), \(\operatorname{exp}_e(x)\) represents \(e^x\), and (*) follows because \(2^t \leq 1+t \leq e^t\) for \(0 \leq t \leq 1\). Thus, summing \(\operatorname{Pr}\{|\tilde{T}^N(\mathbf{C}, s^N)| < (1-2\cdot2^{-N\nu_1})M'\}\) over all \(s^N \in \mathcal{S}^N\), we have
\begin{equation}\label{eq:app:lem:goodc:2}
\begin{array}{lll}
\operatorname{Pr}\{\mathbf{C} \text{ is not ``good''}\} &\leq& \displaystyle\sum_{s^N\in\mathcal{S}^N}\operatorname{Pr}\{|\tilde{T}^N(\mathbf{C}, s^N)| < (1-2\cdot2^{-N\nu_1})L'\}\\
&\leq& \operatorname{exp_2}\Bigg(N\log|\mathcal{S}|-(2-\log e)\cdot2^{-N\nu_1}L'\Bigg).
\end{array}
\end{equation}
The proof is completed by letting \(\epsilon_1 = \operatorname{exp_2}(N\log|\mathcal{S}|-(2-\log e)\cdot2^{-N\nu_1}L')\). \(\hfill\blacksquare\)

\section{Proof of Lemma \ref{lem:goodp2}}\label{app:lem:goodp2}

This appendix constructs a secure partition over a ``good'' codebook, to ensure secure transmission against wiretapping through an AVC with constrained channel states. To be precise, let \(\mathcal{V}=\{V_s: s \in \mathcal{S}\}\) be an AVC, \(X\) be a random variable over \(\mathcal{X}\), and \(Z_{\mathcal{S}}\) be a collection of random variables satisfying \eqref{eq:goodp:25}. Suppose that a ``good'' codebook \(\mathcal{C}\) with respect to \(X\) (defined in Definition \ref{def:goodc}) of size \(L'=2^{R'}\) is given. This appendix proves that for any \(\epsilon > 0, \tau > 0\), and
\[
L < L' \cdot 2^{-N[R_d + \tau]} = 2^{N[R' - R_d - \tau]},
\]
there exists a equipartition \(\{\mathcal{C}_m\}_{m=1}^L\) on it such that
\[
I(\tilde{M}; Z^N(\mathcal{C}, s^N)) < \epsilon
\]
for every \(s^N \in \mathfrak{S}_N\) when \(N\) is sufficiently large, where \(\mathfrak{S}_N\) is a sub-collection of \(\mathcal{S}^N\), and
\[
R_d = \limsup_{n\rightarrow\infty} \max_{s^n\in \mathfrak{S}_n} \bar{I}_{P_{s^n}}(X;Z_{\mathcal{S}}).
\]

The proof is similar to that of Lemma \ref{lem:goodp} with slight adjustments on the parameters. We present the proof outline as follows.

Let \(N\) be sufficiently large such that
\begin{equation}\label{eq:app:goodp2:2}
R_d \geq \max_{s^n\in \mathfrak{S}_N} \bar{I}_{P_{s^N}}(X;Z_{\mathcal{S}}) - \tau/4.
\end{equation}
For any \(s^N\in \mathfrak{S}_N\), let \(\mathcal{B}(\mathcal{C}, s^N)\) be defined as \eqref{eq:goodp:16}. Then it is clear that Formula \eqref{eq:goodp:2} follows. With the help of \(\mathcal{B}(\mathcal{C}, s^N)\), the parameters introduced in Lemma \ref{lemma8} are realized as
\begin{equation}\label{eq:app:goodp2:1}
\begin{array}{c}
\mathcal{A}=\mathcal{C}, \varepsilon= 2^{-N\nu_3/2},\\
 l = 2^{N [R' - R_d - \tau/4]}, \\
 k = L < 2^{N [R' - R_d - \tau]},\\
\mathcal{P} = \{P_{s^N, z^N}: s^N \in \mathfrak{S}_N, z^N \in \mathcal{B}(\mathcal{C}, s^N)\} \cup \{P_0\},
\end{array}
\end{equation}
where \(P_0\) and \(P_{s^N, z^N}\) are given by \eqref{eq:achievable:p_0} and \eqref{eq:achievable:p_zn}, respectively. The verification that parameters in \eqref{eq:app:goodp2:1} satisfy the preconditions in Lemma \ref{lemma8}, is given in Appendix \ref{app:goodp:11}.

The remainder of the proof is now the same as that of Lemma \ref{lem:goodp} in Section \ref{sec:basic}, by replacing the parameters in \eqref{eq:goodp:11} with those in \eqref{eq:app:goodp2:1} and the references of \(s^N\in\mathcal{S}^N\) with \(s^N\in\mathfrak{S}_N\). \(\hfill\blacksquare\)

\section{Verification on the rationality of parameters in Formulas \eqref{eq:goodp:11} and \eqref{eq:app:goodp2:1} }\label{app:goodp:11}

This subsection proves that the parameters realized in \eqref{eq:goodp:11} and \eqref{eq:app:goodp2:1} are rational and they satisfy the preconditions of Lemma \ref{lemma8},  when the codeword length \(N\) is sufficiently large.

\textbf{Verification on the rationality of \eqref{eq:goodp:11}.}

We rewrite the parameters in \eqref{eq:goodp:11} here for convenience.
\[
\begin{array}{c}
\mathcal{A}=\mathcal{C}, \varepsilon= 2^{-N\nu_3/2},\\
 l = 2^{N [R' - \max_{s\in\mathcal{S}}I(X;Z_s) - \tau/2]}, \\
 k = L < 2^{N [R' - \max_{s\in\mathcal{S}}I(X;Z_s) - \tau]},\\
\mathcal{P} = \{P_{s^N, z^N}: s^N \in \mathcal{S}^N, z^N \in \mathcal{B}(\mathcal{C}, s^N)\} \cup \{P_0\},
\end{array}
\]
where the size of \(\mathcal{C}\) is \(L'=2^{NR'}\), and \(\mathcal{B}(\mathcal{C}, s^N)\), \(P_0\) and \(P_{s^N,z^N}\) are given by \eqref{eq:goodp:16}, \eqref{eq:achievable:p_0} and \eqref{eq:achievable:p_zn}, respectively.

\emph{Proof of \(\varepsilon < \frac{1}{9}\).} It follows because \(\varepsilon = 2^{-N\nu_3/2}\) can be arbitrarily small as \(N \rightarrow \infty\).

\emph{Proof of \(k \log k \leq \frac{\epsilon^2l}{3\log(2|\mathcal{P}|)}\).} It follows clearly that \(|\mathcal{P}| < (|\mathcal{S}||\mathcal{Z}|)^N\). Therefore, we have
\[
\begin{array}{lll}
\frac{\varepsilon^2 l} {3 \log (2 |\mathcal{P}|)} & =&\frac{1} {3 \log (2 |\mathcal{P}|)} 2^{N [R'-\max_{s\in\mathcal{S}}I(X; Z_{s}) - \tau/2 - \nu_3]}\\
& > & 2^{N [R'-\max_{s\in\mathcal{S}}I(X; Z_{s}) - \tau/2 - \nu_3 - \frac{\log [3N\cdot \log (2|\mathcal{S}||\mathcal{Z}|)]}{N}]}\\
& \overset{(a)}{>} & 2^{N [R'-\max_{s\in\mathcal{S}}I(X; Z_{s}) - 3\tau/4]} \overset{(b)}{>} L \log L = k \log k,
\end{array}
\]
where the inequalities (a) and (b) follow when \(\nu_3\) is sufficiently small and \(N\) is sufficiently large.

\emph{Proof of \eqref{eq:goodp:7}.} Notice that \(X^N(\mathcal{C})\) is uniformly distributed over the codebook \(\mathcal{C}\) of size \(L'=2^{NR'}\). This indicates that
\[P_0(x^N) = \Pr\{X^N(\mathcal{C})=x^N\} = L'^{-1} = 2^{-NR'} < l^{-1}\]
for any \(x^N \in \mathcal{C}\).  Therefore
\begin{equation}\label{eq:goodp:24}
\sum_{x^N \in \mathcal{C}: P_0(x^N) > l^{-1}} P_0(x^N) = 0 < \varepsilon,
\end{equation}
which is of \eqref{eq:goodp:7} for \(P_0\).

To establish \eqref{eq:goodp:7} for \(P_{s^N, z^N}\) with \(s^N \in \mathcal{S}^N\) and \(z^N \in \mathcal{B}(\mathcal{C}, s^N)\), let
\begin{equation}\label{eq:goodp:27}
\tilde{T}(\mathcal{C}, s^N, z^N) = \mathcal{C} \cap \tilde{T}^N[XZ_{\mathcal{S}}|s^N, z^N]_{2\delta, \eta}
\end{equation}
be the set of codewords which are jointly typical with \(z^N\) under the state sequence \(s^N\). It follows that for any \(x^N \in \tilde{T}(\mathcal{C}, s^N, z^N)\),
\begin{equation}\label{eq:goodp:22}
\begin{array}{lll}
P_{s^N, z^N}(x^N) &=& \frac{\operatorname{Pr}\{X^N(\mathcal{C}) = x^N, Z^N(\mathcal{C}, s^N) = z^N\}}{\operatorname{Pr}\{Z^N(\mathcal{C}, s^N) = z^N\}} \\
& \overset{(a)}{=}& \frac{\prod_{i=1}^N P_{Z_{s_i}|X}(z_i | x_i)}{L' \cdot \operatorname{Pr}\{Z^N(\mathcal{C}, s^N) = z^N\}} \\
& \overset{(b)}{\leq}& \frac{\prod_{i=1}^N P_{Z_{s_i}|X}(z_i | x_i)}{2^{-N\nu_3/2}L' \cdot \prod_{i =1}^N P_{Z_{s_i}}(z_i)} \\
& \overset{(c)}{\leq}& \frac{2^{-N(1-2\delta-\eta)\bar{H}_{P_{s^N}}(Z_{\mathcal{S}}|X)}}{2^{-N\nu_3/2}L' \cdot 2^{-N[(1+2\delta)\bar{H}_{P_{z^N}}(Z_{\mathcal{S}}) - \eta \log m_{Y_{\mathcal{S}}}]}} \\
&\overset{(d)}{\leq}& 2^{-N[R' - \bar{I}_{P_{s^N}}(X; Z_{\mathcal{S}}) - (4\delta + \eta) \bar{H}_{P_{s^N}}(Z_{\mathcal{S}}) + \eta \log m_{Y_{\mathcal{S}}} + \nu_3/2] }\\
&\overset{(e)}{\leq}& 2^{-N[R' - \max_{s\in\mathcal{S}}I(X; Z_{s}) - (4\delta + \eta) \bar{H}_{P_{s^N}}(Z_{\mathcal{S}}) + \eta \log m_{Y_{\mathcal{S}}} - \nu_3/2] }\\
&\overset{(f)}{\leq}& 2^{-N[R' - \max_{s\in\mathcal{S}}I(X; Z_{s}) - \tau / 2] } = l^{-1},
\end{array}
\end{equation}
where
\begin{itemize}
  \item (a) follows from \eqref{eq:goodp:26};
  \item (b) follows because \(z^N \in \mathcal{B}(\mathcal{C}, s^N)\) implies \(z^N \notin \mathcal{B}_1(\mathcal{C}, s^N)\) (see Formulas \eqref{eq:goodp:16} and \eqref{eq:goodp:15});
  \item (c) follows from the fact of \((x^N, z^N) \in \tilde{T}^N[XZ_{\mathcal{S}},s^N]_{2\delta, \eta}\) along with Propositions \ref{remarkTX} and \ref{rem:joint};
  \item (d) follows from the facts that \(L'=2^{NR'}\) and \(\bar{H}_{P_{s^N}}(Z_{\mathcal{S}}|X) \leq \bar{H}_{P_{s^N}}(Z_{\mathcal{S}})\);
  \item (e) follows from \eqref{eq:pre:12};
  \item and (f) follows when \(\delta\), \(\eta\) and \(\nu_3\) are sufficiently small.
\end{itemize}

Recalling that \(z^N \in \mathcal{B}(\mathcal{C}, s^N)\) implies \(z^N \in \mathcal{B}_0(\mathcal{C}, s^N)\) (cf. \eqref{eq:goodp:16}), it follows from \eqref{eq:goodp:22} and \eqref{eq:goodp:12} that
\[
\begin{array}{lll}
\displaystyle \sum_{x^N \in \mathcal{C}: P_{s^N, z^N}(x^N) > l^{-1}} P_{s^N, z^N}(x^N) & \leq& 1  - \operatorname{Pr}\{X^N(\mathcal{C}) \in \tilde{T}(s^N, z^N, \mathcal{C}) | Z^N(\mathcal{C}, s^N) = z^N\}\\
& <& 2^{-N\nu_3/2} = \varepsilon,
\end{array}
\]
which is of \eqref{eq:goodp:7} for \(P_{s^N, z^N}\). The verification is completed. \(\hfill\blacksquare\)

\textbf{Verification on the rationality of \eqref{eq:app:goodp2:1}.}

The verification of \eqref{eq:app:goodp2:1} is quite similar to the verification of \eqref{eq:goodp:11}. The only difference is that the state sequence \(s^N\) runs over a subset \(\mathfrak{S}_N\) of \(\mathcal{S}^N\) when verifying \eqref{eq:app:goodp2:1}. We only prove that \(P_{s^N,z^N}(x^N) < l^{-1}\) for \(s^N\in\mathfrak{S}_N\), \(z^N\in\mathcal{B}(\mathcal{C},z^N)\) and \(x^N\in \tilde{T}(\mathcal{C}, s^N, z^N)\), where \(\mathcal{B}(\mathcal{C},z^N)\) is defined in \eqref{eq:goodp:16} and \(\tilde{T}(\mathcal{C}, s^N, z^N)\) is defined in \eqref{eq:goodp:27}.
In fact,
\[
\begin{array}{lll}
P_{s^N, z^N}(x^N) & \overset{(a)}{\leq}& 2^{-N[R' - \bar{I}_{P_{s^N}}(X; Z_{\mathcal{S}}) - (4\delta + \eta) \bar{H}_{P_{s^N}}(Z_{\mathcal{S}}) + \eta \log m_{Y_{\mathcal{S}}} + \nu_3/2] }\\
&\overset{(b)}{\leq}& 2^{-N[R' - R_d - \tau/4 - (4\delta + \eta) \bar{H}_{P_{s^N}}(Z_{\mathcal{S}}) + \eta \log m_{Y_{\mathcal{S}}} - \nu_3/2] }\\
&\overset{(c)}{\leq}& 2^{-N[R' - R_d - \tau / 2] } = l^{-1},
\end{array}
\]
where (a) follows from the inequalities (a)-(d) in \eqref{eq:goodp:22}, (b) follows from \eqref{eq:app:goodp2:2}, and (c) follows when \(\delta\), \(\eta\) and \(\nu_3\) are sufficiently small. The verification is completed. \(\hfill\blacksquare\)

\section{Proof of Lemma \ref{lem:avc:dc}}\label{app:lem:avc:dc}

In this appendix, we prove that the random coding scheme introduced in Subsection \ref{sec:thm:avc:csr} is able to create a deterministic codebook with high probability, such that the average decoding error probability over the AVC-CSR is vanishing. This establishes the capacity of deterministic code over the AVC-CSR with respect to the average decoding error probability.

To simplify the notation, let
\[
X^N([1:m]) = (X^N(1), X^N(2),...,X^N(m)),
\]
and
\[
x^N([1:m]) = (x^N(1), x^N(2),...,x^N(m)).
\]
According to the decoding scheme in Subsection \ref{sec:thm:avc:csr}, for every given codebook \(\mathcal{C} = \{x^N(m)\}_{m=1}^L\), the decoding error probability of the message \(m\) is totally determined by the codewords \(x^N([1:m])\) and the state sequence $s^N$. Let \(e_m[x^N(m), x^N([1:m-1]), s^N]\) be the decoding error probability of the message \(m\) under the state sequence \(s^N\) when the first \(m\) codewords are \(x^N([1:m])\). Then the average decoding error probability under the state sequence $s^N$ can be rewritten as
\[
\bar{e}^{CSR}(\mathcal{C}, s^N) = \frac{1}{L} \sum_{m=1}^L e_m[x^N(m), x^N([1:m-1]), s^N].
\]

The proof of Lemma \ref{lem:avc:dc} is organized as the following 3 steps.
\begin{itemize}
  \item Step 1 shows that there exists \(\nu > 0\) such that
  \begin{equation}\label{eq:app:avc:csr:dc:1}
  \operatorname{Pr}\{e_m[X^N(m), x^N([1:m-1]), s^N] > 2\cdot 2^{-N\nu_2}\} < 2^{-N\nu}
  \end{equation}
  for every \(1\leq m \leq L\), \(x^N([1:m-1]) \subseteq \mathcal{X}^{N\times (m-1)}\) and \(s^N \in \mathcal{S}^N\), when \(N\) is sufficiently large. This means no matter what values the first $m-1$ codewords are, if choosing the $m$-th codewords randomly, one can get a ``good'' codeword with a high probability, such that the decoding error probability of message $m$ is very small.
  \item Step 2 proves that
  \begin{equation}\label{eq:avc:csr:10}
  \operatorname{Pr}\Bigg\{\frac{1}{L} \sum_{m=1}^N e_m[X^N(m), X^N([1:m-1]), s^N]  > \epsilon \Bigg\} < \exp_2(-L\epsilon / 2)
  \end{equation}
  for every \(s^N \in \mathcal{S}^N\), when \(N\) is sufficiently large. This means for every sequence $s^N$, the probability that random codebook fails to achieve a vanishing average decoding error probability is doubly exponential small.
  \item Step 3 establishes \eqref{eq:thm:avc:dc} by the union bound.
\end{itemize}

\textbf{Proof of Step 1.} In this part, random events \(\operatorname{Err}_1\) and \(\operatorname{Err}_2\) are introduced, and the value of \(\operatorname{Pr}\{e_m[X^N(m), x^N([1:m-1]), s^N] > 2\cdot 2^{-N\nu_2}\}\) is bounded by Formula \eqref{eq:avc:csr:6}. Formula \eqref{eq:app:avc:csr:dc:1} is finally proved by substituting Formulas \eqref{eq:avc:csr:3} and \eqref{eq:avc:csr:4} into \eqref{eq:avc:csr:6}.

According to the decoding scheme, if formulas \eqref{eq:avc:csr:1} and \eqref{eq:avc:csr:2} hold, it must follow that \[e_m[x^N(m), x^N([1:m-1]), s^N] \leq 2 \cdot 2^{-N\nu_2}.\]
Therefore,
\begin{equation}\label{eq:avc:csr:6}
\operatorname{Pr}\{e_m[X^N(m), x^N([1:m-1]), s^N] > 2\cdot 2^{-N\nu_2}\}
\leq \operatorname{Pr}\{\operatorname{Err}_1 = 1\} + \operatorname{Pr}\{\operatorname{Err}_2 = 1 | \operatorname{Err}_1 = 0\},
\end{equation}
where \(\operatorname{Err}_1 = 1\) (or \(= 0\)) represents that the random event \(X^N(m) \in \mathfrak{X}_1\)
happens (or not), and \(\operatorname{Err}_2 = 1\) (or \(= 0\)) represents that the random event \(X^N(m)\in\mathfrak{X}_2\) happens (or not). The sets \(\mathfrak{X}_1\) and \(\mathfrak{X}_2\) are defined as
\[
\mathfrak{X}_1=\mathcal{X}^N \setminus \tilde{T}^N[X,s^N]_{\delta,\eta}
\]
and
\begin{equation}\label{eq:avc:csr:11}
\mathfrak{X}_2=\{x^N: W(\tilde{T}^N[XY_{\mathcal{S}}, s^N| x^N]_{2\delta,\eta}\setminus\tilde{\mathcal{D}}_{m-1}(s^N)|x^N, s^N) < 1-2 \cdot 2^{-N{\nu_2}}\}.
\end{equation}

On account of Proposition \ref{remarkTX}, we have
\begin{equation}\label{eq:avc:csr:3}
\operatorname{Pr}\{\operatorname{Err}_1 = 1\} < 2^{-N\nu_1}.
\end{equation}
It remains to bound the value of \(\operatorname{Pr}\{\operatorname{Err}_2 = 1 | \operatorname{Err}_1 = 0\}\). To achieve this, notice that for every \(x^N \in \tilde{T}^N[X,s^N]_{\delta,\eta}\) (or, equivalently, \(x^N \in \mathcal{X}^N \setminus \mathfrak{X}_1\)), Proposition \ref{remarkTXY} gives
\begin{equation}\label{eq:avc:csr:12}
W(\tilde{T}^N[XY_{\mathcal{S}}, s^N| x^N(m)]_{2\delta,\eta}|x^N, s^N) > 1-2^{-N{\nu_2}}.
\end{equation}
Therefore, if \(x^N\in \mathfrak{X}_2 \setminus \mathfrak{X}_1\), it follows from \eqref{eq:avc:csr:12} and \eqref{eq:avc:csr:11} that
\[
W(\tilde{\mathcal{D}}_{m-1}(s^N)| x^N, s^N) > 2^{-N{\nu_2}}.
\]
This indicates that
\[
(\mathfrak{X}_2 \setminus \mathfrak{X}_1) \subseteq (\mathfrak{X}_3 \setminus \mathfrak{X}_1),
\]
where
\[
\mathfrak{X}_3=\{x^N\in\mathcal{X}^N: W(\tilde{\mathcal{D}}_{m-1}(s^N)| x^N, s^N) > 2^{-N{\nu_2}}\}.
\]
Therefore, we have
\[
\operatorname{Pr}\{\operatorname{Err}_2 = 1 | \operatorname{Err}_1 = 0\} \leq \operatorname{Pr}\{\operatorname{Err}_3 = 1 | \operatorname{Err}_1 = 0\},
\]
where \(\operatorname{Err}_3 = 1\) (or \(= 0\)) represents that the random event \(X^N(m)\in\mathfrak{X}_3\) happens (or not).

To bound the value of \(\operatorname{Pr}\{\operatorname{Err}_3 = 1 | \operatorname{Err}_1 = 0\}\), notice that no matter what the value of \(\tilde{\mathcal{D}}_{m-1}(s^N)\) is, Corollary \ref{cor1} always indicates that
\[
\begin{array}{lll}
\operatorname{Pr}\{Y^N(s^N) \in \tilde{\mathcal{D}}_{m-1}(s^N)\} &=& \displaystyle\sum_{1\leq l \leq m-1} \operatorname{Pr}\{Y^N(s^N) \in \mathcal{D}_{l}(s^N)\} \\
&\leq&  \displaystyle\sum_{1\leq l \leq m-1} \operatorname{Pr}\{Y^N(s^N) \in \tilde{T}^N[XY_{\mathcal{S}}, s^N|x^N(l)]_{2\delta,\eta}\}\\
&\overset{(*)}{<}& L \cdot 2^{-N[\bar{I}_{P_{s^N}}(X; Y_{\mathcal{S}}) - (4\delta +\eta) \bar{H}_{P_{s^N}}(Y_{\mathcal{S}}) + \eta \log m_{XY_{\mathcal{S}}}]} \\
&<& 2^{-N[\tau/2 - (4\delta +\eta) \bar{H}_{P_{s^N}}(Y_{\mathcal{S}}) + \eta \log m_{XY_{\mathcal{S}}}]} \\
&<& 2^{-N\tau/4}
\end{array}
\]
when \(\delta\) and \(\eta\) are sufficiently small, where (*) follows because
\(
L < 2^{N[\min_{s\in\mathcal{S}}I(X;Y_s) - \tau/2]} < 2^{N[\bar{I}_{P_{s^N}}(X; Y_{\mathcal{S}})-\tau/2]}.
\)
Combining the formula above and the following inequality
\[
\begin{array}{lll}
\operatorname{Pr}\{Y^N(s^N) \in \tilde{\mathcal{D}}_{m-1}(s^N)\} &=& \displaystyle\sum_{x^N\in\mathcal{X^N}} \operatorname{Pr}\{X^N(l) = x^N\}W(\tilde{\mathcal{D}}_{m-1}(s^N)|x^N, s^N)\\
&>& \operatorname{Pr}\{X^N(l)\in\mathfrak{X}_3\}\cdot 2^{-N{\nu_2}}= \operatorname{Pr}\{\operatorname{Err}_3 = 1\}\cdot 2^{-N{\nu_2}},
\end{array}
\]
it is concluded that
\[
\operatorname{Pr}\{\operatorname{Err}_3 = 1 \} < 2^{-N(\tau/4 - \nu_2)}.
\]
Therefore,
\begin{equation}\label{eq:avc:csr:4}
\operatorname{Pr}\{\operatorname{Err}_2 = 1 | \operatorname{Err}_1 = 0\} \leq \operatorname{Pr}\{\operatorname{Err}_3 = 1 | \operatorname{Err}_1 = 0\} \leq \frac{\operatorname{Pr}\{\operatorname{Err}_3 = 1 \}}{\operatorname{Pr}\{\operatorname{Err}_1 = 0 \}} < 2\cdot2^{-N(\tau/4 - \nu_2)}
\end{equation}
when \(N\) is sufficiently large.
Substituting Formulas \eqref{eq:avc:csr:3} and \eqref{eq:avc:csr:4} into \eqref{eq:avc:csr:6} gives
\[
\operatorname{Pr}\{e_m[X^N(m), x^N([1:m-1]), s^N] > 2\cdot 2^{-N\nu_2}\} < 2^{-N\nu_1} + 2\cdot2^{-N(\tau/4 - \nu_2)} < 2^{-N\nu}
\]
for some \(\nu > 0\), where \(\nu_2\) can be set as arbitrarily small without violating Proposition \ref{remarkTXY}. The proof of Step 1 is completed. \(\hfill\blacksquare\)

\textbf{Proof of Step 2.} The proof depends on the following lemma.

\begin{lem}\label{lem:rv}
Let \(A_m, 1 \leq m \leq L\) be a series of (not necessarily mutually independent) random variables with \(A_m\) distributed over \(\mathcal{A}_m\) such that
\begin{equation}\label{eq:lem:rv}
E[f_m(A_m)| A_1=a_1, A_2=a_2,...,A_{m-1}=a_{m-1}] < b
\end{equation}
for arbitrary \(a_1 \in \mathcal{A}_1, a_2 \in \mathcal{A}_2,...,a_{m-1} \in \mathcal{A}_{m-1}\) with \(\operatorname{Pr}\{A_1=a_1, A_2=a_2,...,A_{m-1}=a_{m-1}\} > 0\), where \(f_m: \mathcal{A}_m \mapsto (0, \infty), 1 \leq m \leq L\) are bounded and \(b > 0\) is a constant real number. Then it follows that
\[
E\Bigg[\prod_{m=1}^L f_m(A_m) \Bigg] < b^{L}.
\]
\end{lem}

The proof of Lemma \ref{lem:rv} is deferred to the end of this appendix.

On account of the Chernoff bound, for any given \(s^N\), it follows that
\begin{equation}\label{eq:avc:csr:8}
\begin{array}{lll}
&& \displaystyle\operatorname{Pr}\Bigg\{\frac{1}{L} \sum_{m=1}^N e_m[X^N(m), X^N([1:m-1]), s^N]  > \epsilon \Bigg\} \\
& < & \displaystyle\operatorname{exp}_2(-L\epsilon) E \Bigg[ \exp_2\Bigg( \sum_{m=1}^L e_m[X^N(m), X^N([1:n-1]), s^N] \Bigg) \Bigg]\\
&=& \displaystyle \operatorname{exp}_2(-L\epsilon) E \Bigg[ \prod_{m=1}^L  \exp_2\Bigg(e_m[X^N(m), X^N([1:n-1]), s^N] \Bigg) \Bigg].
\end{array}
\end{equation}
To proceed with the upper bounding, let \(A_m = X^N([1:m])\) and \(f(A_m) = \exp_2(e_m[X^N(m), X^N([1:n-1]), s^N])\). This indicates that
\begin{equation}\label{eq:avc:csr:7}
\begin{array}{lll}
&& E\Bigg[f(A_m)|X^N([1:m-1]) = x^N([1:m-1])\Bigg] \\
&=& E\Bigg[\exp_2\Bigg(e_m[X^N(m), x^N([1:m-1]), s^N]\Bigg)\Bigg] \\
&\leq & 2\cdot \operatorname{Pr}\Bigg\{e_m[X^N(m), x^N([1:m-1]), s^N] > 2\cdot 2^{-N\nu_2}\Bigg\} + \exp_2(2\cdot 2^{-N\nu_2}) \\
& \overset{(a)}{<} & 2 \cdot 2^{-N\nu} + \exp_2(2\cdot 2^{-N\nu_2}) \\
& < & \exp_2(2\cdot 2^{-N\nu_2}) \cdot (1 + 2^{1-N\nu})\\
& \overset{(b)}{<} & \exp_2(2^{1-N\nu_2}) \cdot \exp_e(2^{1-N\nu})
\end{array}
\end{equation}
for every \(x^N([1:m-1])\), where (a) follows from \eqref{eq:app:avc:csr:dc:1} and (b) follows from the inequality \(1 + t \leq e^t\) for \(t \geq 0\).

Consequently, Lemma \ref{lem:rv} gives
\begin{equation}\label{eq:avc:csr:9}
\begin{array}{lll}
\displaystyle E \Bigg[ \prod_{m=1}^L  \exp_2\Bigg(e_m[X^N(m), X^N([1:n-1]), s^N] \Bigg) \Bigg] &=&
\displaystyle E \Bigg[  \prod_{m=1}^L f(A_m)\Bigg]\\
&<&\Bigg[ \exp_2(2^{1-N\nu_2}) \cdot \exp_e(2^{1-N\nu}) \Bigg]^L.
\end{array}
\end{equation}
Substituting \eqref{eq:avc:csr:9} into \eqref{eq:avc:csr:8}, we have
\[
\begin{array}{lll}
&& \displaystyle\operatorname{Pr}\Bigg\{\frac{1}{L} \sum_{m=1}^N e_m[X^N(m), X^N([1:m-1]), s^N]  > \epsilon \Bigg\} \\
& < & \displaystyle \operatorname{exp}_2(-L\epsilon) E \Bigg[ \prod_{m=1}^L  \exp_2\Bigg(e_m[X^N(m), X^N([1:n-1]), s^N] \Bigg) \Bigg]\\
& < & \exp_2(-L(\epsilon - 2^{1-N\nu_2} - 2^{1-N\nu} \log e))\\
& < & \exp_2(-L\epsilon / 2)
\end{array}
\]
when \(N\) is sufficiently small. The proof of Step 2 is completed. \(\hfill\blacksquare\)

\textbf{Proof of Step 3.} According to the union bound, Formula \eqref{eq:avc:csr:10} gives
\[
\begin{array}{lll}
\displaystyle\operatorname{Pr}\Bigg\{ \max_{s^N \in \mathcal{S}^N} \bar{e}^{CSR}(\mathbf{C}, s^N) > \epsilon \Bigg\} &=&
\displaystyle\operatorname{Pr}\Bigg\{ \max_{s^N \in \mathcal{S}^N} \Bigg[ \frac{1}{L} \sum_{m=1}^N e_m[X^N(m), X^N([1:m-1]), s^N]  \Bigg] > \epsilon \Bigg\} \\
& \leq & \displaystyle \sum_{s^N \in \mathcal{S}^N} \operatorname{Pr}\Bigg\{ \frac{1}{L} \sum_{m=1}^N e_m[X^N(m), X^N([1:m-1]), s^N] > \epsilon \Bigg\} \\
& < & |\mathcal{S}|^N \cdot \exp_2(-L \epsilon / 2) = \epsilon_3.
\end{array}
\]
The proof of Lemma \ref{lem:avc:dc} is completed. \(\hfill\blacksquare\)

\textbf{Proof of Lemma \ref{lem:rv}.} The lemma is proved by induction on the number \(L\) of random variables. It is clear that the lemma is true when \(L=1\). Suppose that the lemma holds for \(L = L_0 \geq 1\). When \(L=L_0+1\) we have
\[
\begin{array}{lll}
\displaystyle E[\prod_{m=1}^{L_0+1}f(A_m)] &=& \displaystyle\sum_{a[1:L_0+1]}\Bigg\{p\{a[1:L_0+1]\} \prod_{m=1}^{L_0+1}f(a_m)\Bigg\}\\
&=& \displaystyle\sum_{a[1:L_0]}\Bigg\{\Bigg[p\{a[1:L_0]\} \prod_{m=1}^{L_0}f(a_m)\Bigg]\Bigg[\sum_{a_{L_0+1}}p\{a_{L_0+1} | a[1:L_0]\}f(a_{L_0+1})\Bigg]\Bigg\}\\
&=& \displaystyle\sum_{a[1:L_0]}\Bigg\{\Bigg[p\{a[1:L_0]\} \prod_{m=1}^{L_0}f(a_m)\Bigg]E\{f(A_{L_0+1}) | A[1:L_0]=a[1:L_0]\}\Bigg\}\\
&\overset{(a)}{\leq}& b\cdot \displaystyle\sum_{a[1:L_0]}\Bigg\{p\{a[1:L_0]\} \Bigg[\prod_{m=1}^{L_0}f(a_m)\Bigg]\Bigg\}\\
&=& \displaystyle b\cdot E[\prod_{m=1}^{L_0}f(A_m)]\\
&\overset{(b)}{\leq}& b^{L_0+1},
\end{array}
\]
where \(p\{a[1:L]\}\) is short for \(\Pr\{A[1:L]=a[1:L]\}\), \(p\{a_{L_0+1} | a[1:L_0]\}\) is short for \(\Pr\{A_{L_0+1} = a_{L_0+1} | A[1:L_0]= a[1:L_0]\}\), and \(\sum_{a[1:L]}\) is short for \(\sum_{a_1\in\mathcal{A}_1,a_2\in\mathcal{A}_2,...,a_L\in\mathcal{A}_L}\). Inequality (a) follows from \eqref{eq:lem:rv} and Inequality (b) follows from the induction assumption. The proof is completed. \(\hfill\blacksquare\)

\section{Proof of Lemma \ref{lem:outer1}}\label{app:lem:outer1}

This appendix establishes a pair of upper bounds on the secrecy capacities of stochastic code over the AVWC and AVWC-CSR. Since the proofs of AVWC and AVWC-CSR are similar, we only give the proof of AVWC. To be precise, for any \(\epsilon > 0\), suppose that there exists a stochastic code \((F,\phi)\) of length \(N\), satisfying
\begin{equation}\label{eq:app:lem:outer1:0}
\lambda(\mathcal{W},F,\phi) < \epsilon \text{ and } \max_{s^N\in\mathcal{S}^N} I(M; Z^N(s^N)) < \epsilon,
\end{equation}
where \(M\) is the source message uniformly distributed over the message set \(\mathcal{M}\), and \(Z^N(s^N)\), defined in \eqref{eq:def:avwc:z}, is the output of wiretap AVC under the state sequence \(s^N\). We need to show that its transmission rate follows that
\[
R = \frac{1}{N}\log|\mathcal{M}| \leq \min_{q\in\mathcal{P}(\mathcal{S}),s\in\mathcal{S}}\max_{U_{q,s}}[I(U_{q,s};Y_q)-I(U_{q,s};Z_s)] + \epsilon',
\]
where \(\epsilon'\rightarrow 0\) as \(\epsilon \rightarrow 0\), and \((U_{q,s},Y_q,Z_s)\) satisfies \eqref{eq:thm:avwc:rc:1} and \eqref{eq:thm:avwc:rc:2} with \(U=U_{q,s}\).

The proof is organized as the following two steps.
\begin{itemize}
  \item The first step proves that for any stochastic code satisfying \eqref{eq:app:lem:outer1:0}, its transmission rate should satisfy \eqref{eq:app:lem:outer1:1}. The key idea is to show that if a stochastic code achieves a vanishing decoding error over the AVC \(\mathcal{W}\), it also achieves a vanishing decoding error over the AVC \(\bar{\mathcal{W}}\), where \(\bar{\mathcal{W}}\) is the convex hull of \(\mathcal{W}\), defined in \eqref{eq:pre:13}.
  \item The final step establishes the upper bound using the standard technique. To be concrete, we first obtain \eqref{eq:app:lem:outer1:3} by definition, which is further derived to \eqref{eq:app:lem:outer1:5} by the technique of single letterization introduced in \cite{BCC}.
\end{itemize}

\textbf{Proof of Step 1.} Let $(F, \phi)$ be a pair of stochastic encoder and decoder over the AVWC $(\mathcal{W},\mathcal{V})$ satisfying \eqref{eq:app:lem:outer1:0}.
For any \(q\in\mathcal{P}(\mathcal{S})\), let the random sequence \(Y^N_q\) satisfy
\[
\Pr\{Y^N_q=y^N|X^N=x^N\}=W_q(y^N|x^N)=\prod_{i=1}^N W_q(y_i|x_i)
\]
for \(y^N \in \mathcal{Y}^N\) and \(x^N\in\mathcal{X}^N\), where \(W_q\) is defined in \eqref{eq:pre:4}. Then for any \(m\in\mathcal{M}\), the decoding error probability of the message \(m\) achieved by the code \((F,\phi)\) over the DMC \(W_q\) is given by
\[
\begin{array}{lll}
\lambda_{m}(\{W_q\},F,\phi) &=& \displaystyle\sum_{x^N\in\mathcal{X}^N} \Bigg[\Pr\{F(m)=x^N\} (1-W_q(\phi^{-1}(m)|x^N))\Bigg]\\
&=& \displaystyle\sum_{x^N\in\mathcal{X}^N} \sum_{s^N\in\mathcal{S}^N} \Bigg[ \Pr\{F(m)=x^N\} q(s^N) (1-W(\phi^{-1}(m)|x^N, s^N))\Bigg]\\
&=& \displaystyle\sum_{s^N\in\mathcal{S}^N} \Bigg\{ q(s^N)\Bigg[\sum_{x^N\in\mathcal{X}^N} \Pr\{F(m)=x^N\} (1-W(\phi^{-1}(m)|x^N, s^N))\Bigg]\Bigg\}\\
&=& \displaystyle\sum_{s^N\in\mathcal{S}^N} \Bigg[ q(s^N)\lambda_{m}(\mathcal{W},F,\phi, s^N)\Bigg]\\
&\leq& \displaystyle\sum_{s^N\in\mathcal{S}^N} \Bigg[ q(s^N)\lambda(\mathcal{W},F,\phi, s^N)\Bigg]\\
&\leq& \lambda(\mathcal{W},F,\phi) < \epsilon,
\end{array}
\]
where \(\lambda_{m}(\mathcal{W},F,\phi, s^N)\) is given in \eqref{eq:pre:11}, \(\lambda(\mathcal{W},F,\phi, s^N)\) is given in \eqref{eq:pre:10}, \(\lambda(\mathcal{W},F,\phi)\) is given in \eqref{eq:pre:8}, and
\[
q(s^N)=\prod_{i=1}^N q(s_i).
\]

On account of the Fano's inequality, the transmission rate satisfies that
\begin{equation}\label{eq:app:lem:outer1:1}
NR \leq I(M;Y^N_q) + N\delta(\epsilon)
\end{equation}
for every \(q\in\mathcal{P}(\mathcal{S})\), where \(\delta(\epsilon)\rightarrow 0\) as \(\epsilon\rightarrow 0\). \(\hfill\blacksquare\)

\textbf{Proof of Step 2.} For any \(s\in\mathcal{S}\), let \(Z^N_s\) be the output of the wiretap AVC when the state sequence is \(N\) copies of \(s\), i.e.
\[
\Pr\{Z^N_s=z^N|X^N=x^N\}=\prod_{i=1}^N V_s(z_i|x_i)
\]
for \(z^N \in \mathcal{Z}^N\) and \(x^N\in\mathcal{X}^N\).
By the definition of achievability in Definition \ref{def:avwc:sc}, it follows that
\begin{equation}\label{eq:app:lem:outer1:2}
I(M;Z^N_s) < \epsilon
\end{equation}
for any \(s\in\mathcal{S}\). Combining \eqref{eq:app:lem:outer1:1} and \eqref{eq:app:lem:outer1:2} gives
\begin{equation}\label{eq:app:lem:outer1:3}
NR \leq I(M;Y^N_q) - I(M;Z^N_s) + N\delta(\epsilon) + \epsilon
\end{equation}
for any \(q\in\mathcal{P}(\mathcal{S})\) and \(s\in\mathcal{S}\). By the similar way of establishing Equations (38) and (39) in \cite{BCC} along with Lemma 7 in \cite{BCC}, the value of \(I(M;Y^N_q) - I(M;Z^N_s)\) follows that
\[
I(M;Y^N_q) - I(M;Z^N_s) = \sum_{i=1}^N\Bigg[I(M;Y_{q,i}|Y^{i-1}_q,Z^N_{s,i+1})-I(M;Z_{s,i}|Y^{i-1}_q,Z^N_{s,i+1})\Bigg],
\]
where \(Y^{i-1}_q=(Y_{q,1},Y_{q,2},...,Y_{q,i-1})\) and \(Z^N_{s,i+1}=(Z_{s,i+1},Z_{s,i+2},...,Z_{s,N})\).
Let \(V_{q,s,i}=(Y^{i-1}_q,Z^N_{s,i+1})\). The formula above can be rewritten as
\[
\begin{array}{lll}
I(M;Y^N_q) - I(M;Z^N_s) &=& \displaystyle\sum_{i=1}^N\Bigg[I(M;Y_{q,i}|V_{q,s,i})-I(M;Z_{s,i}|V_{q,s,i})\Bigg] \\
&=& \displaystyle\sum_{i=1}^N\Bigg[I(M,V_{q,s,i};Y_{q,i}|V_{q,s,i})-I(M,V_{q,s,i};Z_{s,i}|V_{q,s,i})\Bigg].
\end{array}
\]
Let \(J\) be a random variable uniformly distributed over \([1:N]\) and independent of \(M, Y^N_q, Z^N_s\) and \(V^N_{q,s}=(V_{q,s,1},V_{q,s,2},...,V_{q,s,N})\). We obtain that
\[
\begin{array}{lll}
I(M;Y^N_q) - I(M;Z^N_s) &=& N[I(M,V_{q,s,J};Y_{q,J}|V_{q,s,J},J)-I(M,V_{q,s,J};Z_{s,J}|V_{q,s,J},J)]\\
&=& N[I(M,V_{q,s,J},J;Y_{q,J}|V_{q,s,J},J)-I(M,V_{q,s,J},J;Z_{s,J}|V_{q,s,J},J)]
\end{array}
\]

Set
\begin{equation}\label{eq:app:lem:outer1:6}
V_{q,s}=(V_{q,s},J), U_{q,s}=(M,V_{q,s}),X=X_J,Y_q=Y_{q,J} \text{ and } Z_s=Z_{s,J}.
\end{equation}
It follows that
\begin{equation}\label{eq:app:lem:outer1:4}
I(M;Y^N_q) - I(M;Z^N_s) = N\cdot[I(U_{q,s};Y_q|V_{q,s}) - I(U_{q,s};Z_s|V_{q,s})] \leq \max_{U_{q,s}} N\cdot[I(U_{q,s};Y_q) - I(U_{q,s};Z_s)],
\end{equation}
where \(V_{q,s}\rightarrow U_{q,s} \rightarrow X \rightarrow (Y_q,Z_s)\) forms a Markov chain.

Substituting \eqref{eq:app:lem:outer1:4} into \eqref{eq:app:lem:outer1:3} gives
\[
R \leq \max_{U_{q,s}} [I(U_{q,s};Y_q) - I(U_{q,s};Z_s)] + \epsilon + \delta(\epsilon).
\]
Since \(q\) and \(s\) can be arbitrary, we have
\begin{equation}\label{eq:app:lem:outer1:5}
R \leq \min_{q\in\mathcal{P}(\mathcal{S}),s\in\mathcal{S}}\max_{U_{q,s}} [I(U_{q,s};Y_q) - I(U_{q,s};Z_s)] + \epsilon + \delta(\epsilon).
\end{equation}
The proof is completed by setting \(\epsilon'=\epsilon+\delta(\epsilon)\). \(\hfill\blacksquare\)

\section{Proof of Proposition \ref{prop:less:avwc}}\label{app:prop:less:avwc}

This appendix establishes the secrecy capacity of stochastic code over the AVWC, when the main AVC is severely less noisy than the wiretap AVC. To be precise, for any \(\epsilon > 0\), suppose that there exists a stochastic code \((F,\phi)\) of length \(N\), satisfying
\begin{equation}\label{eq:app:prop:less:avwc:0}
\lambda(\mathcal{W},F,\phi) < \epsilon \text{ and } \max_{s^N\in\mathcal{S}^N} I(M; Z^N(s^N)) < \epsilon,
\end{equation}

where \(M\) is the source message uniformly distributed over the message set \(\mathcal{M}\), and \(Z^N(s^N)\), defined in \eqref{eq:def:avwc:z}, is the output of the wiretap AVC under the state sequence \(s^N\). It suffices to show that its transmission rate follows that
\[
R = \frac{1}{N}\log|\mathcal{M}| \leq \max_{X}\min_{q\in\mathcal{P}(\mathcal{S}),s\in\mathcal{S}}[X;Y_q)-I(X;Z_s)] + \epsilon',
\]
where \(\epsilon'\rightarrow 0\) as \(\epsilon \rightarrow 0\) and \((X,Y_q,Z_s)\) satisfies \eqref{eq:sec:less:1}.

The proof is organized as follows. Equation \eqref{eq:app:prop:less:avwc:1} is first obtained by the technique introduced in Appendix \ref{app:lem:outer1}. Then, Equation \eqref{eq:app:prop:less:avwc:2} is obtained by the definition of severely less noisy AVCs. The proposition is finally established by \eqref{eq:app:prop:less:avwc:3}.

Let \((F,\phi)\) be a pair of stochastic encoder and decoder satisfying \eqref{eq:app:prop:less:avwc:0}. By the same way of establishing \eqref{eq:app:lem:outer1:3} and \eqref{eq:app:lem:outer1:4}, for any \(q\in\mathcal{P}(\mathcal{S})\) and \(s\in\mathcal{S}\), we have
\begin{equation}\label{eq:app:prop:less:avwc:1}
R \leq I(U_{q,s};Y_q|V_{q,s})-I(U_{q,s};Z_s|V_{q,s}) + \epsilon + \delta(\epsilon),
\end{equation}
where \((V_{q,s},U_{q,s}, X, Y_q, Z_s)\) is a collection of random variables introduced in \eqref{eq:app:lem:outer1:6}, and \(V_{q,s}\rightarrow U_{q,s}\rightarrow X \rightarrow (Y_q,Z_s)\) forms a Markov chain.

The value of \(I(U_{q,s};Y_q|V_{q,s})-I(U_{q,s};Z_s|V_{q,s})\) can be bounded by
\begin{equation}\label{eq:app:prop:less:avwc:2}
\begin{array}{lll}
&&I(U_{q,s};Y_q|V_{q,s})-I(U_{q,s};Z_s|V_{q,s}) \\
&=& I(U_{q,s},X;Y_q|V_{q,s})-I(U_{q,s}, X;Z_s|V_{q,s}) - I(X;Y_q|U_{q,s},V_{q,s}) + I(X;Z_s|U_{q,s},V_{q,s})\\
&\leq& I(U_{q,s},X;Y_q|V_{q,s})-I(U_{q,s}, X;Z_s|V_{q,s})\\
&=& I(U_{q,s},V_{q,s},X;Y_q)-I(U_{q,s}, V_{q,s}, X;Z_s) - I(V_{q,s};Y_q) + I(V_{q,s};Z_s)\\
&\leq& I(U_{q,s},V_{q,s},X;Y_q)-I(U_{q,s}, V_{q,s}, X;Z_s)\\
&=& I(X;Y_q)-I(X;Z_s),
\end{array}
\end{equation}
where the inequalities follow from the definition of severely less noisy AVCs, and the last equation follows from the Markov chain \(V_{q,s}\rightarrow U_{q,s}\rightarrow X \rightarrow (Y_q,Z_s)\).

Substituting \eqref{eq:app:prop:less:avwc:2} into \eqref{eq:app:prop:less:avwc:1} gives
\[
R \leq I(X;Y_q)-I(X;Z_s) + \epsilon + \delta(\epsilon).
\]
Since \(q\) and \(s\) are arbitrary, we have
\begin{equation}\label{eq:app:prop:less:avwc:3}
R \leq \min_{q\in\mathcal{P}(\mathcal{S}),s\in\mathcal{S}}[I(X;Y_q)-I(X;Z_s)] + \epsilon + \delta(\epsilon) \leq \max_{X}\min_{q\in\mathcal{P}(\mathcal{S}),s\in\mathcal{S}}[I(X;Y_q)-I(X;Z_s)] + \epsilon + \delta(\epsilon).
\end{equation}
Notice that the random variable \(X\) is unrelated to \(q\) and \(s\). The proof is completed by setting \(\epsilon'=\epsilon+\delta(\epsilon)\). \(\hfill\blacksquare\)

\end{appendices}

\end{document}